%% file: flowdecomp.tex
\documentclass[letterpaper,conference,dvipsnames]{IEEEtran}

\usepackage{tikz}
\usetikzlibrary{decorations.pathreplacing}
\usetikzlibrary{shapes.misc, positioning}
\usepackage{microtype}
\usepackage{relsize}
\usepackage{setspace}
\usepackage[linesnumbered,ruled,vlined]{algorithm2e}
\usepackage{algpseudocode}
\usepackage{bbm}
\usepackage{epsfig}  % for figures
\usepackage{graphicx}  % another package that works for figures
\usepackage{longtable}
\usepackage{amsmath}  % for math spacing
\usepackage{array}
\usepackage{multirow}
\usepackage{rotating}
\usepackage{amssymb}  % for math spacing
\usepackage{lscape}  % Useful for wide tables or figures.
\usepackage{amsthm}
\usepackage[justification=raggedright]{caption}	% makes captions ragged right - thanks to Bryce Lobdell
\usepackage[margin=0.75in]{geometry}
\allowdisplaybreaks
\newtheorem{theorem}{Theorem}%[section]

\newtheorem{lemma}{Lemma}%[section]
\ifCLASSINFOpdf
\else
\fi
\begin{document}
\title{Flow Decomposition}
%\newgeometry{right=0.75in,left=0.75in,bottom=0.75in,top=1in}
\author{\IEEEauthorblockN{Jonathan Ponniah~\IEEEmembership{Member,~IEEE,}
}
\IEEEauthorblockA{Department of Electrical Engineering\\
San Jose State University%\\
%College Station, Texas\\
%Email: jonathan.ponniah@gmail.com
}
\and
\IEEEauthorblockN{Liang-Liang Xie~\IEEEmembership{Fellow,~IEEE,}
}
\IEEEauthorblockA{Department of Electrical and Computer Engineering\\
University of Waterloo%\\
%Waterloo, Ontario\\
%Email: llxie@uwaterloo.ca
}
%\and
%\IEEEauthorblockN{P.R Kumar}
%\IEEEauthorblockA{Department of Electrical and\\Computer Engineering\\
%Texas A\&M%\\
%College Station, Texas\\
%Email: prk@tamu.edu
%}
%\and
%\IEEEauthorblockN{Lisa Ponniah}
%\IEEEauthorblockA{Department of Tax and Auditing\\
%MNP LLP}
%\and
%\IEEEauthorblockN{Christina Ponniah}
%\IEEEauthorblockA{Department of Geotechnical Engineering\\
%University of Laval}
\thanks{This material is based upon work partially supported by NSF Contract CNS-1302182,  AFOSR Contract FA9550-13-1-0008, and NSF Science \& Technology Center Grant CCF-0939370.}}

%\markboth{Journal of \LaTeX\ Class Files,~Vol.~13, No.~9, September~2014}%
%{Shell \MakeLowercase{\textit{et al.}}: Bare Demo of IEEEtran.cls for Journals}



\maketitle

\begin{abstract}
The decode-forward achievable region is studied for general networks.  The region is subject to a fundamental tension in which nodes individually benefit at the expense of others.  The complexity of the region depends on all the ways of resolving this tension.  Two sets of constraints define an outer-bound on the decode-forward region: first, the conventional mutual-information inequalities implied by the one-relay channel, and second, causality constraints that ensure nodes only forward messages they have already decoded.  The framework of \textit{flow decomposition} is introduced to show these constraints are also sufficient.  Flow decomposition provides a way of manipulating regular decode-forward schemes without the long encoding delays and restrictions on bidirectional communication of backward decoding.   The two structures that define a flow decomposition are \textit{flows} and \textit{layerings}.  Flows specify the nodes which encode messages from each source (i.e., the routes) and the encoding delays.  Layerings specify the messages decoded at a specific node in the channel.  We focus on two types of flow: hierarchical flow, with tree-like routes, and all-cast flow, where each route covers all nodes.  For arbitrary flows of either type and any rate-vector satisfying the mutual-information constraints at a specific node, we prove there are equivalent flows and a layering that satisfy both the mutual-information and causality constraints.  In separate work, we show that only the mutual-information constraints are active in channels with hierarchical flow, which implies the achievable region has minimal complexity.  In channels with all-cast flow, the achievable region is computable.
\end{abstract}

%\begin{IEEEkeywords}
%IEEEtran, journal, \LaTeX, paper, template.
%\end{IEEEkeywords}

% \fi
\IEEEpeerreviewmaketitle

\input{"introduction.tex"}

\begin{figure*}[t]
        \center{\includegraphics[width=\textwidth]
        {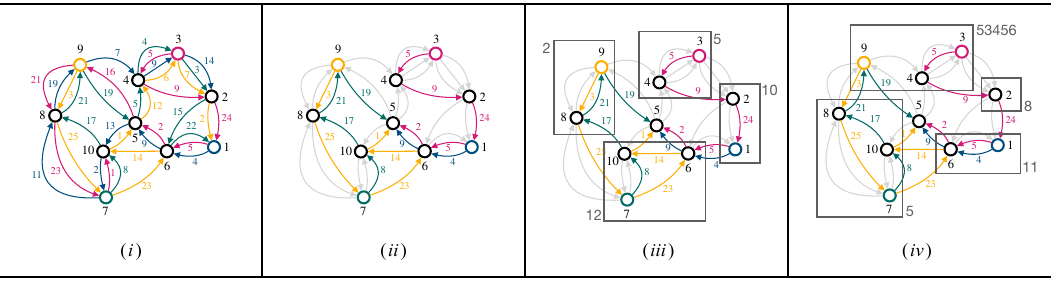}}
        \caption{(i) The flows ${\bf F}=\{{\bf f}(1),{\bf f}(3),{\bf f}(7),{\bf f}(9)\}$, where ${\bf f}(1)=1\xrightarrow{4}6\xrightarrow{9}5\xrightarrow{13}10\xrightarrow{2}7\xrightarrow{11}8\xrightarrow{19}9\xrightarrow{4}7\xrightarrow{9}3\xrightarrow{14}2$, ${\bf f}(3)=3\xrightarrow{5}4\xrightarrow{9}2\xrightarrow{24}1\xrightarrow{5}6\xrightarrow{2}5\xrightarrow{16}9\xrightarrow{21}8\xrightarrow{23}7\xrightarrow{1}10$, ${\bf f}(7)=7\xrightarrow{8}10\xrightarrow{17}8\xrightarrow{21}9\xrightarrow{19}5\xrightarrow{5}4\xrightarrow{4}3\xrightarrow{3}2\xrightarrow{15}6\xrightarrow{22}1$, and ${\bf f}(9)=9\xrightarrow{3}8\xrightarrow{25}7\xrightarrow{23}6\xrightarrow{14}10\xrightarrow{1}5\xrightarrow{12}4\xrightarrow{6}3\xrightarrow{7}2\xrightarrow{2}1$. (ii) ${\bf f}(1,5)=1\xrightarrow{4}6$, ${\bf f}(3,5)=3\xrightarrow{5}4\xrightarrow{9}2\xrightarrow{24}1\xrightarrow{5}6$, ${\bf f}(7,5)=7\xrightarrow{8}10\xrightarrow{17}8\xrightarrow{21}9$, and ${\bf f}(9,5)=9\xrightarrow{3}8\xrightarrow{25}7\xrightarrow{23}6\xrightarrow{14}10$. The flows ${\bf F}$ and layering ${\bf L}_{5}$ with $L_{2}=\{8,9\}$, $L_{5}=\{3,4\}$, $L_{10}=\{1,2\}$, and $L_{12}=\{6,7,10\}$.  The flows ${\bf F}$ and layering ${\bf L}_{5}$ with $L_{5}=\{7,8,10\}$, $L_{8}=\{2\}$, $L_{11}=\{1,6\}$, and $L_{53456}=\{3,4,9\}$.}
        \label{layering}
 \end{figure*}

\input{"fd.tex"}

\begin{figure*}[!ht]
        \center{\includegraphics[width=\textwidth]
        {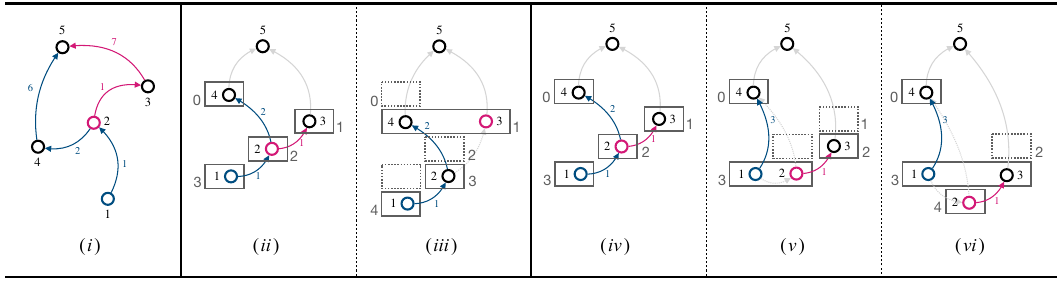}}
        \caption{(i) The flows ${\bf F}=\{{\bf f}(1),{\bf f}(2)\}$ with ${\bf f}(1)=1\xrightarrow{1}2\xrightarrow{2}4\xrightarrow{6}5$ and ${\bf f}(2)=2\xrightarrow{1}3\xrightarrow{7}5$. (ii) The flows ${\bf F}$ and layering ${\bf L}_{5}$ where ${\bf L}_{5}=(L_{0},L_{1},L_{2},L_{3})$,  $L_{0}=\{4\}$, $L_{1}=\{3\}$, $L_{2}=\{2\}$, and $L_{3}=\{1\}$.  The virtual flows are ${\bf v}(1)=1\xrightarrow{1}2\xrightarrow{2}4$ and ${\bf v}(2)=2\xrightarrow{1}3$. (iii) The flows ${\bf F}$ and layering ${\bf L}^{\prime}_{5}$ where ${\bf L}^{\prime}_{5}=\text{\sc shift}({\bf L}_{5},\{1\})$, ${\bf L}^{\prime}_{5}=(L^{\prime}_{0},L^{\prime}_{1},L_{2},L^{\prime}_{3},L^{\prime}_{4})$, $L^{\prime}_{1}=\{3,4\}$, $L^{\prime}_{3}=\{2\}$ and $L^{\prime}_{4}=\{1\}$.  The virtual flows are ${\bf v}^{\prime}(1)=1\xrightarrow{1}2\xrightarrow{2}4$ and ${\bf v}^{\prime}(2)=3$. (iv) The flows ${\bf F}$ and layering ${\bf L}_{5}$ where ${\bf L}_{5}=(L_{0},L_{1},L_{2},L_{3})$,  $L_{0}=\{4\}$, $L_{1}=\{3\}$, $L_{2}=\{2\}$, and $L_{3}=\{1\}$. (v) The flows ${\bf F}$ and layering ${\bf L}^{\prime}_{5}$ where ${\bf L}^{\prime}_{5}=\text{\sc shift}({\bf L}_{5},\{2\})$ and $L^{\prime}_{0}=\{4\}$, $L^{\prime}_{2}=\{3\}$, and $L^{\prime}_{3}=\{1,2\}$.  The virtual flows are ${\bf v}^{\prime}(1)=1\xrightarrow{3}4$ and ${\bf v}^{\prime}(2)=2\xrightarrow{1}3$. (vi) The flows ${\bf F}$ and layering ${\bf L}^{\prime\prime}_{5}$ where ${\bf L}^{\prime\prime}_{5}=\text{\sc shift}({\bf L}^{\prime}_{5},\{2\})$ and $L^{\prime\prime}_{0}=\{4\}$, $L^{\prime\prime}_{3}=\{1,3\}$, and $L^{\prime\prime}_{4}=\{2\}$.  The virtual flows are ${\bf v}^{\prime\prime}(1)=1\xrightarrow{3}4$ and ${\bf v}^{\prime\prime}(2)=2\xrightarrow{1}3$.}
        \label{shiftingexample}
 \end{figure*}

\section{Proof of Theorem \ref{theoremone}}
\label{prooftheoremone}
%The achievable regions of different regular decoding schemes geometrically interlock into a composite region with boundary facets defined by (\ref{biggiewiggie}), as depicted in Figures \ref{polytope}(i) and (ii) for the the two-source multiple-access relay channel (MARC) and three-source MARC respectively.  This interlocking behavior is at the core of Theorem \ref{theoremone}.  Figure \ref{polytope}(ii) is a two-dimensional representation of this behavior with three-dimensional achievable regions.
To prove Theorem \ref{theoremone}, we first fix ${\bf R}\in{\cal R}_{d}({\bf F})$ and pick an arbitrary layering ${\bf L}_{d}$.  Then we find the largest subset $S\subseteq{\cal S}(d)$ of (\ref{flowdecomp}) violated by ${\bf R}$, and use this subset to ``shift" ${\bf L}_{d}$.  We repeat this process until ${\bf R}$ satisfies (\ref{flowdecomp}) for all $S\subseteq{\cal S}(d)$.  Finally, we modify ${\bf F}$ and ${\bf L}_{d}$ to give an equivalent ${\bf F}^{\prime}$ and ${\bf L}^{\prime}_{d}$ that both satisfy (C1) and (C2).

Fix some ${\bf R}\in{\cal R}_{d}({\bf F})$ and let $U\subseteq{\cal S}(d)$ be the largest subset that violates (\ref{flowdecomp}).  Define ${\bf L}^{\prime}_{d}:=\text{\sc shift}({\bf L}_{d}, U)$ as follows.  Let $\text{\sc layer}^{\prime}(\cdot)$ correspond to ${\bf L}^{\prime}_{d}$ with respect to (L5), and for every $i\in F_{d}({\cal N})$:  
\begin{align}
\label{shift2}
\text{\sc layer}^{\prime}(i)=\begin{cases}l&i\in A_{l}({\cal N})\setminus A_{l}(U),\\ l+1 & i\in A_{l}(U).\end{cases}	
\end{align}

For ${\bf L}^{\prime}_{d}=\text{\sc shift}(,{\bf L}_{d},U)$, let ${\cal R}({\bf F},{\bf L}^{\prime}_{d})$ denote the set of rate vectors that satisfy:
\begin{align}
\label{flowdecomp2}
	R_{S}<\displaystyle\sum^{|{\bf L}^{\prime}_{d}|-1}_{l=0}I(X_{A^{\prime}_{l}(S)};Y_{d}|X_{\tilde{A}^{\prime}_{l}(S)}),
\end{align}
for all $S\subseteq{\cal S}(d)$, where
\begin{align}
\label{Aprime}
A^{\prime}_{l}(S)&:=\{i\in{\bf v}^{\prime}(s):s\in S\cap{\cal S}(d)\}\cap L^{\prime}_{l},\\
\label{Atildeprime}
	\tilde{A}^{\prime}_{l}(S)&:=(\cup_{i\in I(d)}F_{i}({\cal N}))\setminus(\cup^{|{\bf L}^{\prime}_{d}|-1}_{q=l+1}L^{\prime}_{q}\cup A^{\prime}_{l}(S)).
\end{align}
By definition, ${\bf L}^{\prime}_{d}:=(L^{\prime}_{0},L^{\prime}_{1},\ldots,L^{\prime}_{|{\bf L}^{\prime}_{d}|-1})$, and $\{{\bf v}^{\prime}(s):s\in{\cal S}(d)\}$ is the set of virtual flows generated by $({\bf F},{\bf L}^{\prime}_{d})$.  %Lemmas \ref{disjoint}-\ref{grandfinale} mirror \cite{compdecomp}:(Lemmas 2-7) for compress-forward schemes.  
For all $j\in{\bf f}(v(s),d)$:  
\begin{align}
\label{firstIndexCodingShiftedComplete2}
\text{\sc layer}(v(s))-\text{\sc layer}(j)&\leq k_{v(s),j},
\end{align}
where (\ref{firstIndexCodingShiftedComplete2}) follows from (\ref{firstIndexCoding}).  For all $j\in{\bf f}(s,v(s))$:  
\begin{align}
\label{secondIndexShiftedComplete2}
\text{\sc layer}(j)-\text{\sc layer}(v(s))&> k_{j,v(s)},
\end{align} 
where (\ref{secondIndexShiftedComplete2}) follows from (\ref{secondIndexCoding}).  Similarly, for all $j\in{\bf f}(v^{\prime}(s),d)$:
\begin{align}
\label{firstIndexCodingShiftedComplete}
\text{\sc layer}^{\prime}(v^{\prime}(s))-\text{\sc layer}^{\prime}(j)&\leq k_{v^{\prime}(s),j},
\end{align}
where (\ref{firstIndexCodingShiftedComplete}) follows from (\ref{firstIndexCoding}).  For all $j\in{\bf f}(s,v^{\prime}(s))$:
\begin{align}
\label{secondIndexShiftedComplete}
\text{\sc layer}^{\prime}(j)-\text{\sc layer}^{\prime}(v^{\prime}(s))&> k_{j,v^{\prime}(s)},
\end{align}
where (\ref{secondIndexShiftedComplete}) follows from (\ref{secondIndexCoding}).
%For Lemmas \ref{disjoint}-\ref{join}, fix some $S\subseteq{\cal S}(d)$ and some $l=0,\ldots,|{\bf L}_{d}|-1$.

\input{"disjoint.tex"}
\input{"remain.tex"}
\input{"join.tex"}
%First create an intermediate ordered partition ${\bar L}_{\beta}$.  For each $l\in\{0,\dots,|{\bar L}_{\alpha}|\}$, set ${L_{\beta,\textcolor{Black}{l}}}:={L_{\alpha,\textcolor{Black}{l}}}\setminus{A_{\alpha,\textcolor{Black}{l}}(}S{)}$ and set ${L}_{\gamma,\textcolor{Black}{l}}:={L}_{\beta,l}$.  Then set ${L_{\gamma,\textcolor{Black}{l+}{\delta_{\alpha}(}\textcolor{Black}{S}{)}}}:={L_{\beta,\textcolor{Black}{l+}{\delta_{\alpha}(}\textcolor{Black}{S}{)}}}\cup{A_{\alpha,\textcolor{Black}{l}}(}S{)}$.  By construction, any admissible shift of a stable ordered partition yields another stable ordered partition. 

%If $U$ is empty the proof is done since this implies $V={\cal S}$ and ${\bar R}\in{\cal R}(F,\bar{L}_{d})$.  %Since $(F,\bar{L}_{d})$ is complete, Lemma $U\neq{\cal S}$.  
%Note that ${S}_{\alpha}\neq{\cal S}$ since ${\bar L}_{\alpha}$ is stable on ${\cal S}$.   
%Lemmas \ref{induction} and \ref{termination} use only the mutual information chain rule and basic set-theoretic arguments (intersection, union, and exclusion).

For the ${\bf R}\in{\cal R}_{d}({\bf F})$ chosen at the onset of Section \ref{prooftheoremone}, recall that $U\subseteq{\cal S}(d)$ is the largest set that violates (\ref{flowdecomp}).  Given this ${\bf R}$, let $Z\subseteq{\cal S}(d)$ denote the set that satisfies (\ref{flowdecomp}) for all $S\subseteq Z$ with respect to ${\bf L}_{d}$ and let $Z^{\prime}\subseteq{\cal S}(d)$ denote the set that satisfies (\ref{flowdecomp2}) for all $S\subseteq Z^{\prime}$ with respect to ${\bf L}^{\prime}_{d}$, where ${\bf L}_{d}$ and ${\bf L}^{\prime}_{d}$ are related through (\ref{shift2}).  

%The set ${\cal S}\setminus Z$ is a combinatorial metric that indicates whether $({\bf F},{\bf L})$ is ``close'' to the target rate vector ${\bf R}\in{\cal R}({\bf F})$; $Z={\cal S}$ implies ${\bf R}\in{\cal R}({\bf F},{\bf L})$.  Shifting flow decompositions in reference to ${\bf R}$ creates a sequence of flow decompositions that approaches ${\bf R}$.  Lemma \ref{induction} shows that $Z\subseteq Z^{\prime}$.

\input{"induction.tex"}

\begin{figure*}[!ht]
        \center{\includegraphics[width=\textwidth]
        {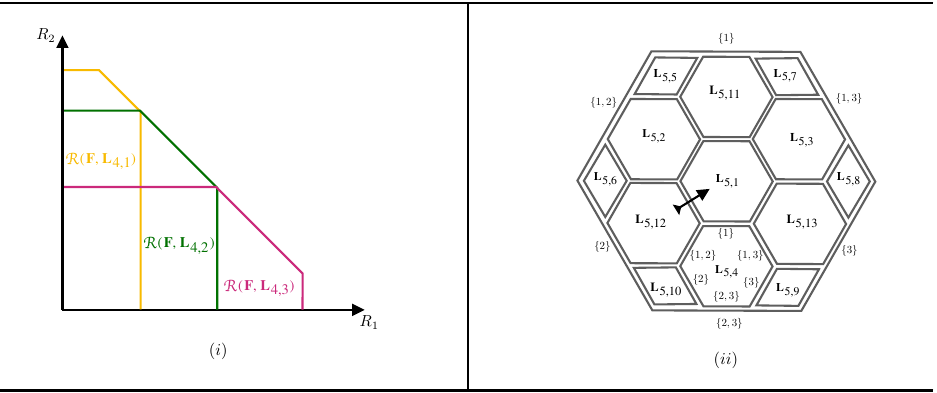}}
        \caption{(i) ${\cal R}_{4}({\bf F})$ for the two-user multiple-access relay channel in which ${\cal N}=\{1,2,3,4\}$, ${\bf F}=\{{\bf f}(1),{\bf f}(2)\}$,   ${\bf f}(1):=1\xrightarrow{1}3\xrightarrow{1}4$, and ${\bf f}(2):=2\xrightarrow{1}3\xrightarrow{1}4$.  The layerings are ${\bf L}_{4,1}=(\{1,3\},\{2\})$, ${\bf L}_{4,2}=(\{3\},\{1,2\})$ and ${\bf L}_{4,3}=(\{2,3\},\{1\})$.  The boundaries of each subregion (identified by color) map to (\ref{flowdecomp}) for $S=\{1\}, \{2\}, \{1,2\}$.  The boundaries of the full region map to (\ref{biggiewiggie}) for $S=\{1\}, \{2\}, \{1,2\}$.  A layering shifted by $\{1\}$ or $\{2\}$ generates an adjacent subregion: ${\bf L}_{4,2}=\text{\sc shift}({\bf L}_{4,1},\{1\})$ where ${\cal R}({\bf F},{\bf L}_{4,1})$ interlocks with ${\cal R}({\bf F},{\bf L}_{4,2})$.  (ii) A 2-D projection of a 3-D region ${\cal R}_{5}({\bf F})$ for the three-user multiple-access relay in which ${\cal N}=\{1,2,3,4,5\}$, ${\bf F}=\{{\bf f}(1),{\bf f}(2),{\bf f}(3)\}$, ${\bf f}(1):=1\xrightarrow{1}4\xrightarrow{1}5$, ${\bf f}(2):=2\xrightarrow{1}4\xrightarrow{1}5$, and ${\bf f}(3):=3\xrightarrow{1}4\xrightarrow{1}5$.  The layerings are ${\bf L}_{5,1}=(\{4\},\{1,2,3\})$, ${\bf L}_{5,2}=(\{3,4\},\{1,2\})$, ${\bf L}_{5,3}=(\{2,4\},\{1,3\})$, ${\bf L}_{5,4}=(\{1,4\},\{2,3\})$, ${\bf L}_{5,5}=(\{3\},\{2,4\},\{1\})$, ${\bf L}_{5,6}=(\{3\},\{2,4\},\{2\})$, ${\bf L}_{5,7}=(\{2\},\{3,4\},\{1\})$, ${\bf L}_{5,8}=(\{2\},\{1,4\},\{3\})$, ${\bf L}_{5,9}=(\{1\},\{2,4\},\{3\})$, ${\bf L}_{5,10}=(\{1\},\{3,4\},\{2\})$, ${\bf L}_{5,11}=(\{2,3,4\},\{1\})$, ${\bf L}_{5,12}=(\{1,3,4\},\{2\})$ and ${\bf L}_{5,13}=(\{1,2,4\},\{3\})$.  The six facets of each subregion map to (\ref{flowdecomp}) for the sets $S=\{1\}, \{1,2\}, \{1,3\}, \{2\}, \{3\},\{2,3\}$.  The six facets of the full region map to (\ref{biggiewiggie}) for $S=\{1\}, \{1,2\}, \{1,3\}, \{2\}, \{3\},\{2,3\}$.  A  shifted layering generates an adjacent subregion: ${\bf L}_{5,1}=\text{\sc shift}({\bf L}_{5,12},\{1,3\})$, where $\{1,3\}$ is the facet of ${\cal R}({\bf F},{\bf L}_{5,12})$ that interlocks  with ${\cal R}({\bf F},{\bf L}_{5,1})$.}
        \label{polytope}
 \end{figure*}
 
For clarity, we update some notation.  Given a sequence of layerings $\{{\bf L}_{d,n}:n\in\mathbb{N}\}$, let ${\bf L}_{d,n}:=(L_{0,n},L_{1,n},\ldots,L_{|{\bf L}_{d,n}|-1})$ be a layering that satisfies (L1)-(L5) and let $\{{\bf v}_{n}(s):s\in{\cal S}(d)\}$ and $\{v_{n}(s):s\in{\cal S}(d)\}$ denote the set of virtual flows and virtual sources respectively, generated by $({\bf F},{\bf L}_{d,n})$ that satisfy (V1)-(V5).  For any $S\subseteq{\cal S}(d)$, let $V_{n}(S):=\{i\in{\bf v}_{n}(s):s\in S\cap{\cal S}(d)\}$.  For every $S\subseteq{\cal S}(d)$, let:
\begin{align}
	\label{aln}
	A_{l,n}(S)&:=\{i\in{\bf v}_{n}(s):s\in S\cap{\cal S}(d)\}\cap L_{l,n},\\
	\label{alntilde}
	\tilde{A}_{l,n}(S)&:=F_{I(d)}\cup(\cup^{l}_{q=0}L_{q,n})\setminus A_{l,n}(S),
\end{align}
where $F_{I(d)}$ is defined in (\ref{induction3}).  For every $S\subseteq{\cal S}(d)$,  let ${\cal R}({\bf F},{\bf L}_{d,n})$ denote the set of rate vectors ${\bf R}$ that satisfy:  
\begin{align}
	\label{flowdecomp3}
	R_{S}&<\displaystyle\sum^{|{\bf L}_{d,n}|-1}_{l=0}I(X_{A_{l,n}(S)};Y_{d}|X_{\tilde{A}_{l,n}(S)}).
\end{align}   
For some fixed ${\bf R}\in{\cal R}_{d}({\bf F})$, let $U_{n}\subset{\cal S}(d)$ denote the largest set that violates (\ref{flowdecomp3}) with respect to $({\bf F},{\bf L}_{d,n})$ and ${\bf R}$, and $Z_{n}\subseteq{\cal S}(d)$ denote the set in which all subsets $S\subseteq Z_{n}$ satisfy (\ref{flowdecomp3}) with respect to $({\bf F},{\bf L}_{d,n})$ and ${\bf R}$.

Let ${\bf L}_{d,n+1}:=\text{\sc shift}({\bf L}_{d,n},U_{n})$, and define the $\text{\sc shift}(\cdot\hspace{1mm},\cdot)$ operator as follows.  Let $\text{\sc layer}_{n}(\cdot)$ correspond to the layering ${\bf L}_{d,n}$, and for every $i\in F_{d}({\cal N})$:
\begin{align}
\label{shift4}
\text{\sc layer}_{n+1}(i)=\begin{cases}l&i\in A_{l,n}({\cal N})\setminus A_{l,n}(U_{n}),\\ l+1 & i\in A_{l,n}(U_{n}).\end{cases}	
\end{align}
For every $n\in\mathbb{N}$, the pair $(U_{n},Z_{n})$ satisfies:
\begin{align}
	\label{inductionupdate}
	Z_{n+1}=({\cal S}(d)\setminus U_{n})\cup Z_{n},
\end{align} 
where (\ref{inductionupdate}) follows from Lemma \ref{induction}.  For every $j\in{\bf f}(v_{n}(s),d)$, the virtual flow $v_{n}(s)$ satisfies:
	\begin{align}
		\label{pregrandfinale3}
		\text{\sc layer}_{n}(v_{n}(s))-\text{\sc layer}_{n}(j)&\leq k_{v_{n}(s),j},
	\end{align}
	where (\ref{pregrandfinale3}) follows from (\ref{firstIndexCoding}).  For every $j\in{\bf f}(s,v_{n}(s))$, the virtual flow $v_{n}(s)$ satisfies:
	\begin{align}
		\label{pregrandfinale4}
		\text{\sc layer}_{n}(j)-\text{\sc layer}_{n}(v_{n}(s))&>k_{j,v_{n}(s)},
	\end{align}
where (\ref{pregrandfinale4}) follows from (\ref{secondIndexCoding}).  For every $j\in{\bf f}(v_{n+1}(s),d)$, the virtual flow $v_{n+1}(s)$ satisfies:
	\begin{align}
		\label{grandfinale3}
		\text{\sc layer}_{n+1}(v_{n+1}(s))-\text{\sc layer}_{n+1}(j)&\leq k_{v_{n+1}(s),j},
	\end{align}
	where (\ref{grandfinale3}) follows from (\ref{firstIndexCoding}).  For all $j\in{\bf f}(s,v_{n+1}(s))$, the virtual flow $v_{n+1}(s)$ satisfies:
	\begin{align}
		\label{grandfinale4}
		\text{\sc layer}_{n+1}(j)-\text{\sc layer}_{n+1}(v_{n+1}(s))&>k_{j,v_{n+1}(s)},
	\end{align}
	where (\ref{grandfinale4}) follows from (\ref{secondIndexCoding}).  We now introduce the last two steps of the proof of Theorem \ref{theoremone}.

\hspace{-3mm}{\bf Step 1:} Given some node $d\in{\cal N}$, flows ${\bf F}$ that satisfies (C2), and some chosen ${\bf R}\in{\cal R}_{d}({\bf F})$, pick an arbitrary layering ${\bf L}_{d,0}$.  Construct a sequence of layerings $\{{\bf L}_{d,n}:n\in\mathbb{N}\}$ as follows:
	\begin{align}
		\label{shift3}
		{\bf L}_{d,n+1}&:=\text{\sc shift}({\bf L}_{d,n},U_{n}),
	\end{align}
	where the $\text{\sc shift}(\cdot\hspace{1mm},\cdot)$ operator is defined in (\ref{shift4}) with respect to the chosen ${\bf R}$.  We will show that ${\bf R}\in{\cal R}_{d}({\bf F},{\bf L}_{d,n^{*}})$ for some $n^{*}\in\mathbb{N}$.
	
\hspace{-3mm}{\bf Step 2:} Construct equivalent flows ${\bf F}^{\prime}$ and a layering ${\bf L}_{d}$ such that ${\cal R}({\bf F}^{\prime},{\bf L}_{d})\supseteq{\cal R}({\bf F}, {\bf L}_{d,n^{*}})$ and $({\bf F}^{\prime},{\bf L}_{d})$ satisfies (C1).  We will show that $({\bf F}^{\prime},{\bf L}_{d})$ exists.

\input{"grandefinale.tex"}

\begin{figure*}[!ht]
        \center{\includegraphics[width=\textwidth]
        {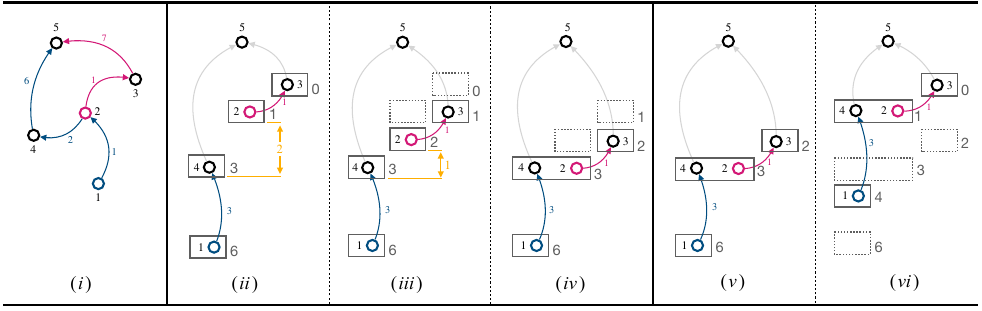}}
        \caption{(i) ${\bf F}=\{{\bf f}(1),{\bf f}(2)\}$, ${\bf f}(1):=1\xrightarrow{1}2\xrightarrow{2}4\xrightarrow{6}5$, ${\bf f}(2):=2\xrightarrow{1}3\xrightarrow{7}5$. (ii) The flows ${\bf F}$ and layering ${\bf L}_{5}$ where ${\bf L}_{ 5}=(\{3\},\{2\},\{\},\{4\},\{\},\{\},\{1\})$ and ${\cal S}(5)=\{1,2\}$.  The virtual flows are ${\bf v}(1)=1\xrightarrow{3}4$ and ${\bf v}(2)=2\xrightarrow{1}3$.  Since $V(\{2\})\subseteq\cup^{1}_{q=0}L_{q}$ and $V(\{1\})\subseteq\cup^{6}_{q=3}L_{q}$, $({\bf F},{\bf L}_{5})$ satisfies (P1) for $S_{c}=\{2\}$, $l=1$, and $k=2$. (iii) The flows ${\bf F}$ and layering ${\bf L}^{\prime}_{5}$ where ${\bf L}^{\prime}_{5}=\text{\sc shift}({\bf L}_{5},\{2\})$ and $L^{\prime}_{5}=(\{\},\{3\},\{2\},\{4\},\{\},\{\},\{1\})$.  The virtual flows are ${\bf v}^{\prime}(1)=1\xrightarrow{3}4$ and ${\bf v}^{\prime}(2)=2\xrightarrow{1}3$.  Since $V^{\prime}(\{2\})\subseteq\cup^{2}_{q=0}L_{q}$ and $V^{\prime}(\{1\})\subseteq\cup^{6}_{q=3}L_{q}$, $({\bf F},{\bf L}^{\prime}_{5})$ satisfies (P1) for $S_{c}=\{2\}$, $l=2$, and $k=1$.  (iv) The flows ${\bf F}$ and layering ${\bf L}^{\prime\prime}_{5}$ where ${\bf L}^{\prime\prime}_{5}=\text{\sc shift}({\bf L}^{\prime}_{5},\{2\})$ and ${\bf L}^{\prime\prime}_{5}=(\{\},\{\},\{3\},\{2,4\},\{\},\{\},\{1\})$.  The virtual flows are ${\bf v}^{\prime\prime}(1)=1\xrightarrow{3}4$ and ${\bf v}^{\prime\prime}(2)=2\xrightarrow{1}3$.  There are no subsets $S_{c}\subseteq{\cal S}(5)$ for which $({\bf F},{\bf L}^{\prime\prime}_{5})$ satisfies (P1). (v) The flows ${\bf F}$ and layering ${\bf L}^{\prime\prime}_{5}$ where ${\bf L}^{\prime\prime}_{5}=(\{\},\{\},\{3\},\{2,4\},\{\},\{\},\{1\})$. (vi) The flows ${\bf F}$ and layering ${\bf L}^{\prime\prime\prime}_{5}=\text{\sc norm}({\bf L}^{\prime\prime}_{5},2)$ where $L^{\prime\prime\prime}_{5}=(\{3\},\{2,4\},\{\},\{\},\{1\})$.  The virtual flows are ${\bf v}^{\prime\prime\prime}(1)=1\xrightarrow{3}4$ and ${\bf v}^{\prime\prime\prime}(2)=2\xrightarrow{1}3$.  ${\cal R}({\bf F},{\bf L}^{\prime\prime}_{5})={\cal R}({\bf F},{\bf L}^{\prime\prime\prime}_{5})$ since $L^{\prime\prime\prime}_{q}=L^{\prime\prime}_{q+2}$ for $q=0,\ldots,|{\bf L}^{\prime\prime\prime}_{5}|$.}
        \label{coarsepartitions}
 \end{figure*}

For Step 2 of the proof of Theorem \ref{theoremone}, we reuse the notation used in the proofs of Lemmas \ref{disjoint} - Lemma \ref{induction}.  Let ${\bf F}:=\{{\bf f}(s):s\in{\cal N}\}$ where ${\bf f}(s)$ satisfies (F1)-(F5) for all $s\in{\cal N}$.  For any $S\subseteq{\cal N}$, let $F_{d}({\cal N}):=\{j\in{\bf f}(s,d):s\in S\cap{\cal S}(d)\}$ for some fixed $d\in{\cal N}$.  Let ${\bf L}_{d}:=(L_{0},L_{1},\ldots,L_{|{\bf L}_{d}|-1})$ be a layering of $F_{d}({\cal N})$ that satisfies (L1)-(L5).  Let $\{{\bf v}(s):s\in{\cal S}(d)\}$ denote the virtual flows generated by $({\bf F},{\bf L}_{d})$ where ${\bf v}(s)$ satisfies (V1)-(V5) for all $s\in{\cal S}(d)$.  

Let ${\cal R}({\bf F},{\bf L}_{d})$ denote the set of rate vectors that satisfy:
\begin{align}
	\label{rfd}
	R_{S}&<\sum^{|{\bf L}_{d}|-1}_{l=0}I(X_{A_{l}(S)};Y_{d}|X_{\tilde{A}_{l}(S)}),
\end{align}
for all $S\subseteq{\cal S}(d)$ where, 
\begin{align}
	\label{newA}
	A_{l}(S)&:=\{i\in{\bf v}(s):s\in S\cap{\cal S}(d)\}\cap L_{l},\\
	\label{newAtilde}
	\tilde{A}_{l}(S)&:=(\cup_{i\in I(d)}F_{i}({\cal N}))\setminus(\cup^{|{\bf L}_{d}|-1}_{q=l+1}L_{q}\cup A_{l}(S)).
\end{align}
and $I(d)$ is defined in (\ref{id}).  Let ${\cal R}_{d}({\bf F})$ denote the set of rate vectors that satisfy (\ref{biggiewiggie}) for all $S\subseteq{\cal S}(d)$.  Lemma \ref{grandefinale} proves that for any ${\bf R}\in{\cal R}_{d}({\bf F})$, there exists some layering ${\bf L}_{d}$ such that ${\bf R}\in{\cal R}({\bf F},{\bf L}_{d})$, which finishes Step 1 of the proof of Theorem \ref{theoremone}.  

For any $S\subseteq{\cal S}(d)$, let $V(S):=\{j\in{\bf v}(s):s\in S\}$.  For Step 2, we define the notion of a \textit{coarse} layering or partition in which $({\bf F},{\bf L}_{d})$ satisfies the following condition:

(P1) For some non-empty $S_{c}\subseteq{\cal S}(d)$, fixed $k\geq 1$, and all $l\in\{0,\ldots,|{\bf L}_{d}|-1\}$:
\begin{align}
	\label{coarse1}
	V(S_{c})&\subseteq\cup^{l}_{q=0}L_{q},\\
	\label{coarse2}
	V({\cal S}(d)\setminus S_{c})&\subseteq\cup^{|{\bf L}_{d}|-1}_{q=l+k}L_{q}.
\end{align}

We successively refine the coarse layering by removing partition layers using a $\text{\sc shift}$ operator.  Given some $({\bf F},{\bf L}_{d})$ and $S_{c}\subseteq{\cal S}(d)$ that satisfies (P1), define ${\bf L}^{\prime}_{d}:=\text{\sc shift}({\bf L}_{d},S_{c})$ where ${\bf L}^{\prime}_{d}:=(L^{\prime}_{0},L^{\prime}_{1},\ldots,L^{\prime}_{|{\bf L}^{\prime}_{d}|-1})$ is a layering of $F_{d}({\cal N})$ that satisfies (L1)-L(5).  Let $\text{\sc layer}^{\prime}(\cdot)$ correspond to ${\bf L}^{\prime}_{d}$, and for every $i\in F_{d}({\cal N})$:
\begin{align}
\label{shift5}
\text{\sc layer}^{\prime}(i)=\begin{cases}l&i\in A_{l}({\cal N})\setminus A_{l}(S_{c}),\\ l+1 & i\in A_{l}(S_{c}).\end{cases}	
\end{align}
Let $\{{\bf v}^{\prime}(s):s\in{\cal S}(d)\}$ denote the virtual flows generated by $({\bf F},{\bf L}^{\prime}_{d})$ where ${\bf v}^{\prime}(s)$ satisfies (V1)-(V5) for all $s\in{\cal S}(d)$.  Let ${\cal R}({\bf F},{\bf L}^{\prime}_{d})$ denote the set of rate vectors that satisfy:
\begin{align}
	\label{rfdprime}
	R_{S}&<\sum^{|{\bf L}_{d}|-1}_{l=0}I(X_{A^{\prime}_{l}(S)};Y_{d}|X_{\tilde{A}^{\prime}_{l}(S)}),
\end{align}
for all $S\subseteq{\cal S}(d)$ where, 
\begin{align}
	\label{newAprime}
	A^{\prime}_{l}(S)&:=\{i\in{\bf v}^{\prime}(s):s\in S\cap{\cal S}(d)\}\cap L^{\prime}_{l},\\
	\label{newAprimetilde}
	\tilde{A}^{\prime}_{l}(S)&:=(\cup_{i\in I(d)}F_{i}({\cal N}))\setminus(\cup^{|{\bf L}^{\prime}_{d}|-1}_{q=l+1}L^{\prime}_{q}\cup A^{\prime}_{l}(S)).
\end{align}
For all $j\in{\bf f}(v(s),d)$:
\begin{align}
	\label{above}
	\text{\sc layer}(v(s))-\text{\sc layer}(j)\leq k_{v(s),j},
\end{align}
where (\ref{above}) follows from (\ref{firstIndexCoding}).  For all $j\in{\bf f}(s,v(s))$:
\begin{align}
	\label{below}
	\text{\sc layer}(j)-\text{\sc layer}(v(s))>k_{j,v(s)},
\end{align}
where (\ref{below}) follows from (\ref{secondIndexCoding}).  Similarly, for all $j\in{\bf f}(v^{\prime}(s),d)$:
\begin{align}
	\label{aboveprime}
	\text{\sc layer}^{\prime}(v^{\prime}(s))-\text{\sc layer}^{\prime}(j)\leq k_{v^{\prime}(s),j},
\end{align}
where (\ref{aboveprime}) follows from (\ref{firstIndexCoding}).  For all $j\in{\bf f}(s,v^{\prime}(s))$:
\begin{align}
	\label{belowprime}
	\text{\sc layer}^{\prime}(j)-\text{\sc layer}^{\prime}(v^{\prime}(s))>k_{j,v^{\prime}(s)},
\end{align}
where (\ref{belowprime}) follows from (\ref{secondIndexCoding}).  We now show that redundant layers can be eliminated from a coarse layering using the $\text{\sc shift}$ operator defined in (\ref{shift5}).  

\begin{lemma}
	\label{fine}
	 Suppose $({\bf F},{\bf L}_{d})$ satisfies (P1) for some $S_{c}\subseteq{\cal S}(d)$.  If ${\bf L}^{\prime}_{d}=\text{\sc shift}({\bf L}_{d},S_{c})$ then ${\cal R}({\bf F},{\bf L}_{d})\subseteq{\cal R}({\bf F},{\bf L}^{\prime}_{d})$.
\end{lemma}
\begin{proof}
We prove some preliminary claims.

\hspace{-3.5mm}{\bf Claim \ref{fine}.1.} \textit{If $i\in{\bf v}(s)$ for some $s\in{\cal S}(d)\setminus S_{c}$ then $i\in{\bf v}^{\prime}(s)$.}
\begin{proof}
The premise of Lemma \ref{fine} is that $({\bf F},{\bf L}_{d})$ satisfies (P1) for $S_{c}$.  Therefore (\ref{coarse2}) and the premise of Claim \ref{fine}.1, that $i\in{\bf v}(s)$ for some $s\in{\cal S}(d)\setminus S_{c}$, together imply that $i\in\cup^{|{\bf L}_{d}|-1}_{q=l+k}L_{q}$.  It follows that $\text{\sc layer}(i)\geq l+k$.  Moreover,
\begin{align}
	\label{vsnotinVS}
	v(s)\notin V(S_{c}),
\end{align}
since otherwise (\ref{coarse1}) implies $\text{\sc layer}(v(s))\leq l$, and $\text{\sc layer}(i)>\text{\sc layer}(v(s))$ which contradicts (V5).  Therefore (\ref{newA}), (\ref{shift5}), and (\ref{vsnotinVS}) imply that:
\begin{align}
	\label{layervs}
	\text{\sc layer}^{\prime}(v(s))=\text{\sc layer}(v(s)).
\end{align}
First we show that $v^{\prime}(s)=v(s)$ by verifying (\ref{aboveprime}) and (\ref{belowprime}).

\hspace{-3.5mm}{\bf Case 1:} $j\in\cup^{|{\bf L}_{d}|-1}_{q=l+k}L_{q}$.

If $j\in{\bf f}(v^{\prime}(s),d)$, we have the following inequalities.
\begin{align}
	\nonumber
	&\text{\sc layer}^{\prime}(v^{\prime}(s))-\text{\sc layer}^{\prime}(j)\\
	\label{fine1}
	&\hspace{30mm}=\text{\sc layer}^{\prime}(v(s))-\text{\sc layer}^{\prime}(j),\\
	\label{fine2}
	&\hspace{30mm}=\text{\sc layer}(v(s))-\text{\sc layer}^{\prime}(j),\\
	\label{prefine3}
	&\hspace{30mm}=\text{\sc layer}(v(s))-\text{\sc layer}(j),\\
	\label{fine3}
	&\hspace{30mm}\leq k_{v(s),j},\\
	\label{fine4}
	&\hspace{30mm}=k_{v^{\prime}(s),j},
\end{align}
where (\ref{fine1}) follows from the hypothesis that $v^{\prime}(s)=v(s)$, (\ref{fine2}) follows from (\ref{layervs}), and (\ref{prefine3}) follows because the premise of Case 1 that $j\in\cup^{|{\bf L}_{d}|-1}_{q=l+k}L_{q}$ and (\ref{coarse1}) imply that $j\notin V(S_{c})$.  Therefore (\ref{newA}) and (\ref{shift5}) imply that $\text{\sc layer}^{\prime}(j)=\text{\sc layer}(j)$.  Moreover, (\ref{fine3}) follows from (\ref{above}) and the hypothesis that $v^{\prime}(s)=v(s)$ since $j\in{\bf f}(v^{\prime}(s),d)$ implies $j\in{\bf f}(v(s),d)$ if $v^{\prime}(s)=v(s)$, and (\ref{fine4}) follows from the hypothesis that $v^{\prime}(s)=v(s)$.  It follows from (\ref{fine4}) that $v^{\prime}(s)=v(s)$ satisfies (\ref{aboveprime}) for all $j\in{\bf f}(v^{\prime}(s),d)$ and $j\in\cup^{|{\bf L}_{d}|-1}_{q=l+k}L_{q}$.  

If $j\in{\bf f}(s,v^{\prime}(s))$, we have the following inequalities: 
\begin{align}
	\nonumber
	&\text{\sc layer}^{\prime}(j)-\text{\sc layer}^{\prime}(v^{\prime}(s))\\
	\label{fine5}
	&\hspace{30mm}=\text{\sc layer}^{\prime}(j)-\text{\sc layer}^{\prime}(v(s)),\\
	\label{fine6}
	&\hspace{30mm}=\text{\sc layer}^{\prime}(j)-\text{\sc layer}(v(s)),\\
	\label{prefine7}
	&\hspace{30mm}=\text{\sc layer}(j)-\text{\sc layer}(v(s)),\\
	\label{fine7}
	&\hspace{30mm}>k_{j,v(s)},\\
	\label{fine8}
	&\hspace{30mm}=k_{j,v^{\prime}(s)},
\end{align} 
where (\ref{fine5}) follows from the hypothesis that $v^{\prime}(s)=v(s)$, (\ref{fine6}) follows from (\ref{layervs}), (\ref{prefine7}) follows because the premise of Case 1 that $j\in\cup^{|{\bf L}_{d}|-1}_{q=l+k}L_{q}$ and (\ref{coarse1}) imply that $j\notin V(S_{c})$.  Therefore (\ref{newA}) and (\ref{shift5}) imply that $\text{\sc layer}^{\prime}(j)=\text{\sc layer}(j)$.  Moreover, (\ref{fine7}) follows from (\ref{below}) and the hypothesis that $v^{\prime}(s)=v(s)$ since $j\in{\bf f}(v^{\prime}(s),d)$ implies $j\in{\bf f}(v(s),d)$ if $v^{\prime}(s)=v(s)$, and (\ref{fine8}) follows from the hypothesis that $v^{\prime}(s)=v(s)$.  It follows from (\ref{fine8}) that $v^{\prime}(s)=v(s)$ satisfies (\ref{secondIndexShiftedComplete}) for all $j\in{\bf f}(v^{\prime}(s),d)$ and $j\in\cup^{|{\bf L}_{d}|-1}_{q=l+k}L_{q}$.

\hspace{-3.5mm}{\bf Case 2:} $j\in\cup^{l}_{q=0}L_{q}$.

If $j\in{\bf f}(v^{\prime}(s),d)$, we have the following inequalities.
\begin{align}
	\nonumber
	&\text{\sc layer}^{\prime}(v^{\prime}(s))-\text{\sc layer}^{\prime}(j)\\
	\label{fine9}
	&\hspace{21mm}=\text{\sc layer}^{\prime}(v(s))-\text{\sc layer}^{\prime}(j),\\
	\label{fine10}
	&\hspace{21mm}=\text{\sc layer}(v(s))-\text{\sc layer}^{\prime}(j),\\
	\label{fine11}
	&\hspace{21mm}\leq\text{\sc layer}(v(s))-\text{\sc layer}(j),\\
	\label{fine12}
	&\hspace{21mm}\leq k_{v(s),j},\\
	\label{fine13}
	&\hspace{21mm}=k_{v^{\prime}(s),j},
\end{align}
where (\ref{fine9}) follows from the hypothesis that $v^{\prime}(s)=v(s)$, (\ref{fine10}) follows from (\ref{layervs}), (\ref{fine11}) follows because (\ref{shift5}) implies that $\text{\sc layer}(j)\leq\text{\sc layer}^{\prime}(j)\leq\text{\sc layer}(j)+1$.  Moreover, (\ref{fine12}) follows from (\ref{above}) and the hypothesis that $v^{\prime}(s)=v(s)$ since $j\in{\bf f}(v^{\prime}(s),d)$ implies $j\in{\bf f}(v(s),d)$ if $v^{\prime}(s)=v(s)$, and (\ref{fine13}) follows from the hypothesis that $v^{\prime}(s)=v(s)$.  It follows from (\ref{fine13}) that $v^{\prime}(s)=v(s)$ satisfies (\ref{aboveprime}) for all $j\in{\bf f}(v^{\prime}(s),d)$ and $j\in\cup^{l}_{q=0}L_{q}$.

If $v^{\prime}(s)=v(s)$ and $j\in\cup^{l}_{q=0}L_{q}$ we show that $j\notin{\bf f}(s,v^{\prime}(s))$.  It follows from (\ref{coarse1}) and (\ref{vsnotinVS}) that $\text{\sc layer}(v(s))\geq l+k$.  If $j\in{\bf f}(s,v^{\prime}(s))$ and $v^{\prime}(s)=v(s)$ then $j\in{\bf f}(s,v(s))$.  By the case assumption, $j\in\cup^{l}_{q=0}L_{q}$, which implies $\text{\sc layer}(j)\leq l$.  This contradicts (\ref{belowprime}), so $j\notin{\bf f}(s,v^{\prime}(s))$.  Therefore $v^{\prime}(s)=v(s)$ satisfies (\ref{aboveprime}) and (\ref{belowprime}) for both Case 1 and Case 2. 

Finally, for any $s\in{\cal S}(d)\setminus S_{c}$, we show that $i\in{\bf v}^{\prime}(s)$ if $i\in{\bf v}(s)$.  We have the following equalities:
\begin{align}
	\nonumber
	&\text{\sc layer}^{\prime}(v^{\prime}(s))-\text{\sc layer}^{\prime}(i)\\
	\label{fine14}
	&\hspace{21mm}=\text{\sc layer}^{\prime}(v(s))-\text{\sc layer}^{\prime}(i),\\
	\label{fine15}
	&\hspace{21mm}=\text{\sc layer}(v(s))-\text{\sc layer}^{\prime}(i),\\
	\label{fine16}
	&\hspace{21mm}=\text{\sc layer}(v(s))-\text{\sc layer}(i),\\
	\label{fine17}
	&\hspace{21mm}=k_{v(s),i},\\
	\label{fine18}
	&\hspace{21mm}=k_{v^{\prime}(s),i},
\end{align}
where (\ref{fine14}) follows from the hypothesis that $v^{\prime}(s)=v(s)$ and (\ref{fine15}) follows from (\ref{layervs}).  Since $i\in{\bf v}(s)$ for some $s\in{\cal S}(d)\setminus S$, it follows from (\ref{coarse1}) and (\ref{coarse2}) that $i\notin V(S_{c})$.  Therefore (\ref{fine16}) follows from both (\ref{newA}) and (\ref{shift5}).  The premise of Claim \ref{fine}.1 that $i\in{\bf v}(s)$ and (V5) imply (\ref{fine17}), and (\ref{fine18}) follows because $v^{\prime}(s)=v(s)$.  It follows from (V5) and (\ref{fine18}) that $i\in{\bf v}^{\prime}(s)$ which completes the proof of Claim \ref{fine}.1.
\end{proof}

\hspace{-3.5mm}{\bf Claim \ref{fine}.2.} \textit{If $i\in{\bf v}(s)$ for some $s\in S_{c}$ then $i\in{\bf v}^{\prime}(s)$.}
\begin{proof}
There are two cases to consider.  If $v^{\prime}(s)=v(s)$ then Claim \ref{join}.1 implies $i\in{\bf v}^{\prime}(s)$.  Otherwise if $v^{\prime}(s)\neq v(s)$, then Claim \ref{join}.2 implies $i\in{\bf v}^{\prime}(s)$.  Both claims impose no conditions on the $S\subseteq{\cal S}(d)$ in the shift operator.
\end{proof}

We can now complete the proof of Lemma \ref{fine}.  Together (\ref{coarse1}), (\ref{coarse2}) and Claim \ref{fine}.1 imply that $A_{q}(S)\subseteq A_{q}^{\prime}(S)$ for all $q\in\{l+k,\ldots,|{\bf L}^{\prime}_{d}|-1\}$ and $S\subseteq{\cal S}(d)$.  It follows from Claim \ref{induction}.2 that:  
	\begin{align}
		\label{firstA}
		I(X_{A^{\prime}_{q}(S)};Y_{d}|X_{\tilde{A}^{\prime}_{q}(S)})>I(X_{A_{q}(S)};Y_{d}|X_{\tilde{A}_{q}(S)}),
	\end{align}
	for all $S\subseteq{\cal S}(d)$ and $q\in\{l+k,\ldots,|{\bf L}^{\prime}_{d}|-1\}$.  Similarly, (\ref{newA}), (\ref{coarse1}), (\ref{shift5}), and Claim \ref{fine}.2 imply that $A_{q-1}(S)\subseteq A_{q}^{\prime}(S)$ for all $q\in\{1,\ldots,l+1\}$ and $S\subseteq{\cal S}(d)$.  It follows from Claim \ref{induction}.2 that:
	\begin{align}
		\label{secondA}
		\hspace{-2mm}I(X_{A^{\prime}_{q}(S)};Y_{d}|X_{\tilde{A}^{\prime}_{q}(S)})>I(X_{A_{q-1}(S)};Y_{d}|X_{\tilde{A}_{q-1}(S)}),
	\end{align}  
	for all $S\subseteq{\cal S}(d)$ and $q\in\{0,\ldots,l\}$.  Since ${\bf R}\in{\cal R}({\bf F},{\bf L}_{d})$ as defined (\ref{rfd}) is the premise of Lemma \ref{fine}, it follows from (\ref{firstA}) and (\ref{secondA}) that ${\bf R}\in{\cal R}({\bf F},{\bf L}^{\prime}_{d})$ as defined in (\ref{rfdprime}).
\end{proof}

We also normalize ${\bf L}_{d}$ to remove any redundant initial layers so that $L_{0}\neq\{\}$.  Given ${\bf L}_{d}=\{L_{0},\ldots,L_{|{\bf L}_{d}|-1}\}$, define:
\begin{align}
	\nonumber
	M&:=\{l:L_{l}=\{\}\},\\
	\nonumber
	l_{0}&:=\min_{l\in M}l.
\end{align}
Now define ${\bf L}^{\prime}_{d}:=\text{\sc norm}({\bf L}_{d},l_{0})$ where ${\bf L}^{\prime}_{d}:=(L^{\prime}_{0},L^{\prime}_{1},\ldots,L^{\prime}_{|{\bf L}^{\prime}_{d}|-1})$ is a layering of $F_{d}({\cal N})$.  Let $\text{\sc layer}^{\prime}(\cdot)$ correspond to ${\bf L}^{\prime}_{d}$ and for every $i\in F_{d}({\cal N})$:
\begin{align}
	\label{norm}
	\text{\sc layer}^{\prime}(i)&=\text{\sc layer}(i)-l_{0}.
\end{align}
Let $\{{\bf v}^{\prime}(s):s\in{\cal S}(d)\}$ denote the virtual flows generated by $({\bf F},{\bf L}^{\prime}_{d})$.  Let ${\cal R}({\bf F},{\bf L}^{\prime}_{d})$ denote the set of rate vectors that satisfy (\ref{flowdecomp3}) where $A^{\prime}_{l}(\cdot)$ and $\tilde{A}^{\prime}_{l}(\cdot)$ are defined in (\ref{newAprime}) and (\ref{newAprimetilde}).  Lemma \ref{vssuffices} implies that $v(s)$ and $v^{\prime}(s)$ uniquely satisfy (\ref{above}), (\ref{below}), (\ref{aboveprime}) and (\ref{belowprime}).
\begin{lemma}
	\label{normalize}
	Fix ${\bf F}$ and suppose ${\bf L}^{\prime}_{d}=\text{\sc norm}({\bf L}_{d},l_{0})$.  Then ${\cal R}({\bf F},{\bf L}_{d})\subseteq{\cal R}({\bf F},{\bf L}^{\prime}_{d})$.
\end{lemma}
\begin{proof}
	First we show that $v^{\prime}(s)=v(s)$ satisfies (\ref{aboveprime}) and (\ref{belowprime}).  For all $j\in{\bf f}(v(s),d)$:
	\begin{align}
		\nonumber
		&\text{\sc layer}(v(s))-\text{\sc layer}^{\prime}(j)\\
		\label{normalize1}
		&\hspace{7mm}=(\text{\sc layer}(v(s))-l_{0})-(\text{\sc layer}(j)-l_{0}),\\
		\nonumber
		&\hspace{7mm}=\text{\sc layer}(v(s))-\text{\sc layer}(j),\\
		\label{normalize3}
		&\hspace{7mm}\leq k_{v(s),j},\\
		\label{normalize4}
		&\hspace{7mm}=k_{v^{\prime}(s),j},
	\end{align}
	where (\ref{normalize1}) follows from (\ref{norm}), (\ref{normalize3}) follows from (\ref{above}) and the hypothesis that $v^{\prime}(s)=v(s)$ since $j\in{\bf f}(v^{\prime}(s),d)$ implies $j\in{\bf f}(v(s),d)$ if $v^{\prime}(s)=v(s)$, and (\ref{normalize4}) follows from the hypothesis that $v^{\prime}(s)=v(s)$.  It follows from (\ref{normalize4}) that $v^{\prime}(s)=v(s)$ satisfies (\ref{aboveprime}).  For all $j\in{\bf f}(s,v^{\prime}(s))$:
	\begin{align}
		\nonumber
		&\text{\sc layer}^{\prime}(j)-\text{\sc layer}^{\prime}(v(s))\\
		\label{normalize5}
		&\hspace{7mm}=(\text{\sc layer}(j)-l_{0})-(\text{\sc layer}(v(s))-l_{0}),\\
		\nonumber
		&\hspace{7mm}=\text{\sc layer}(j)-\text{\sc layer}(v(s)),\\
		\label{normalize6}
		&\hspace{7mm}> k_{j,v(s)},\\
		\label{normalize7}
		&\hspace{7mm}=k_{j,v^{\prime}(s)},
	\end{align}
	where (\ref{normalize5}) follows from (\ref{norm}), (\ref{normalize6}) follows from (\ref{below}) and the hypothesis that $v^{\prime}(s)=v(s)$ since $j\in{\bf f}(s,v^{\prime}(s))$ implies $j\in{\bf f}(s,v(s))$ if $v^{\prime}(s)=v(s)$, and (\ref{normalize7}) follows from the hypothesis that $v^{\prime}(s)=v(s)$.  It follows from (\ref{normalize7}) that $v^{\prime}(s)=v(s)$ satisfies (\ref{belowprime}).  Finally we show that $i\in{\bf v}(s)$ implies $i\in{\bf v}^{\prime}(s)$ for all $s\in{\cal S}(d)$.
	\begin{align}
		&\text{\sc layer}^{\prime}(v^{\prime}(s))-\text{\sc layer}^{\prime}(i),\\
		\label{normalize8}
		&\hspace{7mm}=\text{\sc layer}^{\prime}(v(s))-\text{\sc layer}^{\prime}(i),\\
		\label{normalize9}
		&\hspace{7mm}=(\text{\sc layer}(v(s))-l_{0})-(\text{\sc layer}(i)-l_{0}),\\
		\nonumber
		&\hspace{7mm}=\text{\sc layer}(v(s))-\text{\sc layer}(i),\\
		\label{normalize11}
		&\hspace{7mm}=k_{v(s),i},\\
		\label{normalize12}
		&\hspace{7mm}=k_{v^{\prime}(s),i},
	\end{align}
	where  (\ref{normalize8}) follows because $v^{\prime}(s)=v(s)$, (\ref{normalize9}) follows from (\ref{norm}), (\ref{normalize11}) follows from (V5) and the premise of Lemma \ref{normalize} that $i\in{\bf v}(s)$, and (\ref{normalize12}) follows because $v^{\prime}(s)=v(s)$.  It follows from (V5) and (\ref{normalize12}) that $i\in{\bf v}^{\prime}(s)$.
	
	Since $i\in{\bf v}(s)$ implies $i\in{\bf v}^{\prime}(s)$, it follows from (\ref{norm}) that $A_{q}(S)\subseteq A^{\prime}_{q-l_{0}}(S)$ for all $q=0,\ldots,|{\bf L}_{d}|-1$ and $S\subseteq{\cal S}(d)$.  Therefore Claim \ref{induction}.2 implies:
	\begin{align}
		\label{normalize13}
		\hspace{-2.8mm}I(X_{A^{\prime}_{q}(S)};Y_{d}|X_{\tilde{A}^{\prime}_{q}(S)})\geq I(X_{A_{q-l_{0}}(S)};Y_{d}|X_{\tilde{A}_{q-l_{0}}(S)})
	\end{align}
	for all $S\subseteq{\cal S}(d)$ and $q\in\{0,\ldots,|{\bf L}_{d}|-1\}$.  Moreover, (\ref{norm}) implies $|{\bf L}_{d}|\geq|{\bf L}^{\prime}_{d}|$.  Therefore, if ${\bf R}\in{\cal R}({\bf F},{\bf L}_{d})$, it follows from (\ref{normalize13}) that ${\bf R}\in{\cal R}({\bf F},{\bf L}^{\prime}_{d})$, which implies ${\cal R}({\bf F},{\bf L}_{d})\subseteq{\cal R}({\bf F},{\bf L}^{\prime}_{d})$.
\end{proof}

We modify the ${\bf L}_{d,n^{*}}$ obtained at the end of Step 1 of the proof of Theorem \ref{theoremone}, by removing the superfluous partition layers as described in Algorithm \ref{algo}.

\begin{figure}[t]
        \center{\includegraphics[width=0.5\textwidth]
        {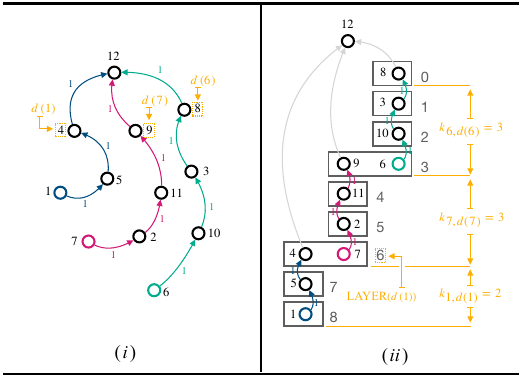}}
        \caption{(i) The flows ${\bf F}=\{{\bf f}(1),{\bf f}(6),{\bf f}(7)\}$ where ${\bf f}(1)=1\xrightarrow{1}5\xrightarrow{1}4\xrightarrow{1}12$, ${\bf f}(6)=6\xrightarrow{1}10\xrightarrow{1}3\xrightarrow{1}8\xrightarrow{1}12$, and ${\bf f}(7)=7\xrightarrow{1}2\xrightarrow{1}11\xrightarrow{1}9\xrightarrow{1}12$.  ${\cal S}(12)=\{1,6,7\}$. (ii) The flows ${\bf F}$ and layering ${\bf L}_{12}$ where ${\bf L}_{12}=(\{8\},\{3\},\{10\},\{6,9\},\{11\},\{2\},\{4,7\},\{5\},\{1\})$.  The virtual flows are ${\bf v}(1)={\bf f}(1,12)=1\xrightarrow{1}5\xrightarrow{1}4$, ${\bf v}(6)={\bf f}(6,12)=6\xrightarrow{1}10\xrightarrow{1}3\xrightarrow{1}8$, and ${\bf v}(7)={\bf f}(7,12)=7\xrightarrow{1}2\xrightarrow{1}11\xrightarrow{1}9$.  For $d=12$, it follows that $d(1)=4$, $d(6)=8$, and $d(7)=9$.  Since $({\bf F},{\bf L}_{12})$ violates (P1) for all subsets of ${\cal S}(12)$ and $L_{0}\neq\{\}$, it follows that $\text{\sc layer}(d(4))\leq k_{6,d(6)}+k_{7,d(7)}=6$.}
        \label{stacked}
 \end{figure}

\begin{algorithm}
\SetAlgoLined
\KwResult{An ${\bf L}_{d}$ for which ${\cal R}({\bf F},{\bf L}_{d,n^{*}})\subseteq{\cal R}({\bf F},{\bf L}_{d})$, $({\bf F}, {\bf L}_{d})$ violates (P1) for all $S_{c}\subseteq{\cal S}(d)$, and $L_{0}\neq\{\}$.}
 ${\bf L}_{d}:={\bf L}_{d,n^{*}}$\\
 \While{${\bf L}_{d}$ satisfies (P1) for some $S_{c}\subseteq{\cal S}(d)$}{
${\bf L}^{\prime}_{d}:=\text{\sc shift}({\bf L}_{d},S_{c})$\;
${\bf L}_{d}:={\bf L}^{\prime}_{d}$\;
 }
 ${\bf L}^{\prime}_{d}:=\text{\sc norm}({\bf L}_{d},l_{0})$\;
 ${\bf L}_{d}:={\bf L}^{\prime}_{d}$\;
 \caption{Trim and normalize $({\bf F},{\bf L}_{d,n^{*}})$}
 \label{algo}
\end{algorithm}

First we show that Algorithm \ref{algo} terminates.  Given $({\bf F},{\bf L}^{\prime}_{d})$ in Line 3, it follows from Claim \ref{fine}.1 and Claim \ref{fine}.2 that $V(S_{c})\subseteq V^{\prime}(S_{c})$ and $V({\cal S}(d)\setminus S_{c})\subseteq V^{\prime}({\cal S}(d)\setminus S_{c})$.  Moreover, $\text{\sc layer}^{\prime}(j)=\text{\sc layer}(j)$ if $j\in\cup^{|{\bf L}_{d}|-1}_{q=l+k}L_{q}$ and $\text{\sc layer}^{\prime}(j)=\text{\sc layer}(j)+1$ if $j\in\cup^{l+1}_{q=0}L_{q}$.  Now suppose $({\bf F},{\bf L}^{\prime}_{d})$ still satisfies (P1) at $S_{c}$.  For some $m_{1}\geq 1$, and $m_{2}\geq 0$ it follows that:
\begin{align}
	\label{coarse1prime}
	V^{\prime}(S_{c})\subseteq\cup^{l+m_{1}}_{q=0}L^{\prime}_{q},\\
	\label{coarse2prime}
	V^{\prime}({\cal S}(d)\setminus S_{c})\subseteq\cup^{|{\bf L}^{\prime}_{d}|-1}_{q=l+k-m_{2}}L^{\prime}_{q}.
\end{align} 
Comparing (\ref{coarse1prime}) and (\ref{coarse2prime}) with (\ref{coarse1}) and (\ref{coarse2}), it follows that ${\bf L}^{\prime}_{d}$ has $m_{1}+m_{2}$ fewer partition layers than ${\bf L}_{d}$.  Therefore the exit condition in Line 2 for $S_{c}$ will be triggered after some finite number of iterations.  Since the number of subsets of ${\cal S}(d)$ is finite, Algorithm \ref{algo} must terminate.  

Next we show that Algorithm \ref{algo} returns a layering ${\bf L}_{d}$ in which $L_{0}\neq\{\}$, $({\bf F},{\bf L}_{d})$ violates (P1) defined by (\ref{coarse1}) and (\ref{coarse2}), for all $S_{c}\subseteq{\cal S}(d)$, and ${\cal R}({\bf F},{\bf L}_{d,n^{*}})\subseteq{\cal R}({\bf F},{\bf L}_{d})$.  Lemma \ref{fine} implies that ${\cal R}({\bf F},{\bf L}_{d})\subseteq{\cal R}({\bf F},{\bf L}^{\prime}_{d})$ in Line 3 and Lemma \ref{normalize} implies that ${\cal R}({\bf F},{\bf L}_{d})\subseteq{\cal R}({\bf F},{\bf L}^{\prime}_{d})$ in Line 6.  Since Algorithm \ref{algo} terminates, it follows that ${\bf L}_{d}$ in Line 6 violates (P1) for all $S_{c}\subseteq{\cal S}(d)$.  By construction, ${\bf L}_{d}$ in Line 7 is normalized, so $L_{0}\neq\{\}$, which proves that Algorithm \ref{algo} is correct. 

Finally, we construct an ${\bf F}^{\prime}$ equivalent to ${\bf F}$ so that $({\bf F}^{\prime},{\bf L}_{d})$ satisfies (C1) and (C2).  Let ${\bf f}(s):=N_{1}\xrightarrow{k_{1}}N_{2}\xrightarrow{k_{2}}\cdots\xrightarrow{k_{e_{s}-1}}N_{e_{s}}$ and ${\bf F}:=\{{\bf f}(s):s\in{\cal N}\}$.  By assumption ${\bf F}$ satisfies (C2).  Suppose $d\in N_{l}$ for some $l\in\{1,\ldots,e_{s}\}$, and let $d(s):=N_{l-1}$ denote the one-hop predecessor(s) of node $d$ in ${\bf f}(s,d)$.  Recall from (F5) that $k_{s,d(s)}:=\sum^{l-2}_{q=1}k_{q}$.  Since $d\in N_{l}$, and $d(s)=N_{l-1}$, it follows that $k_{d(s),d}=k_{l-1}$.  

For every $s\in{\cal S}(d)$, set ${\bf f}^{\prime}(s):={\bf f}(s)$.  We modify the one-hop encoding delays $k^{\prime}_{d(s),d}$ as follows.  For every $s\in{\cal S}(d)$, set:
\begin{align}
	\label{finfin1}
  	k^{\prime}_{d(s),d}&:=\left[\sum_{i\in{\cal S}(d)\setminus\{s\}}k_{i,d(i)}\right]+1.
\end{align}

Observe that ${\bf f}^{\prime}(s,d)={\bf f}(s,d)$ and $k^{\prime}_{s,d(s)}=k_{s,d(s)}$ for all $s\in{\cal S}(d)$.  Let ${\bf F}^{\prime}:=\{{\bf f}^{\prime}(s):s\in{\cal S}(d)\}$.  By inspection, ${\bf F}^{\prime}$ is ``equivalent'' to ${\bf F}$ since the flows only differ in the encoding delays.  Therefore ${\bf F}^{\prime}$ satisfies (C2).  To verify (C1), we check that:
\begin{align}
	\label{finfin2}
	\text{\sc layer}(v(s))<k_{v(s),d}.
\end{align}
We have the following:
\begin{align}
	\nonumber
	\hspace{-8mm}\text{\sc layer}(v(s))&=(\text{\sc layer}(v(s))-\text{\sc layer}(d(s)))\\
	\nonumber
	&\hspace{34mm}+\text{\sc layer}(d(s)),\\
	\label{finfin3}
	&\leq k_{v(s),d(s)}+\text{\sc layer}(d(s)),
\end{align}
where (\ref{finfin3}) follows from (\ref{above}) since $d(s)\in{\bf f}(v(s),d)$.  Furthermore,
\begin{align}
	\label{finfin4}
	k_{v(s),d}&=k_{v(s),d(s)}+k_{d(s),d}.
\end{align}

Therefore if ${\bf L}_{d}$ satisfies:
\begin{align}
	\label{finfin5}
	\text{\sc layer}(d(s))<k_{d(s),d},
\end{align}
for all $s\in{\cal S}(d)$, it follows from (\ref{finfin3}) and (\ref{finfin4}) that ${\bf L}_{d}$ satisfies (\ref{finfin2}) and thus (C1).

By construction in Algorithm \ref{algo}, ${\bf L}_{d}$ violates (P1) (as defined by (\ref{coarse1}) and (\ref{coarse2})) for all $S_{c}\subseteq{\cal S}(d)$.  Therefore ${\bf L}_{d}$ violates (P1) for all $S_{c}\subseteq{\cal S}(d)$, since ${\bf f}^{\prime}(s,d)={\bf f}(s,d)$ for all $s\in{\cal S}(d)$.  For every $s\in{\cal S}(d)$, it follows that:
	\begin{align}
		\label{finfin6}
		\text{\sc layer}(d(s))&\leq\sum_{i\in{\cal S}(d)\setminus\{s\}}k^{\prime}_{i,d(i)},
	\end{align} 
since $L_{0}\neq\{\}$.  Equality in (\ref{finfin6}) only occurs if ${\bf v}(s)={\bf f}^{\prime}(s,d)$ for all $s\in{\cal S}(d)$ and the flows in $\{{\bf v}(i):s\in{\cal S}(d)\setminus\{s\}\}$ are sequentially ``stacked'' in ${\bf L}_{d}$ (as depicted in Figure \ref{stacked}(ii).  For all $s\in{\cal S}(d)$, ${\bf F}^{\prime}$ by construction satisfies:
\begin{align}
	\label{finfin7}
	k^{\prime}_{d(s),d}&:=\left[\sum_{i\in{\cal S}(d)\setminus\{s\}}k^{\prime}_{i,d(i)}\right]+1,
\end{align}
where (\ref{finfin7}) follows from (\ref{finfin1}) and the fact that $k^{\prime}_{s,d(s)}:=k_{s,d(s)}$.  It follows from (\ref{finfin6}) and (\ref{finfin7}) that:
\begin{align}
	\label{finfin8}
	\text{\sc layer}(d(s))<k^{\prime}_{d(s),d}.
\end{align}  
Comparing (\ref{finfin5}) and (\ref{finfin8}) proves that ${\bf L}_{d}$ satisfies (C1).  For any ${\bf R}\in{\cal R}_{d}({\bf F})$ where ${\bf F}$ satisfies (C2) it follows that ${\bf R}\in{\cal R}({\bf F}^{\prime},{\bf L}_{d})$ for some equivalent ${\bf F}^{\prime}$ and layering ${\bf L}_{d}$ that satisfies (C1), which completes the proof of Theorem \ref{theoremone}.
\section{Conclusion}
\label{conclusion}
Two sets of constraints define an outer-bound on the DF region: the mutual-information constraints derived from the one-relay channel, and causality constraints that ensure that nodes only forward messages they have already decoded.  We introduced the framework of flow decomposition to prove these constraints are also sufficient.  

In general networks there are intractably many ways of routing and decoding information.  Nevertheless, these individual schemes have an underlying  order; they contribute to a larger region whose high-dimensional shape always projects back to the one-dimensional rate of the one-relay channel.  A similar result also applies to the compress-forward setting using a simplified version of the proof presented here.

Even in channels of fixed size, there are infinitely many flows because there are infinitely many encoding delays.  There are also infinitely many layerings.  The causality constraints indicate whether particulars flow and layerings are causal, but say nothing as to whether such flows and layerings exist.  In separate work, we prove that only the mutual-information constraints are active in channels with hierarchical flow, which implies the achievable region has minimal complexity.  In all-cast channels, the achievable region is computable, but at higher complexity.

The DF schemes we consider belong to the index-coding family, which requires independent input distributions and either hierarchical or all-cast flow (effectively).  Correlated inputs support multi-cast flows, but the region itself while still a derivative of the one-relay DF rate, depends on auxiliary random variables.  The capacity of the one-relay channel remains unresolved.

\bibliographystyle{IEEEtran}
\bibliography{flowdecomp}
\appendices
\input{"example.tex"}

\input{"appendixbcd.tex"}

\end{document}

%% file: introduction.tex
\section{Introduction}
  The relay channel consists of three nodes: a source, a relay, and a destination \cite{Meulen1971}.  The source transmits messages to the destination through a “noisy” channel (the output symbols in each use of the channel are statistically related to the input symbols via some probability transition function).  After observing its own channel outputs, the relay attempts to help the destination decode messages sent by the source.  Although the capacity of the relay channel is not fully understood, the channel is still a useful building block for studying general networks from an information-theoretic perspective.  

Two strategies of interest in this channel are decode-forward (DF), where the relay decodes and forwards each message to the destination, and compress-forward (CF), where the relay forwards a compression of the channel observations instead of the message itself \cite{GamalCover}.  Neither strategy is universally better than the other.  %Both strategies capture fundamental issues in general networks involving cooperation, competition, and complexity. 

%The core problem in network information theory is deciding (in finite time) whether or not any desired rate-vector is achievable in some general network.  Since this problem is unsolved for even three-node channels, it may be tempting to dismiss any further study of general channels with arbitrarily many nodes until a solution for three nodes is found.  This would be a mistake.  The difficulty in three-node channels is probabilistic in nature and tied to the optimal way of exploiting correlation between random variables \cite{distributedcoord}.  The problems in networking are of a different flavor, and tend to fall under the scope of graph theory and combinatorics.  Ultimately, all of these problems need to be solved in a single comprehensive framework.  In the meantime, it makes sense to separate the probabilistic issues in network information theory from the combinatorial ones.

The ``complexity'' of an achievable region describes the number of computations required to check whether a desired rate-vector is in the region, with respect to the number of nodes (another aspect of complexity is finding an encoding/decoding scheme that actually achieves this rate-vector).  
%While the capacity region of general channels is unknown, it appears to depend in some way on the number of nodes in the channel and the correlation between inputs.  
We focus on a class of regular DF schemes that generalize index coding.  Regular schemes have the short delays (linear in the channel usage) associated with sliding-window decoding instead of the long (exponential) delays and restrictions on bidirectional communication associated with backward decoding.  A fundamental tension exists in channels with bidirectional communication; nodes can individually benefit by waiting for others to decode first, but some node must always be the first to decode  \cite{Ponniah2016}.  The complexity of the achievable region depends on all the ways of resolving this tension.

Two sets of constraints are necessary for the DF region: first, the mutual-information inequalities derived from the one-relay channel, and second, causality constraints that ensure nodes only forward messages they have already decoded.  We introduce the framework of flow decomposition to prove these two sets of constraints are also sufficient.  Flow decomposition consists of two structures: \textit{flows} and \textit{layerings}.       
Flows define the nodes which forward messages from each source (i.e., the routes), along with the integer-valued encoding delays along each “hop”.  Layerings define the messages decoded by a particular node in each block of channel uses, and correspond to ordered partitions of all the preceding nodes.  %Flow sets define the flows for each source node, and by implication, the messages encoded by all nodes in each block of channel uses.  For any channel with a fixed number of nodes, there are infinitely many flow sets since there are infinitely many possible encoding delays.
%The message-vectors decoded in each block at a particular node correspond to \textit{layerings}, which are ordered partitions of all the preceding nodes in the flow set.  Each disjoint set in the partition is a layer, whose order maps to a previous block of channel uses.  The aforementioned node decodes the information sent by all nodes in this layer in the targeted block.  Deeper layers probe deeper in time.  For any given flow set, there are infinitely many different layerings.
%A \textit{flow decomposition} is a flow set paired with a layering.  It specifies the messages encoded at all nodes and the messages decoded at a fixed node.  
%A distinctive feature of this framework is that messages decoded in the current block of channel uses, are not completely determined by the messages forwarded in the next block.  Decoding schemes have some freedom to operate within the confines of the encoding scheme.  There are constraints, namely, that every scheme be causal; nodes can only forward messages they have already decoded.  However, different “flows” give each node a different level of flexibility in choosing which messages to decode. 
%There is a natural way of extending the decode-forward and compress-forward rates in the one-relay channel to general multi-source multi-relay channels, which gives an outer-bound on the decode-forward region.  This outer-bound depends on the sequence of nodes in the flow set, but not on the encoding delays.  There are no explicit conditions on causality, so the bound could be loose.

Index-coding fits naturally within two types of flow: hierarchical flow, where the routing topology is similar to a tree, and all-cast flow, where the routes cover all nodes.  For any arbitrary flow of either type and any rate-vector satisfying the mutual-information constraints at a specific node, we show  there are ``equivalent'' flows and a layering that achieves the rate-vector and satisfies the causality constraints.  Equivalence, in this context, means the flows are identical up to (but excluding) the encoding delays and thus generate the same mutual-information constraints.  This result implies both sets of constraints are sufficient to define the achievable region; the causality constraints at any fixed node can be ``deactivated'' using equivalent flows.  

Whether or not the causality constraints can be deactivated at \textit{all} nodes simultaneously is a question addressed in separate work.  By way of preview, the answer is yes for hierarchical flow, which implies the achievable region has minimal complexity.  The answer is no for all-cast channels, which raises another question of whether the all-cast region is even computable; there are infinitely many flows and layerings in a channel with a fixed number of nodes.  It turns out the region is computable, but at much higher complexity \cite{Ponniah2018}.
%The main result of this paper is that for any flow set, any rate-vector in the outer-bound \textit{evaluated at a fixed relay}, is achievable.  Moreover, the corresponding flow decomposition is both causal and equivalent to the original flow (causal flow decompositions only encode messages they have already decoded and equivalent flows are identical up to the encoding delays).  Whether this rate-vector is simultaneously achievable at all nodes is a question addressed in future work.  By way of preview, it turns out that the complexity of enforcing causality at all nodes depends on the type of flow.  Hierarchical flow achieves the decode-forward outer-bound, which implies minimal complexity \cite{Ponniah2020}.  All-cast flow, on the other hand, includes bidirectional communication and is subject to the internal tensions mentioned before.  In this complicated setting, flow decomposition provides a computable way of enforcing causality \cite{Ponniah2018}.

The proof of the main result relies on a \textit{shift} operation.  For any fixed  rate-vector satisfying the mutual-information constraints, the shift operation returns a layering with an achievable region  that is ``closer'' to the target rate-vector.  We construct a sequence of shifted layerings and prove the sequence of achievable regions eventually includes the target.  An extra step is required to find equivalent flows and a layering that satisfies the causality constraints.  

The CF setting (not addressed here) requires no notion of flow and no causality constraints, since relays do not forward source  messages.  A simplified version of this proof works for CF schemes \cite{PonniahCompressForward2019}.

%In the compress-forward setting, the notion of “flow” is meaningless because relays decode compressions not the messages themselves.  However, “layerings” retain their meaning.  The proof technique outlined here thus carries over in simpler form, without the tensions and complications from causality that arise in the decode-forward setting.

%The class of decode-forward schemes considered in this paper encodes multiple messages into a single index.  The index rate must exceed the sum rate of all the encoded source messages.  Any desired set of messages can only be decoded if all other messages represented in the index are already known as “side information” \cite{FundamentalsofIndexCoding2018}.  For this reason, hierarchical and all-cast flow fit naturally within the index-coding framework. 

%The decode-forward schemes in this paper inherit the independent input distributions from index-coding and network coding.
%The original decode-forward scheme used superposition coding to induce correlated input distributions.  Although index-coding requires independent input distributions, it is not a special case of super-position coding and the former is sometimes better than the latter \cite{Ponniah2008}.  Flow decomposition, with further development, can also support superposition coding.

Section \ref{background} provides a survey of some previous work.  Section \ref{OutlineAndPreliminaries} provides a high-level overview of the main results and an outline of the proof.  The flow decomposition framework is introduced in Section \ref{FlowDecomposition}, and the main result is presented in Section \ref{mainresult}.  Section \ref{prooftheoremone} includes the proof and Section \ref{conclusion} concludes the paper.%The three-node relay channel, first proposed by Van der Meulen, consists of a source node that communicates to a destination node with the help of a relay node.  By assumption, the symbols transmitted and received by all nodes are discrete, the received symbols depend only the symbols transmitted in the same channel use, and symbols transmitted by the relay may only depend on symbols received in previous channel uses.

\section{Literature Review}
\label{background}
The relay channel was first proposed in \cite{Meulen1971}.  The CF and DF schemes for the one-relay channel appear in \cite{GamalCover}.  In its original form, the DF scheme combined super-position coding, random binning, and list-decoding.  This scheme was simplified and streamlined for multi-relay channels in \cite{KramerGastpar} and \cite{XieMultiLevel}.  The binning and list decoding strategy was replaced by a joint typicality decoding scheme called sliding-window decoding, that first appeared for the “multiple-access channel with generalized feedback” in \cite{Carleial1982}.  Sliding-window (or regular) decoding, was a precursor for the DF schemes we consider.  A different scheme from regular decoding called “backward-decoding” was proposed for the “multiple-access channel with cribbing encoders” \cite{willems}.  

While backward-decoding and regular decoding achieve the same rates in the one-source multi-relay channel in \cite{XieMultiLevel}, backward-decoding achieves higher rates in general multi-source multi-relay channels \cite{XieKumar}.  However, backward-decoding requires much longer encoding delays.  To circumvent this delay problem a variation of regular coding called “offset encoding” was proposed for the multiple-access relay channel (MARC) in \cite{Sankar}.  Three different regular decoding schemes collectively achieve the backward-decoding region in \cite{Sankar} thus solving the delay problem in the MARC.  %One of implication of this paper: we solve the delay problem for general multi-source multi-relay channels as well.

Backward-decoding cannot support bidirectional communication in the DF  setting \cite{XieKumar}.  Inspired by \cite{Sankar} offset-encoding was applied in the two-way two-relay channel, where it was first discovered that the causality constraints can not always be simultaneously deactivated at all nodes \cite{Ponniah2008}. 

A parallel effort applying offset-encoding to CF schemes for multi-relay channels was proposed in \cite{Yassaee}, which recognized that regular decoding schemes correspond to layerings.  This work was studied in \cite{PonniahCompressForward2019} to show the mutual-information constraints are sufficient using the proof techniques here. 

Another independent line of inquiry revealed a relationship between CF and network coding \cite{AhlswedeCai2000}.  Noisy Network Coding (NNC) \cite{noisynetworkcoding} is a CF scheme that generalizes the network coding scheme in \cite{Avestimehr2011}.  It turns out that backward-decoding generalizes NNC \cite{Wu2013}\cite{HouKramer}.  The CF schemes in \cite{Yassaee}\cite{PonniahCompressForward2019} rely on regular coding; the long encoding delays and restrictions on bidirectional communication in joint backward-decoding DF-CF schemes \cite{Wu2014}\cite{Hou2016} can be avoided. 

The index-coding DF scheme we use was originally proposed in \cite{Xie2007} and simplified by Xiugang Wu in private correspondence with the authors. 

\section{Outline and Preliminaries}
\label{OutlineAndPreliminaries}
  We give a rough overview of the main results and proof, preparing for the rigor in Sections \ref{FlowDecomposition}, \ref{mainresult} and \ref{prooftheoremone}.  Let ${\cal N}=\{1,\ldots,|{\cal N}|\}$ denote the set of all nodes in the channel.  Assume all nodes are sources and each node is a destination for some subset of sources.  Let ${\bf y}_{\cal N}:=(y_{1},\ldots,y_{|{\cal N}|})$ and ${\bf x}_{\cal N}:=(x_{1},\ldots,x_{|{\cal N}|})$.  The input-output dynamics conform to the discrete memoryless channel:
\begin{align}
	\label{discretememorylesschannel}
	%\nonumber
	(\displaystyle\prod_{i\in {\cal N}}{\cal X}_{i},\hspace{1mm}p({\bf y}_{\cal N}|{\bf x}_{\cal N}),\displaystyle\prod_{i\in{\cal N}}{\cal Y}_{i}).
\end{align}

Let ${\bf f}(s)$ define the flow for each $s\in{\cal N}$ and let ${\bf F}:=\{{\bf f}(s):s\in{\cal N}\}$ be the flows for all nodes in the channel.  Fix some destination node $d\in{\cal N}$ rate-vector ${\bf R}=(R_{1},\ldots,R_{|{\cal N}|})$.  For any subset $S\subseteq{\cal N}$, let $R_{S}:=\sum_{s\in S}R_{s}$.  The mutual-information constraints implied by the one-relay channel are given by: 
\begin{align}
	\label{outline1}
	R_{S}<I(X_{F_{d}(S)};Y_{d}|X_{\tilde{F}_{d}(S)}),
\end{align}   
where $F_{d}(S)$ includes all the nodes upstream of node $d$, that encode the sources $S$, and $\tilde{F}_{d}(S)$ roughly speaking, includes all the nodes downstream of node $d$.  The constraint in (\ref{outline1}) applies for each $S\subseteq{\cal N}$.

 Let ${\bf L}_{d}$ be the layering at node $d$.  We show in Section \ref{FlowDecomposition}, that $({\bf F},{\bf L}_{d})$ corresponds to a well-defined encoding/decoding scheme (see Lemma \ref{splitting}); the messages decoded in each block of channel are encoded in the received sequences and the side information assumed in every typicality check is known or already decoded.  These consistency conditions are best expressed using so-called  ``virtual'' flows $\{{\bf v}(s):s\in{\cal N}\}$ ``seen'' by node $d$, from the actual flows in ${\bf F}$.  Virtual flows have some useful properties (see Lemma \ref{vssuffices}).  
 
A rate-vector is only achievable if the corresponding encoding/decoding scheme is causal (i.e., relays only forward messages they have already decoded).  These causality constraints correspond to the following set of inequalities for every source $s\in{\cal N}$ decoded by node $d$:
\begin{align}
	\label{outline2}
	\text{\sc layer}(v(s))<k_{v(s),d},
\end{align}
where $\text{\sc layer}(v(s))$ is the layer assigned to the virtual source $v(s)$ by ${\bf L}_{d}$ and $k_{v(s),d}$ is the encoding delay; the number of blocks that elapse between the virtual source $v(s)$ encoding a message from the actual source $s$, and node $d$ encoding the same message.  The constraint in (\ref{outline2}) applies to every source $s$ decoded by node $d$.   Both sets of constraints (\ref{outline1}) and (\ref{outline2}) are necessary for DF schemes.

We focus on two types of flow consistent with index-coding: hierarchical flow, where the routes are tree-structured, and all-cast flow, where each route covers all nodes.  The main result in Theorem \ref{theoremone} is that for an arbitrary ${\bf F}$ of either type, and any rate-vector ${\bf R}$ satisfying (\ref{outline1}), there exists an equivalent ${\bf F}^{\prime}$ and a layering ${\bf L}_{d}$ such that the encoding/decoding scheme $({\bf F}^{\prime}, {\bf L}_{d})$ achieves ${\bf R}$ and satisfies (\ref{outline2}).  Equivalent flows generate the same mutual information constraints in (\ref{outline1}).  Theorem \ref{theoremone}  shows that (\ref{outline1}) and (\ref{outline2}) are sufficient conditions for any rate-vector in the DF region.  However, there are infinitely many flows and layerings in a channel of fixed size, which creates ambiguity as to whether   (\ref{outline2}) is computable.

To prove Theorem \ref{theoremone} we pick an arbitrary rate-vector ${\bf R}$ that satisfies (\ref{outline1}) and an arbitrary ${\bf L}_{d}$.  If $({\bf F}, {\bf L}_{d})$ does not achieve ${\bf R}$, we define the following ``shift'' operation: 
\begin{align}
{\bf L}^{\prime}_{d}&=\text{\sc shift}({\bf L}_{d},S),
\end{align}
where $S$ is a selected subset of sources decoded by node $d$.  Lemma \ref{induction} shows that the achievable region generated by $({\bf F}, {\bf L}^{\prime}_{d})$ is closer to ${\bf R}$ than $({\bf F}, {\bf L}_{d})$.  The proof of Lemma \ref{induction} relies on Lemmas \ref{disjoint}, \ref{remain} and \ref{join}.

Next, we create a sequence of layerings $\{{\bf L}_{d,n}:n\in\mathbb{N}\}$, where ${\bf L}_{d,n+1}=\text{\sc shift}({\bf L}_{d,n}, S_{n})$ and $\{S_{n}:n\in\mathbb{N}\}$ is a selected sequence of subsets of ${\cal N}$.  Lemma \ref{grandefinale} proves there is some $n^{*}\in\mathbb{N}$, such that $({\bf F},{\bf L}_{d,n^{*}})$ achieves ${\bf R}$.  The proof of Lemma \ref{grandefinale} uses Lemma \ref{induction}.  

Finally, we construct an equivalent flow ${\bf F}^{\prime}$ and layering ${\bf L}_{d}$ that achieves ${\bf R}$ and satisfies (\ref{outline2}).  Lemmas \ref{fine} and \ref{normalize} show that $({\bf F}^{\prime},{\bf L}_{d})$ exists which completes the proof.

The proofs of each lemma are organized into claims and sub-claims, but   occasionally, the proof of one lemma references the claims of another.  Examples in Appendices \ref{Example}, \ref{appendix2}, and \ref{appendix3} explore the shift operation.

The following definition of typicality is used in the paper.  Let $X_{{\cal N}}:=\{X_{i}:i\in{\cal N}\}$ denote a finite collection of discrete random variables  with a fixed joint distribution $p(x_{{\cal N}})$ for some $x_{{\cal N}}:=\{x_{i}\in{\cal X}_{i}:i\in{\cal N}\}$.  Similarly, let ${\bf x}_{i}:=\{x^{(m)}_{i}\in{\cal X}_{i}:1\leq m\leq n\}$ denote an $n$-length vector of ${\cal X}_{i}$ and let ${\bf x}_{{\cal N}}:=\{{\bf x}_{i}:i\in{\cal N}\}$.  The set of typical $n$-sequences is given by: 
\begin{align}
	\nonumber
	&\hspace{-1mm}T^{(n)}_{\epsilon}(X_{{\cal N}}):=\\
	\nonumber
	&\hspace{7mm}\bigg\{{\bf x}_{{\cal N}}:\left|-\frac{1}{n}\log\text{Prob}({\bf x}_{S})-H(X_{S})\right|<\epsilon, \forall S\subseteq{\cal N}\bigg\},
\end{align}
where $\text{Prob}({\bf x}_{S}):=\prod^{n}_{m=1}p(x^{(m)}_{S})$.

%% file: fd.tex
\section{Flow Decomposition}
\label{FlowDecomposition}
%Let ${\cal N}=\{1,\ldots,|{\cal N}|\}$ denote the set of all nodes in the channel.  Assume that all nodes are sources and each node is a destination for some subset of sources.  Let ${\bf y}_{\cal N}:=(y_{1},\ldots,y_{|{\cal N}|})$ and ${\bf x}_{\cal N}:=(x_{1},\ldots,x_{|{\cal N}|})$.  The input-output dynamics conform to the discrete memoryless channel:
%\begin{align}
%	\label{discretememorylesschannel}
	%\nonumber
%	(\displaystyle\prod_{i\in {\cal N}}{\cal X}_{i},\hspace{1mm}p({\bf y}_{\cal N}|{\bf x}_{\cal N}),\displaystyle\prod_{i\in{\cal N}}{\cal Y}_{i}).
%\end{align}

%\begin{align}
%	\nonumber
%	\text{Prob}({\bf s})&=\text{Prob}({\bf z}_{i_{1}},{\bf z}_{i_{2}},\ldots,{\bf z}_{i_{l}})\\
%	\nonumber
%	&=\prod^{n}_{k=1}p(z_{i_{1},k},z_{i_{2},k},\ldots,z_{i_{l},k}).	
%\end{align}
%In the decode-forward setting, messages \textit{hop} from node to node.    %\textit{Encoding delays} are by definition, the difference between the blocks in which sources and relays transmit the same messages.     
%Two fundamentally different categories of decode-forward schemes exist: one based on correlated input distributions and super-position coding \cite{GamalCover}\cite{XieMultiLevel}, the other based on independent input distributions and index coding \cite{Xie2007}.  Both generate rate regions with similar structure and have their respective advantages, but neither is universally better than the other.  This paper only examines index-coding decode-forward schemes.

The \textit{flow} ${\bf f}(s)$ describes both the sequence of nodes along which messages from source $s$ are relayed through the network and the corresponding one-hop encoding delays.  %Messages either reach all of the network (in all-cast channels) or some of the network (in multi-cast channels).  
Formally, ${\bf f}(s):=N_{1}\xrightarrow{k_{1}}N_{2}\xrightarrow{k_{2}}\cdots\xrightarrow{k_{e_{s}-1}}N_{e_{s}}$ where: 

({F1}) $N_{1}:=\{s\}$ and $s\in{\cal N}$,

({F2}) $N_{l}\subseteq{\cal N}$ and $N_{l}\cap N_{q}=\{\null\}$ for all $1\leq l\neq q\leq e_{s}$, 

 ({F3}) $k_{l}\in\mathbb{N}$ is a one-hop encoding delay for $l=1,\ldots,e_{s}$,

(F4) $i\in{\bf f}(s)$ if and only if $i\in N_{l}$ for some $l=1,\ldots,e_{s}$,

({F5}) ${\bf f}(s,i):=N_{1}\xrightarrow{k_{1}}N_{2}\xrightarrow{k_{2}}\cdots\xrightarrow{k_{l-2}}N_{l-1}$ and \\
\text{\hspace{10mm}} $k_{s,i}:=\sum^{l-1}_{q=1}k_{q}$ if $i\in N_{l}$ for some $l=1,\ldots,e_{s}$.

%By definition, node $j$ is \textit{upstream} of node $i$ in ${\bf f}(s)$ if $k_{s,j}<k_{s,i}$ and \textit{downstream} of node $i$ if $k_{s,j}>k_{s,i}$.

In each \textit{block} of $n$ channel uses, sources generate new messages, and relays forward messages from previous blocks.  Transmission occurs over $B$ blocks.  The flows ${\bf F}=\{{\bf f}(s):s\in{\cal N}\}$ form a multi-edge directed graph on ${\cal N}$ that determines  the encoded messages in each block.  For a fixed rate vector ${\bf R}=(R_{1},\ldots,R_{|{\cal N}|})$,  the codebooks are generated as follows:
\begin{itemize}
	\item For each node $i\in{\cal N}$ independently generate  $2^{n\sum_{s:i\in{\bf f}(s)}R_{s}}$ $n$-length codewords ${\bf x}_{i}(w)$ according to the distribution $p(x)$ over $x\in{\cal X}_{i}$ where $w\in\{1,\ldots,2^{n\sum_{s:i\in{\bf f}(s)}R_{s}}\}$.
\end{itemize}
The encoding proceeds as follows for each $b\in\{1,\ldots,B\}$:
\begin{itemize}
	\item In block $b$, source $s$ generates the message $m_{s}(b)\in\{1,\ldots,2^{nR_{s}}\}$
	\item In block $b$, node $i$ transmits the codeword ${\bf x}_{i}(w(b))$ where $w(b)\in\{1,\ldots,2^{n(\sum_{s:i\in{\bf f}(s)}R_{s})}\}$ is the index assigned to the message vector ${\bf w}(b):=(w_{1},\ldots,w_{|{\cal N}|})$ and:
\begin{align}
\label{dubenc}
w_{s}=\begin{cases}m_{s}(b-k_{s,i})&i\in {\bf f}(s)\\ ``1" & i\notin {\bf f}(s)\text{ or }b-k_{s,i}\leq 0\end{cases}	
\end{align}
\end{itemize}
%In block $b$, source $s$ generates the message $m(b)\in\{1,\ldots,2^{nR_{s}}\}$ and node $i\in{\cal Z}$ sends the index $w_{i}(b)\in\{1,\ldots,2^{n(\sum_{s\in{\cal  S}_{i}}R_{s})}\}$ assigned to the message vector 
%\begin{align}
%\label{doubleudb}
%\bar{w}_{d}(b):=\{\langle s,m(b-k_{s,d})\rangle: s\in{\cal S}_{d}\}	.
%\end{align}

%The codebook at node $i\in{\cal Z}$ consists of $2^{n\sum_{s\in{\cal S}_{i}}R_{s}}$ $n$-length codewords independently generated according to the distribution $p(x)$ over $x\in{\cal X}_{i}$.  The index $w\in\{1,\ldots,2^{n\sum_{s\in{\cal S}_{i}}R_{s}}\}$ maps to a unique codeword ${\bar x}_{i}(w)$.  In block $b$, node $i$ transmits the codeword ${\bar x}_{i}(w_{i}(b))$.

In each block, the relays decode source messages.  To simplify the analysis,   assume each node encodes the same set of sources that it decodes, and   define:
\begin{align}
	\label{sd}
	{\cal S}(d)&:=\{s:d\in{\bf f}(s)\}.
\end{align}  
Observe that ${\cal S}(d)$ denotes the set of sources encoded (and thus decoded) by node $d\in{\cal N}$.  

For any subset $S\subseteq{\cal N}$, let $F_{d}(S):=\{i\in{\bf f}(s,d):S\cap{\cal S}(d)\}$.  The messages decoded by node $d$ in each block correspond to ordered partitions or  \textit{layerings} of $F_{d}({\cal N})$, all the nodes upstream of (or preceding) node $d$.  A layering ${\bf L}_{d}:=(L_{0},\ldots,L_{|{\bf L}_{d}|-1})$ satisfies the following conditions: 
 
({L1})  $L_{l}\subseteq F_{d}({\cal N})$ for every $l=0,\ldots,|{\bf L}_{d}|-1$,

({L2}) $L_{l}\cap L_{q}=\{\}$ for $l\neq q$,

({L3}) $F_{d}({\cal N})=\cup^{|{\bf L}_{d}|-1}_{l=0}L_{l}$,

({L4}) $L_{|{\bf L}_{d}|-1}\neq\{\}$,

({L5}) $\text{\sc layer}(i):=l$ by definition if $i\in L_{l}$.  

Each layer corresponds to previous block of channel uses.  %A flow decomposition at node $d$ denoted by $({\bf F},{\bf L}_{d})$ is a flow set and a layering that in each block  defines, respectively, the messages encoded at all nodes and the messages decoded by node $d$.  
It is convenient to associate with $({\bf F},{\bf L}_{d})$ a \textit{virtual source} $v(s)$ and a \textit{virtual flow} ${\bf v}(s)$ for each $s\in{\cal S}(d)$, where ${\bf v}(s)$ is a subsequence of ${\bf f}(s)$.  Virtual sources and flows are the sources and flows ``seen'' by node $d$ given the layering ${\bf L}_{d}$.    Different layerings change how node $d$ decodes messages from nodes in $F_{d}({\cal N})$.

For each $s\in{\cal S}(d)$, define the function $u(s,i):=\text{\sc layer}(i)+k_{s,i}$ and the node subset $M(s):=\{\arg\min_{i\in{\bf f}(s,d)}u(s,i)\}$.  The virtual source $v(s)$ corresponding to $s$ is defined as:
\begin{align}
\label{vs}
v(s)&:=\{\arg\min_{i\in M(s)}k_{s,i}\}.  	
\end{align}
Suppose $d\in N_{e_{s,d}+1}$ for some $1\leq e_{s,d}\leq e_{s}$ and, per (F5), ${\bf f}(s,d):=N_{1}\xrightarrow{k_{1}}N_{2}\xrightarrow{k_{2}}\cdots\xrightarrow{k_{e_{s,d}-1}}N_{e_{s,d}}$.  By construction $v(s)\subseteq N_{l}$ for some $l=1,\ldots,e_{s,d}$.  The virtual flow ${\bf v}(s):= N^{\prime}_{1}\xrightarrow{k^{\prime}_{1}} N^{\prime}_{2}\xrightarrow{k^{\prime}_{2}}\cdots N^{\prime}_{p}$ is the subsequence of ${\bf f}(s,d)$ that satisfies the following conditions: 

({V1}) $N^{\prime}_{1}:=\{v(s)\}$,

({V2}) $N^{\prime}_{l}\subseteq N_{z_{l}}$ where $1\leq l\leq p$ and $z_{1}, z_{2},\ldots, z_{p}$ is a 

\hspace{7.5mm}subsequence of $1, 2,\ldots, e_{s,d}$, 

({V3}) $k^{\prime}_{l}=\sum^{z_{l+1}-1}_{q=z_{l}}k_{q}$,

({V4})  $i\in{\bf v}(s)$ implies $i\in N^{\prime}_{l}$ for some $1\leq l\leq p$, 

({V5}) $\text{\sc layer}(v(s))-\text{\sc layer}(i)=k_{v(s),i}$ for all  $i\in{\bf v}(s)$ 

\hspace{8mm}where $k_{v(s),i}:=k_{s,i}-k_{s,v(s)}$ as per (F5).

There is a convenient way of checking whether or not a particular node is in fact the virtual source of a given flow.

\input{"clocksecurity.tex"}

%and  
%\begin{align}
%\label{virtualflow}
%\nonumber
%\text{\sc layer}(v(s))-\text{\sc layer}(i)=k_{v(s),i},
%\end{align}

%\hspace{8mm}where $k_{v(s),i}:=k_{s,i}-k_{s,v(s)}$ as per (F5).

In block $b$, node $d$ decodes the message vector ${\bf m}(b):=(m_{1},\ldots,m_{|{\cal N}|})$ where
\begin{align}
\label{emdb}
m_{s}:=\begin{cases}
m_{s}(b-k_{s,v(s)}-\text{\sc layer}(v(s)))& s\in{\cal S}(d)\\
``1" & s\notin {\cal S}(d),\end{cases}
\end{align}
and $m_{s}:= ``1"$ if $b-k_{s,v(s)}-\text{\sc layer}(v(s))\leq0$.  Define: 
\begin{align}
\label{id}
I(d)&:=\{i:{\cal S}(i)\subseteq{\cal S}(d)\}.	
\end{align}
Observe that (\ref{id}) denotes the set of relays that only encode sources decoded by node $d$.  For every $S\subseteq{\cal N}$ and $0\leq l\leq|{\bf L}_{d}|-1$, let: 
\begin{align}
\label{A}
A_{l}(S)&:=\{i\in{\bf v}(s):s\in S\cap{\cal S}(d)\}\cap L_{l},\\
\label{Atilde}
	\tilde{A}_{l}(S)&:=(\cup_{i\in I(d)}F_{i}({\cal N}))\setminus(\cup^{|{\bf L}_{d}|-1}_{q=l+1}L_{q}\cup A_{l}(S)).
\end{align}

To decode ${\bf m}(b)$, node $d$ finds the message vector ${\bf\hat{m}}(b):=(\hat{m}_{1},\ldots,\hat{m}_{|{\cal N}|})$ that satisfies the following typicality checks for $0\leq l\leq |{\bf L}_{d}|-1$:
\begin{align}
\nonumber
	&\hspace{-3.2mm}(\{{\bf x}_{i}(\hat{w}(b-l)):i\in A_{l}({\cal N})\}, \{{\bf X}_{i}(b-l):i\in\tilde{A}_{l}({\cal N})\},\hspace{4.5mm}\\
\label{tippie}
	&\hspace{21mm}{\bf Y}_{d}(b-l))\in T^{(n)}_{\epsilon}(X_{A_{l}({\cal N})\cup\tilde{A}_{l}({\cal N})},Y_{d})
\end{align}
where ${\bf X}_{i}(b-l)$ is the sequence sent by node $i$ in block $b-l$ and ${\bf Y}_{d}(b-l)$ is the sequence received by node $d$ in block $b-l$.  For all $i\in A_{l}({\cal N})$ and ${\bf x}_{i}(\hat{w}(b-l))$, $\hat{w}(b-l)$ is the index that maps to ${\bf\hat{w}}(b-l):=(\hat{w}_{1},\ldots,\hat{w}_{|{\cal N}|})$ where:
\begin{align}
\label{dubdb}
\hat{w}_{s}:=\begin{cases}
\hat{m}_{s}& i\in{\bf v}(s)\\
m_{s}(b-l-k_{s,i})&i\notin{\bf v}(s)\text{ and }i\in{\bf f}(s)\\ 
``1" & i\notin {\bf f}(s).\end{cases}	
\end{align}

The index $\hat{m}_{s}$ is the message to be decoded in block $b$, that is, $m_{s}(b-k_{s,v(s)}-\text{\sc layer}(v(s)))$ as given in (\ref{emdb}).  The message actually encoded in $w(b-l)$ is $m_{s}(b-l-k_{s,i})$ as given in (\ref{dubenc}).  These messages must be the same in order for the scheme to work.  If $i\notin{\bf v}(s)$ but $i\in{\bf f}(s)$ then $m_{s}(b-l-k_{s,i})$ must be already known to node $d$ (decoded in some previous block).  If $i\notin{\bf f}(s)$ then node $i$ encodes ``1'' in place of messages from source $s$ as per (\ref{dubenc}).  Observe that (\ref{A}) and (\ref{Atilde}) imply:
\begin{align}
\label{AN}
A_{l}({\cal N})&=\{i\in{\bf v}(s):s\in {\cal S}(d)\}\cap L_{l},\\
\label{AtildeN}
	\tilde{A}_{l}({\cal N})&=(\cup_{i\in I(d)}F_{i}({\cal N}))\setminus(\cup^{|{\bf L}_{d}|-1}_{q=l+1}L_{q}\cup A_{l}({\cal N})).
\end{align}

There are five conditions that must be satisfied in order for the typicality checks to be feasible.  The first follows from causality and the rest are  index-coding primitives.  

An encoding/decoding scheme is causal if in block $b$, each relay only encodes source messages it decoded in previous blocks.  The messages encoded and decoded by node $d$ in block $b$ are defined by (\ref{dubenc}) and (\ref{emdb}) respectively.

(C1) For every $s\in{\cal S}(d)$:
\begin{align}
\label{causal}
	b-k_{s,d}<b-k_{s,v(s)}-\text{\sc layer}(v(s)).
\end{align}
It is convenient to simplify (\ref{causal}):
\begin{align}
\label{causal2}
	\text{\sc layer}(v(s))<k_{v(s),d},
\end{align}
which yields the causality constraints in (\ref{outline2}).  The next condition pertains to the flow itself, and states that node $d$ must decode or know all sources encoded in node $i$ to decode any one source encoded in node $i$.   

(C2) For every $s\in{\cal S}(d)$ and $i\in {\bf f}(s,d)$:
\begin{align}
\label{C1}
{\cal S}(i)\subseteq {\cal S}(d).
\end{align}
 
Two flow types of interest that satisfy (C2) are all-cast channels, in which all nodes decode all sources, and multi-cast channels with hierarchical (or tree structured) flow, in which messages proceed from the branches to the root.  

The remaining conditions ensure that (\ref{dubenc})-(\ref{AtildeN}) are consistent.  Comparing (\ref{dubenc}), (\ref{emdb}), (\ref{tippie}) and (\ref{dubdb}), the messages decoded by node $d$ in block 
$b$ must match the corresponding messages encoded by node $i$ in block $b-l$.

(C3) Suppose $i\in{\bf f}(s)$ and $i\in A_{l}({\cal N})$ for some $l\in\{0,\ldots,|{\bf L}_{d}|-1\}$.  If $i\in{\bf v}(s)$ then:
\begin{align}
\label{C2}
	b-l-k_{s,i} = b-k_{s,v(s)}-\text{\sc layer}(v(s)).
\end{align}

Again comparing (\ref{dubenc}), (\ref{emdb}), (\ref{tippie}) and (\ref{dubdb}), any messages not decoded by node $d$ in block $b$ but encoded by node $i$ in block $b-l$, must already be known by node $d$ in some previous block.

(C4) Suppose $i\in{\bf f}(s)$ and $i\in L_{l}$ for some $l\in\{0,\ldots,|{\bf L}_{d}|-1\}$.  If $i\notin{\bf v}(s)$ then:
\begin{align}
\label{C3}
	b-l-k_{s,i} < b-k_{s,v(s)}-\text{\sc layer}(v(s)).
\end{align}

Finally, comparing (\ref{dubenc}), (\ref{emdb}) and (\ref{tippie}), the sequences in $\{{\bf X}_{i}(b-l):i\in\tilde{A}_{l}({\cal N})\}$, which are side information in (\ref{tippie}), must only encode messages known to node $d$.

(C5) Suppose $i\in{\bf f}(s)$ and $i\in\tilde{A}_{l}({\cal N})$ for some $l\in\{0,\ldots,|{\bf L}_{d}|-1\}$.  Then:
\begin{align}
	b-l-k_{s,i} < b-k_{s,v(s)}-\text{\sc layer}(v(s)).
\end{align}

%The previous typicality checks are valid if it can verified for all $0\leq l\leq |\bar{L}_{d}|-1$ and $i\in F_{d}({\cal S})$ that $\bar{m}_{d}(b)$ splits $\bar{w}_{i}(b-l)$ into a set of messages that node $d$ decodes in block $b$ and a set of messages that node $d$ knows from its side information.  Let $A_{l}:=\{i:\bar{w}_{i}(b-l)\cap\bar{m}_{d}(b)\neq\{\}, \bar{w}_{i}(b-l)\subseteq\cup^{b}_{q=1}\bar{m}_{d}(q)\}$ and $\tilde{A}_{l}:=\{i:\bar{w}_{i}(b-l)\subseteq\cup^{b-1}_{q=1}\bar{m}_{d}(q)\}$.  Equivalently, $\tilde{A}_{l}:=\{i:\bar{w}_{i}(b-l)\subseteq\cup^{b}_{q=1}\bar{m}_{d}(q)\}\setminus A_{l}$.  If $L_{d,l}\subseteq A_{l}\cup\tilde{A}_{l}$ for every $0\leq l\leq |\bar{L}_{d}|-1$, then $\bar{m}_{d}(b)$ is a \textit{splitting vector}.

\begin{figure*}[t]
        \center{\includegraphics[width=\textwidth]
        {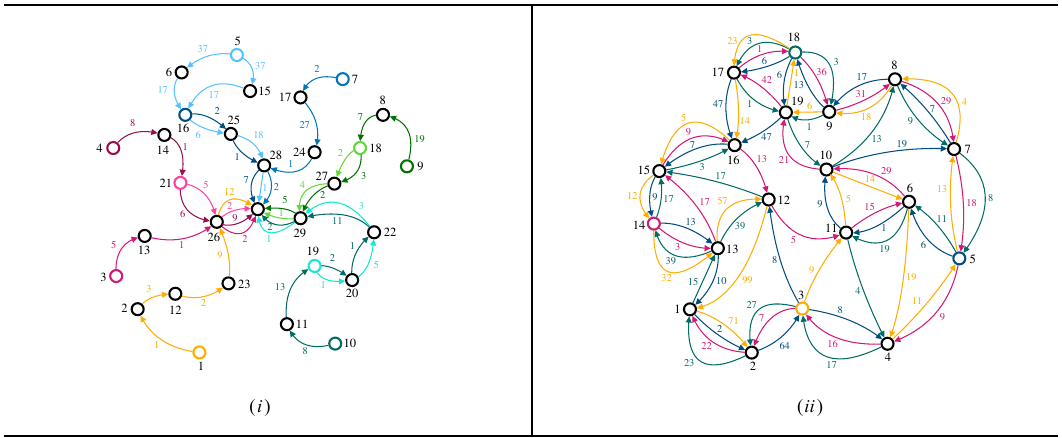}}
        \caption{(i) Hierarchical flows ${\bf F}=\{{\bf f}(1),{\bf f}(3),{\bf f}(4),{\bf f}(5),{\bf f}(7),{\bf f}(9),{\bf f}(10),{\bf f}(16),{\bf f}(18),{\bf f}(19), {\bf f}(21)\}$ where ${\bf f}(1)=1\xrightarrow{1}2\xrightarrow{3}12\xrightarrow{2}23\xrightarrow{9}26\xrightarrow{12}30$, ${\bf f}(3)=3\xrightarrow{5}13\xrightarrow{1}26\xrightarrow{2}30$, ${\bf f}(4)=4\xrightarrow{8}14\xrightarrow{1}21\xrightarrow{6}26\xrightarrow{9}30$, ${\bf f}(5)=5\xrightarrow{37}\{6,15\}\xrightarrow{17}16\xrightarrow{6}25\xrightarrow{18}28\xrightarrow{1}30$, ${\bf f}(7)=7\xrightarrow{2}17\xrightarrow{27}24\xrightarrow{1}28\xrightarrow{2}30$, ${\bf f}(9)=9\xrightarrow{19}8\xrightarrow{7}18\xrightarrow{3}27\xrightarrow{2}29\xrightarrow{5}30$, ${\bf f}(10)=10\xrightarrow{8}11\xrightarrow{13}19\xrightarrow{2}20\xrightarrow{1}22\xrightarrow{11}29\xrightarrow{2}30$, ${\bf f}(16)=16\xrightarrow{2}25\xrightarrow{1}28\xrightarrow{7}30$, ${\bf f}(18)=18\xrightarrow{2}27\xrightarrow{4}29\xrightarrow{1}30$, ${\bf f}(19)=19\xrightarrow{1}20\xrightarrow{5}22\xrightarrow{3}29\xrightarrow{1}30$, and ${\bf f}(21)=21\xrightarrow{5}26\xrightarrow{2}30$. (ii) All-cast flows ${\bf F}=\{{\bf f}(3),{\bf f}(5),{\bf f}(14),{\bf f}(18)\}$ where ${\bf f}(3)=9\xrightarrow{9}11\xrightarrow{5}10\xrightarrow{14}6\xrightarrow{19}4\xrightarrow{11}5\xrightarrow{13}7\xrightarrow{4}8\xrightarrow{18}9\xrightarrow{6}19\xrightarrow{1}18\xrightarrow{23}17\xrightarrow{14}16\xrightarrow{5}15\xrightarrow{12}14\xrightarrow{32}13\xrightarrow{57}12\xrightarrow{99}1\xrightarrow{71}2$, ${\bf f}(5)=5\xrightarrow{6}6\xrightarrow{1}11\xrightarrow{9}10\xrightarrow{19}7\xrightarrow{7}8\xrightarrow{17}9\xrightarrow{13}18\xrightarrow{6}19\xrightarrow{47}16\xrightarrow{7}15\xrightarrow{9}14\xrightarrow{13}13\xrightarrow{10}1\xrightarrow{2}2\xrightarrow{64}3\xrightarrow{8}\{4,12\}$, ${\bf f}(14)=14\xrightarrow{3}13\xrightarrow{17}15\xrightarrow{9}16\xrightarrow{13}12\xrightarrow{5}11\xrightarrow{15}6\xrightarrow{29}10\xrightarrow{21}19\xrightarrow{42}17\xrightarrow{1}18\xrightarrow{36}9\xrightarrow{31}8\xrightarrow{29}7\xrightarrow{18}5\xrightarrow{9}4\xrightarrow{16}3\xrightarrow{7}2\xrightarrow{22}1$, ${\bf f}(18)=18\xrightarrow{3}\{9,17\}\xrightarrow{1}19\xrightarrow{7}10\xrightarrow{13}8\xrightarrow{9}7\xrightarrow{8}5\xrightarrow{11}6\xrightarrow{19}11\xrightarrow{4}4\xrightarrow{17}3\xrightarrow{27}2\xrightarrow{23}1\xrightarrow{15}13\xrightarrow{39}14\xrightarrow{17}15\xrightarrow{3}16$.}
        \label{hierarchicalallcast}
 \end{figure*}

\input{"splitting.tex"}

  For any rate vector ${\bf R}:=(R_{1},\ldots,R_{|{\cal N}|})$ and some subset $S\subseteq{\cal S}(d)$, let $R_{S}=\sum_{s\in S} R_{s}$.  Lemma \ref{splitting} implies the encoding/decoding scheme defined by (\ref{dubenc})-(\ref{AtildeN}) is well-defined contingent on (C1) and (C2).  For the product distribution $\{\Pi^{|{\cal N}|}_{i=1}p(x_{i}):x_{i}\in{\cal X}_{i},i=1,\ldots,|{\cal N}|\}$, it follows from (\ref{tippie}) that the probability node $d$ decodes some subset of sources $S\subseteq{\cal S}(d)$ incorrectly, goes to zero if the following constraint is satisfied:

\begin{align}
\label{flowdecomp}
R_{S}<\displaystyle\sum^{|{\bf L}_{d}|-1}_{l=0}I(X_{A_{l}(S)};Y_{d}|X_{\tilde{A}_{l}(S)}).
\end{align}

Let ${\cal R}({\bf F},{\bf L}_{d})$ denote the region of rate vectors that satisfy (\ref{flowdecomp}) for every $S\subseteq{\cal S}(d)$.

\section{Main Result}
\label{mainresult}
For any subset $S\subseteq{\cal N}$, recall that $F_{d}(S):=\{i\in{\bf f}(s,d):s\in S\cap{\cal S}(d)\}$ and let $\tilde{F}_{d}(S):=(\cup_{i\in I(d)}F_{i}({\cal N}))\setminus F_{d}(S)$.  For the product distribution $\{\Pi^{|{\cal N}|}_{i=1}p(x_{i}):x_{i}\in{\cal X}_{i},i=1,\ldots,|{\cal N}|\}$, let ${\cal R}_{d}({\bf F})$ denote the set of rate vectors that satisfy the following constraint for all subsets $S\subseteq{\cal S}(d)$:
\begin{align}
\label{biggiewiggie}
	R_{S}<I(X_{F_{d}(S)};Y_{d}|X_{\tilde{F}_{d}(S)}).
\end{align}   
Since $F_{d}(S)$ includes all the nodes upstream of node $d$, that encode the sources $S$, and $\tilde{F}_{d}(S)$ includes all the nodes downstream of node $d$ as well as the remaining nodes in $F_{d}({\cal N})$, it follows that (\ref{biggiewiggie}) is is an upper-bound on the decode-forward achievable sum rate at node $d$.
%, not only for the class of regular decode-forward schemes defined by (\ref{dubenc})-(\ref{AtildeN}), but for all decode-forward schemes in general, including those based on backward-decoding.  An important caveat is that (\ref{biggiewiggie}), as an outer-bound, applies to all input distributions $p(x_{1},\ldots,x_{|{\cal N}|})$ not just the product input distribution used in (\ref{dubenc})-(\ref{AtildeN}).

Two flows ${\bf F}$ and ${\bf F}^{\prime}$ are \textit{equivalent}, by definition, if for every $s\in{\cal N}$, ${\bf f}(s)\in{\bf F}$ and ${\bf f}^{\prime}(s)\in{\bf F}^{\prime}$ differ only in their encoding delays, $\{k_{i}:1\leq i\leq e_{s}\}$ and $\{k^{\prime}_{i}: 1\leq i\leq e_{s}\}$ respectively, but not the node sequences, so that  $N_{i}=N^{\prime}_{i}$ for all $1\leq i\leq e_{s}$ where ${\bf f}(s):=N_{1}\xrightarrow{k_{1}}N_{2}\xrightarrow{k_{2}}\cdots\xrightarrow{k_{e_{s}-1}}N_{e_{s}}$ and ${\bf f}^{\prime}(s):=N^{\prime}_{1}\xrightarrow{k^{\prime}_{1}}N^{\prime}_{2}\xrightarrow{k^{\prime}_{2}}\cdots\xrightarrow{k^{\prime}_{e_{s}-1}}N^{\prime}_{e_{s}}$.  Since (\ref{biggiewiggie}) depends only the nodes in the flows and not the encoding delays, it follows that ${\cal R}_{d}({\bf F})={\cal R}_{d}({\bf F}^{\prime})$ for all $d\in{\cal N}$ and channel statistics $p({\bf y}_{{\cal N}}|{\bf x}_{{\cal N}})$ if ${\bf F}$ and ${\bf F}^{\prime}$ are equivalent.
\begin{theorem}
\label{theoremone}
Fix some $d\in{\cal N}$.  For any ${\bf F}$ satisfying (C2) and ${\bf R}\in{\cal R}_{d}({\bf F})$, there exists an equivalent ${\bf F}^{\prime}$ and ${\bf L}_{d}$ satisfying (C1) and (C2) such that 
${\bf R}\in{\cal R}({\bf F}^{\prime},{\bf L}_{d})$.
\end{theorem}
\begin{proof}
	See Section \ref{prooftheoremone}.
\end{proof}
%Theorem \ref{theoremone} partially confirms the hypothesis that regular-coding schemes collectively achieve the same region as backward-decoding in the decode-forward setting.
  %The independent input distributions in (\ref{flowdecomp}) and (\ref{biggiewiggie}) are a feature of index-coding.  A fundamentally different scheme called super-position coding creates correlated input distributions and is compatible with regular-coding \cite{Sankar}.  
  %Correlated input distributions introduce auxiliary conditional random variables into the mutual information in (\ref{biggiewiggie}) (see \cite{Sankar} and \cite{XieKumar}), and are unambiguously advantageous in broadcast channels unconstrained by (C2).

If ${\bf R}\in\cap^{|{\cal N}|}_{d=1}{\cal R}_{d}({\bf F})$, Theorem \ref{theoremone} does \textit{not} imply that ${\bf R}\in{\cal R}({\bf F}^{\prime},{\bf L}_{d})$ for each $1\leq d\leq |{\cal N}|$ and some ${\bf L}_{d}$ \textit{simultaneously}.  That result only occurs if the flow is hierarchical.  %All-cast channels are more complicated and require algorithmic methods to verify (C1) simultaneously for all $d\in{\cal N}$. 

%% file: clocksecurity.tex
% Given ${\bf f}(s,d):=N_{1}\xrightarrow{k_{1}}N_{2}\xrightarrow{k_{2}}\cdots\xrightarrow{k_{e-1}}N_{e}$, suppose $v(s)\subseteq N_{l}$ for some $1\leq l\leq e$.  Let ${\bf f}(v(s),d):=v(s)\xrightarrow{k_{l}}N_{l+1}\xrightarrow{k_{l+1}}\cdots\xrightarrow{k_{e-1}}N_{e}$.  For $i\in{\bf f}(s,v(s))$, let $k_{i,v(s)}:=k_{s,v(s)}-k_{s,i}$.  It can be shown that:
 
 \begin{lemma}
 \label{vssuffices}
The $v(s)$ in (\ref{vs}) uniquely satisfies:
\begin{align}
\label{firstIndexCoding}
\text{\sc layer}(v(s))-\text{\sc layer}(i)&\leq k_{v(s),i}\hspace{2mm}\forall i\in{\bf f}(v(s),d),\\
\label{secondIndexCoding}
\text{\sc layer}(i)-\text{\sc layer}(v(s))&> k_{i,v(s)}	\hspace{2mm}\forall i\in{\bf f}(s,v(s)),
\end{align}	
where $k_{v(s),i}:=k_{s,i}-k_{s,v(s)}$ for all $i\in{\bf f}(v(s),d)$ and  $k_{i,v(s)}:=k_{s,v(s)}-k_{s,i}$ for all $i\in{\bf f}(s,v(s))$. 
\end{lemma}
\begin{proof}
By construction $v(s)\subseteq M(s)$, which minimizes $u(s,i):=\text{\sc layer}(i)+k_{s,i}$ for all $i\in{\bf f}(s)$.  Therefore $\text{\sc layer}(v(s))+k_{s,v(s)}\leq\text{\sc layer}(i)+k_{s,i}$ which implies $\text{\sc layer}(v(s))-\text{\sc layer}(i)\leq k_{s,i}-k_{s,v(s)}$.  By definition $k_{v(s),i}:=k_{s,i}-k_{s,v(s)}$ for all $i\in{\bf f}(v(s),d)$ which yields (\ref{firstIndexCoding}).  
	
Furthermore (\ref{vs}) implies $v(s)$ is the subset of $M(s)$ closest to source  $s$ by hop count.  Therefore $\text{\sc layer}(v(s))+k_{s,v(s)}<\text{\sc layer}(i)+k_{s,i}$ if $i\in{\bf f}(s,v(s))$, which implies $\text{\sc layer}(i)-\text{\sc layer}(v(s))>k_{s,v(s)}-k_{s,i}$. By definition $k_{i,v(s)}:=k_{s,v(s)}-k_{s,i}$ for all $i\in{\bf f}(s,v(s))$ which yields (\ref{secondIndexCoding}).
\end{proof}

%% file: splitting.tex
\begin{lemma} 
\label{splitting}
 If (C1) and (C2) are satisfied then (C3), (C4), and (C5) are also satisfied.
\end{lemma}
\begin{proof}
Given the assumptions $i\in{\bf f}(s)$, $i\in A_{l}({\cal N})$, and $i\in{\bf v}(s)$ in (C3), consider the following sequence of inequalities:
\begin{align}
	\label{C2proof1}
	b-l-k_{s,i}&=b-\text{\sc layer}(i)-k_{s,i},\\
	\label{C2proof2}
	&=b-(\text{\sc layer}(v(s))-k_{v(s),i})-k_{s,i},\\
	\nonumber
	&=b-\text{\sc layer}(v(s))-(k_{s,i}-k_{v(s),i}),\\
	\label{C2proof3}
	&=b-\text{\sc layer}(v(s))-k_{s,v(s)},
\end{align}
where (\ref{C2proof1}) follows because (\ref{AN}) implies that $i\in A_{l}({\cal N})$ only if $i\in L_{l}$ which implies $\text{\sc layer}(i)=l$, and (\ref{C2proof2}) follows because (V5) implies $i\in{\bf v}(s)$ only if $\text{\sc layer}(v(s))-\text{\sc layer}(i)=k_{v(s),i}$.  By inspection, (\ref{C2proof3}) satisfies (\ref{C2}) in (C3).

Consider assumptions $i\in{\bf f}(s)$ and $i\in L_{l}$ in (C4).  If (C2) is satisfied, then (\ref{C1}) implies $s\in{\cal S}(d)$ which implies that $v(s)$ and ${\bf v}(s)$ exist at node $d$, as defined by (\ref{vs}) and (V1)-(V5) respectively.  Assume $i\notin{\bf v}(s)$ as stated in (C4).  
Suppose $i\in{\bf f}(v(s),d)$ and consider the following inequalities:
\begin{align}
	\label{feasibility1}
	b-l-k_{s,i}&\leq b-(\text{\sc layer}(v(s))-k_{v(s),i})-k_{s,i},\\
	\nonumber
	&=b-(k_{s,i}-k_{v(s),i})-\text{\sc layer}(v(s)),\\
	\label{feasibility12}
	&=b-k_{s,v(s)}-\text{\sc layer}(v(s)),
\end{align} 
where (\ref{feasibility1}) follows from (\ref{firstIndexCoding}) 
and because (L5) implies $\text{\sc layer}(i)=l$ if $i\in L_{l}$.   Equality in (\ref{feasibility12}) only occurs if property (V5) holds which implies $i\in{\bf v}(s)$ as in (\ref{C2proof3}).  Since $i\notin{\bf v}(s)$, it follows from (\ref{C2proof3}) and (\ref{feasibility12}) that:
\begin{align}
\label{feasibility13}
b-l-k_{s,i}<b-k_{s,v(s)}-\text{\sc layer}(v(s))
\end{align}   
Now suppose $i\in{\bf f}(s,v(s))$.  Note that ${\bf f}(s,v(s))$ precedes $v(s)$ so $i\notin{\bf v}(s)$.  Consider the following inequalities: 
\begin{align}
	\label{feasibility2}
	b-l-k_{s,i}&<b-(k_{i,v(s)}+\text{\sc layer}(v(s)))-k_{s,i},\\
	\nonumber
	&=b-(k_{s,i}+k_{i,v(s)})-\text{\sc layer}(v(s)),\\
	\label{feasibility22}
	&=b-k_{s,v(s)}-\text{\sc layer}(v(s)),
\end{align} 
where (\ref{feasibility2}) follows from (\ref{secondIndexCoding}) and because (L5) implies $\text{\sc layer}(i)=l$ if $i\in L_{l}$ (L5).  Since $i\notin{\bf v}(s)$, it follows from (\ref{feasibility13}) and (\ref{feasibility22}) that: 
\begin{align}
\label{C3proof1}
b-l-k_{s,i}<b-k_{s,v(s)}-\text{\sc layer}(v(s)),  	
\end{align}
which satisfies (\ref{C3}) in (C4).

Consider assumptions $i\in{\bf f}(s)$ and $i\in\tilde{A}_{l}({\cal N})$ in (C5).  It follows from (\ref{AtildeN}) that $\tilde{A}_{l}({\cal N})=((\cup_{j\in I(d)}F_{j}({\cal N}))\setminus F_{d}({\cal N}))\cup(F_{d}({\cal N})\setminus(\cup^{|{\bf L}_{d}|-1}_{q=l+1}L_{q}\cup A_{l}({\cal N})))$.  First, suppose $i\in(F_{d}({\cal N})\setminus(\cup^{|{\bf L}_{d}|-1}_{q=l+1}L_{q}\cup A_{l}({\cal N})))$.  It follows from property (L3) and (\ref{AN}) that $(F_{d}({\cal N})\setminus(\cup^{|{\bf L}_{d}|-1}_{q=l+1}L_{q}\cup A_{l}({\cal N})))=(L_{l}\setminus A_{l}({\cal N}))\cup(\cup^{l-1}_{q=0}L_{q})$.  If $i\in L_{l}\setminus A_{l}({\cal N})$ then (\ref{AN}) implies $i\notin{\bf v}(s)$ for every  $s\in{\cal S}(d)$.  It follows from (\ref{C3proof1}) that (C5) is satisfied.  If $i\in\cup^{l-1}_{q=0}L_{q}$ then $b-q-k_{s,i}\leq b-k_{s,v(s)}-\text{\sc layer}(v(s))$ for some $0\leq q\leq l-1$ and $\text{sc layer}(i)=q$, as implied by (\ref{feasibility12}) when $i\in{\bf f}(v(s),d)$ and (\ref{feasibility22}) when $i\in{\bf f}(s,v(s))$.  Therefore $b-l-k_{s,i}<b-k_{s,v(s)}-\text{\sc layer}(v(s))$ which satisfies (C5).  Second, suppose $i\in(\cup_{j\in I(d)}F_{j}({\cal N}))\setminus F_{d}({\cal N})$ and consider the following inequalities:
\begin{align}
	\label{C5proof1}
	b-k_{s,v(s)}-\text{\sc layer}(v(s))&>b-k_{s,d},\\
	\label{C5proof2}
	&>b-k_{s,i},
\end{align}
where (\ref{C5proof1}) follows from (C1).  By definition in (\ref{id}), $I(d):=\{j:{\cal S}(j)\subseteq{\cal S}(d)\}$, so (\ref{sd}) implies node $i$ only encodes sources decoded by node $d$.  Since $i\notin F_{d}({\cal N})$, node $i$ is ``downstream'' of node $d$ so that $k_{s,i}>k_{s,d}$ which implies (\ref{C5proof2}).  It follows from (\ref{C5proof2}) that $b-k_{s,v(s)}-\text{\sc layer}(v(s))>b-l-k_{s,i}$ which satisfies (C5).
\end{proof}

%% file: disjoint.tex
\begin{lemma}
	\label{disjoint}
	For every $S\subseteq{\cal S}(d)$ and $l\in\{0,\ldots,|{\bf L}_{d}|-1\}$, $A_{l}(S)\setminus A_{l}(U)\subseteq A^{\prime}_{l}(S)$.
\end{lemma}
\begin{proof}
We first prove some preliminary claims.

\hspace{-3.5mm}{\bf Claim \ref{disjoint}.1} \textit{Suppose $s\notin U$ and $\text{\sc layer}^{\prime}(j)=\text{\sc layer}(j)$ for some $j\in{\bf v}(s)$.  Then:}

\begin{align}
\label{Case3v2s}
	v^{\prime}(s)&:=\left\{j:\arg\displaystyle\min_{\substack{j\in {\bf v}(s),\\ \text{\sc layer}^{\prime}(j)=\text{\sc layer}(j)}} k_{v(s),j}\right\}.
\end{align}

\begin{proof}
Since the premise of Claim \ref{disjoint}.1 is that there is some $j\in{\bf v}(s)$ such that $\text{\sc layer}^{\prime}(j)=\text{\sc layer}(j)$, it follows that $v^{\prime}(s)$ is non-empty.  To prove the claim, we invoke Lemma \ref{vssuffices} and check that (\ref{Case3v2s}) satisfies (\ref{firstIndexCodingShiftedComplete}) and (\ref{secondIndexShiftedComplete}).
%$v_{2}(j):=\arg\min_{\tilde{i}\in {\bf v}_{1}(j):\text{\sc layer}_{1}(\tilde{i})=\text{\sc layer}_{2}(\tilde{i})} k_{v_{1}(j),\tilde{i}}$
  
\hspace{-3.5mm}{\bf Claim \ref{disjoint}.1.a} \textit{For all $j\in{\bf f}(v(s),v^{\prime}(s))$ and $\text{\sc layer}^{\prime}(j)=\text{\sc layer}(j)$, the $v^{\prime}(s)$ defined in (\ref{Case3v2s}) satisfies (\ref{secondIndexShiftedComplete}).}
\begin{proof}
Consider the following inequalities:
\begin{align}
	\label{Case3Stage28}
	k_{v(s),j}&\geq\text{\sc layer}(v(s))-\text{\sc layer}(j),\\
	\label{Case3Stage29}
	&>\text{\sc layer}(v(s))-\text{\sc layer}(j),
\end{align} 
where (\ref{Case3Stage28}) follows from (\ref{firstIndexCodingShiftedComplete2}).  Since $\text{\sc layer}^{\prime}(j)=\text{\sc layer}(j)$ and $j\in{\bf f}(v(s),v^{\prime}(s))$, (\ref{Case3v2s}) implies that $j\notin{\bf v}(s)$ because otherwise $j\in{\bf v}(s)$ and $\text{\sc layer}^{\prime}(j)=\text{\sc layer}(j)$ and   $k_{v(s),j}<k_{v(s),v^{\prime}(s)}$ is a contradiction.  Moreover, (V5) implies that equality in (\ref{Case3Stage28}) only holds if $j\in{\bf v}(s)$ which implies (\ref{Case3Stage29}).  Now consider the following inequalities:
\begin{align}
	\label{Case3Stage210}
	k_{v(s),v^{\prime}(s)}&=\text{\sc layer}(v(s))-\text{\sc layer}(v^{\prime}(s)),\\
	\nonumber
	&=\text{\sc layer}(v(s))-\text{\sc layer}(j)\\
	\label{Case3Stage211}
	&\hspace{18mm}+\text{\sc layer}(j)-\text{\sc layer}(v^{\prime}(s)),
\end{align}
where (\ref{Case3Stage210}) follows from (V5) and because (\ref{Case3v2s}) implies $v^{\prime}(s)\in{\bf v}(s)$, and (\ref{Case3Stage211}) follows by adding and subtracting $\text{\sc layer}(j)$.  Rearranging (\ref{Case3Stage211}) gives:
\begin{align}
	\nonumber
	&\hspace{-2mm}\text{\sc layer}(j)-\text{\sc layer}(v^{\prime}(s))\\
	\label{Case3Stage212}
	&\hspace{9mm}=k_{v(s),v^{\prime}(s)}-(\text{\sc layer}(v(s))-\text{\sc layer}(j)),\\
	\label{Case3Stage213}
	&\hspace{9mm}>k_{v(s),v^{\prime}(s)}-k_{v(s),j},\\
	\label{Case3Stage214}
	&\hspace{9mm}=k_{j,v^{\prime}(s)},
\end{align}
where (\ref{Case3Stage213}) follows from substituting (\ref{Case3Stage29}) in (\ref{Case3Stage212}).  Continuing this series of inequalities gives:
\begin{align}
	\label{Case3Stage215}
	k_{j,v^{\prime}(s)}&<\text{\sc layer}(j)-\text{\sc layer}(v^{\prime}(s)),\\
	\label{Case3Stage216}
	&=\text{\sc layer}^{\prime}(j)-\text{\sc layer}^{\prime}(v^{\prime}(s)),
\end{align}
where (\ref{Case3Stage215}) follows from (\ref{Case3Stage214}).  The premise of Claim \ref{disjoint}.1 ensures $v^{\prime}(s)$ as defined in (\ref{Case3v2s}) is non-empty.  By construction $\text{\sc layer}^{\prime}(v^{\prime}(s))=\text{\sc layer}(v^{\prime}(s))$.  Therefore (\ref{Case3Stage216}) follows since $\text{\sc layer}^{\prime}(j)=\text{\sc layer}(j)$ is the premise of Claim \ref{disjoint}.1.a.  It follows from (\ref{Case3Stage216}) that (\ref{secondIndexShiftedComplete}) is satisfied.  
\end{proof}

\hspace{-3.5mm}{\bf Claim \ref{disjoint}.1.b} \textit{For all $j\in{\bf f}(v(s),v^{\prime}(s))$ and $\text{\sc layer}^{\prime}(j)>\text{\sc layer}(j)$, the $v^{\prime}(s)$ defined in (\ref{Case3v2s}) satisfies (\ref{secondIndexShiftedComplete}).}
\begin{proof}
Consider the following inequalities:
\begin{align}
	\label{Case3newStage21}
	\text{\sc layer}(j)-\text{\sc layer}(v^{\prime}(s))&\geq k_{v(s),v^{\prime}(s)}-k_{v(s),j},\\
	\label{Case3newStage22}
	&=k_{j,v^{\prime}(s)},
\end{align}  
where (\ref{Case3newStage21}) follows from substituting (\ref{firstIndexCodingShiftedComplete2}) into (\ref{Case3Stage212}) (nothing in (\ref{Case3Stage210})-(\ref{Case3Stage212}) requires $\text{\sc layer}^{\prime}(j)=\text{\sc layer}(j)$).  We have the following inequalities:
\begin{align}
	\label{Case3newStage23}
	k_{j,v^{\prime}(s)}&\leq\text{\sc layer}(j)-\text{\sc layer}(v^{\prime}(s)),\\
	\label{Case3newStage24}
	&<\text{\sc layer}^{\prime}(j)-\text{\sc layer}^{\prime}(v^{\prime}(s)), 
\end{align}
where (\ref{Case3newStage23}) follows from (\ref{Case3newStage22}), and (\ref{Case3newStage24}) follows because $\text{\sc layer}^{\prime}(j)>\text{\sc layer}(j)$ is the premise of Claim \ref{disjoint}.1.b and (\ref{Case3v2s}) implies $\text{\sc layer}^{\prime}(v^{\prime}(s))=\text{\sc layer}(v^{\prime}(s))$.  It follows from (\ref{Case3newStage24}) that (\ref{secondIndexShiftedComplete}) is satisfied.  
\end{proof}

Claim \ref{disjoint}.1.a and Claim \ref{disjoint}.1.b together show that $v^{\prime}(s)$ as defined in (\ref{Case3v2s}) satisfies (\ref{secondIndexShiftedComplete}) for all $j\in{\bf f}(v(s),v^{\prime}(s))$.  For every $j\in{\bf f}(s,v(s))$:
\begin{align}
	\label{Case3Stage217}
	\hspace{-2mm}k_{j,v^{\prime}(s)}&<\text{\sc layer}(j)-\text{\sc layer}(v(s))+k_{v(s),v^{\prime}(s)},\\
	\label{Case3Stage218}
	&=\text{\sc layer}(j)-\text{\sc layer}(v^{\prime}(s)),\\
	\label{Case3Stage219}
	&\leq\text{\sc layer}^{\prime}(j)-\text{\sc layer}^{\prime}(v^{\prime}(s)),
\end{align}
where (\ref{Case3Stage217}) follows by adding $k_{v(s),v^{\prime}(s)}$ to both sides of (\ref{secondIndexShiftedComplete2}) and simplifying $k_{j,v(s)}+k_{v(s),v^{\prime}(s)}=k_{j,v^{\prime}(s)}$, (\ref{Case3Stage218}) follows from (V5) which implies $\text{\sc layer}(v(s))-\text{\sc layer}(v^{\prime}(s))=k_{v(s),v^{\prime}(s)}$ if $v^{\prime}(s)\in{\bf v}(s)$ and (\ref{Case3v2s}) which implies $v^{\prime}(s)\in{\bf v}(s)$, and (\ref{Case3Stage219}) follows because (\ref{shift2}) implies $\text{\sc layer}^{\prime}(j)\geq\text{\sc layer}(j)$ and (\ref{Case3v2s}) implies $\text{\sc layer}^{\prime}(v^{\prime}(s))=\text{\sc layer}(v^{\prime}(s))$.  It follows from Claim \ref{disjoint}.1.a, Claim \ref{disjoint}.1.b, and (\ref{Case3Stage219}) that (\ref{Case3v2s}) satisfies (\ref{secondIndexShiftedComplete}).  For every $j\in{\bf f}(v^{\prime}(s),d)$:
\begin{align}
	\label{Case3Stage220}
	\hspace{-2mm}k_{v^{\prime}(s),j}&\geq\text{\sc layer}(v(s))-\text{\sc layer}(j)-k_{v(s),v^{\prime}(s)},\\
	\label{Case3Stage221}
	&=\text{\sc layer}(v^{\prime}(s))-\text{\sc layer}(j),\\
	\label{Case3Stage222}
	&\geq\text{\sc layer}^{\prime}(v^{\prime}(s))-\text{\sc layer}^{\prime}(j),
\end{align}
where (\ref{Case3Stage220}) follows by subtracting $k_{v(s),v^{\prime}(s)}$ from (\ref{firstIndexCodingShiftedComplete2}) and simplifying $k_{v(s),j}-k_{v(s),v^{\prime}(s)}=k_{v^{\prime}(s),j}$, (\ref{Case3Stage221}) follows from (V5) which implies $\text{\sc layer}(v(s))-\text{\sc layer}(v^{\prime}(s))=k_{v(s),v^{\prime}(s)}$ if $v^{\prime}(s)\in{\bf v}(s)$ and (\ref{Case3v2s}) which implies $v^{\prime}(s)\in{\bf v}(s)$, and (\ref{Case3Stage222}) follows because (\ref{shift2}) implies $\text{\sc layer}^{\prime}(j)\geq\text{\sc layer}(j)$ and (\ref{Case3v2s}) implies $\text{\sc layer}^{\prime}(v^{\prime}(s))=\text{\sc layer}(v^{\prime}(s))$.  It follows from (\ref{Case3Stage222}) that $v^{\prime}(s)$ as defined in (\ref{Case3v2s}) satisfies (\ref{firstIndexCodingShiftedComplete}) which completes the proof of Claim \ref{disjoint}.1.
\end{proof}

To complete the proof of Lemma \ref{disjoint}, fix $i\in A_{l}(S)\setminus A_{l}(U)$.  First, we show that the premises of Lemma \ref{disjoint} satisfy the premises of Claim \ref{disjoint}.1.  It follows from (\ref{A}) and (\ref{shift2}) that $i\in{\bf v}(s)$ for some $s\in S\setminus U$.  Moreover, (\ref{A}) and (\ref{shift2}) imply that $i\notin{\bf v}(\tilde{s})$ for all $\tilde{s}\in U$.  It follows from (\ref{shift2}) that $\text{\sc layer}^{\prime}(i)=\text{\sc layer}(i)$ which satisfies the premise of Claim \ref{disjoint}.1, that $\text{\sc layer}^{\prime}(j)=\text{\sc layer}(j)$ for some $j\in{\bf v}(s)$ and $s\notin U$.  Therefore $v^{\prime}(s)$ satisfies (\ref{Case3v2s}).  Finally, we have the following equalities for any $i\in{\bf v}(s)$ and $s\in S\setminus U$:
\begin{align}
	\nonumber
	&\text{\sc layer}^{\prime}(v^{\prime}(s))-\text{\sc layer}^{\prime}(i)\\
	\label{Case3Stage223}
	&\hspace{10mm}=\text{\sc layer}(v^{\prime}(s))-\text{\sc layer}(i),\\
	\label{preCase3Stage224}
	&\hspace{10mm}=(\text{\sc layer}(v(s))-k_{v(s),v^{\prime}(s)})-\text{\sc layer}(i),\\
	\nonumber
	&\hspace{10mm}=(\text{\sc layer}(v(s))-k_{v(s),v^{\prime}(s)})\\
	\label{Case3Stage224}
	&\hspace{33mm}\hspace{4mm}-(\text{\sc layer}(v(s))-k_{v(s),i}),\\
	\label{Case3Stage225}
	&\hspace{10mm}=k_{v(s),i}-k_{v(s),v^{\prime}(s)}\\
	\label{Case3Stage226}
	&\hspace{10mm}=k_{v^{\prime}(s),i}.
\end{align}
where (\ref{Case3Stage223}) follows because $\text{\sc layer}^{\prime}(i)=\text{\sc layer}(i)$ as shown previously and (\ref{Case3v2s}) implies $\text{\sc layer}^{\prime}(v^{\prime}(s))=\text{\sc layer}(v^{\prime}(s))$, (\ref{preCase3Stage224}) follows from (V5) since (\ref{Case3v2s}) implies $v^{\prime}(s)\in{\bf v}(s)$, (\ref{Case3Stage224}) follows from (V5) since  $i\in{\bf v}(s)$ by assumption, and (\ref{Case3Stage225}) follows from simplifying.  Since $i\in{\bf v}(s)$ by assumption and $\text{\sc layer}^{\prime}(i)=\text{\sc layer}(i)$ as shown previously, it follows from (\ref{Case3v2s}) that $v^{\prime}(s)\in{\bf v}(s)$ and $k_{v(s),i}>k_{v(s),v^{\prime}(s)}$.  Therefore (\ref{Case3Stage226}) follows from (\ref{Case3Stage225}).  It follows from (\ref{Case3Stage226}) and (V5) that $i\in{\bf v}^{\prime}(s)$ for some $s\in S\setminus U$.

Since $i\in A_{l}(S)\setminus A_{l}(U)$, it follows that $\text{\sc layer}(i)=l$.  Moreover, $\text{\sc layer}^{\prime}(i)=\text{\sc layer}(i)$, which implies $\text{\sc layer}^{\prime}(i)=l$.  By definition in (L5), $i\in L^{\prime}_{l}$ if and only if $\text{\sc layer}^{\prime}(i)=l$.  Since $i\in{\bf v}^{\prime}(s)$ for some $s\in S\setminus U$ and $i\in L^{\prime}_{l}$, it follows from (\ref{Aprime}) that $i\in A^{\prime}_{l}(S)$.
\end{proof}

%% file: remain.tex
\begin{lemma}
\label{remain}
For every $S\subseteq{\cal S}(d)$ and $l\in\{0,\ldots,|{\bf L}_{d}|-1\}$, $A^{\prime}_{l}(S\setminus U)	\subseteq A_{l}(S\setminus U)$
\end{lemma}
\begin{proof}
First we prove some preliminary claims.

\hspace{-3.5mm}\text{\bf Claim \ref{remain}.1}. \textit{For every $s\notin U$, there exists some $i\in{\bf v}(s)$ such that $\text{\sc layer}^{\prime}(i)=\text{\sc layer}(i)$}.

\begin{proof}
The proof is by contradiction.  Fix some $s\notin U$, $i\in{\bf v}(s)$, and assume without lost of generality, that $\text{\sc layer}(i)=l$.  If the claim is false, then $i\in{\bf v}(s)$ implies $\text{\sc layer}^{\prime}(i)>\text{\sc layer}(i)$.  Hence, (\ref{shift2}) implies $i\in A_{l}(U)$.  This argument applies to every $i\in{\bf v}(s)$.  It follows that for every $l=0,\ldots,|{\bf L}_{d}|-1$:
\begin{align}
	\nonumber
	\hspace{-60mm}&I(X_{A_{l}(U\cup\{s\})};Y_{d}|X_{\tilde{A}_{l}(U\cup\{s\})})\\
	\label{remain1}
	&\hspace{38mm}=I(X_{A_{l}(U)};Y_{d}|X_{\tilde{A}_{l}(U)}).
\end{align}
By definition, $U\subseteq{\cal S}(d)$ is the largest subset such that:
\begin{align}
\label{remain2}
	R_{U}&>\displaystyle\sum^{|{\bf L}_{d}|-1}_{l=0}I(X_{A_{l}(U)};Y_{d}|X_{\tilde{A}_{l}(U)}).
\end{align}
Since $R_{U\cup\{s\}}\geq R_{U}$, it follows from (\ref{remain1}) and (\ref{remain2}) that:
\begin{align}
\label{remain3}
	R_{U\cup\{s\}}&>\displaystyle\sum^{|{\bf L}_{d}|-1}_{l=0}I(X_{A_{l}(U\cup\{s\})};Y_{d}|X_{\tilde{A}_{l}(U\cup\{s\})}),
\end{align}
which contradicts the definition of $U$ as the largest subset of ${\cal S}(d)$ that violates (\ref{flowdecomp}).
\end{proof}

\hspace{-3.5mm}\text{\bf Claim \ref{remain}.2}. \textit{If $i\in{\bf v}^{\prime}(s)$ for any $s\in S\setminus U$ then $\text{\sc layer}^{\prime}(i)=\text{\sc layer}(i)$.}
\begin{proof}
Fix $i\in{\bf v}^{\prime}(s)$ for some $s\in S\setminus U$.  We have the following sequence of equalities:
\begin{align}
	\label{remain4}
	k_{v^{\prime}(s),i}&=\text{\sc layer}^{\prime}(v^{\prime}(s))-\text{\sc layer}^{\prime}(i),\\
	\label{remain5}
	&=\text{\sc layer}(v^{\prime}(s))-\text{\sc layer}^{\prime}(i),\\
	\label{remain6}
	&=\text{\sc layer}(v(s))-k_{v(s),v^{\prime}(s)}-\text{\sc layer}^{\prime}(i),
\end{align}
where (\ref{remain4}) follows from (V5) and because $i\in{\bf v}^{\prime}(s)$.  To justify (\ref{remain5}), observe that Claim \ref{disjoint}.3 implies $v^{\prime}(s)$ is defined by (\ref{Case3v2s}) provided there is some $j\in{\bf v}(s)$ such that $\text{\sc layer}^{\prime}(j)=\text{\sc layer}(j)$.  Claim \ref{remain}.1 implies that such a $j\in{\bf v}(s)$ exists.  By construction in (\ref{Case3v2s}), $\text{\sc layer}^{\prime}(v^{\prime}(s))=\text{\sc layer}(v^{\prime}(s))$, so (\ref{remain5}) follows from (\ref{remain4}).  Finally, (\ref{remain6}) follows from (V5) and because (\ref{Case3v2s}) implies  $v^{\prime}(s)\in{\bf v}(s)$.  Since $i\in{\bf v}^{\prime}(s)$ and $v^{\prime}(s)\in{\bf v}(s)$ it follows that $k_{v(s),v^{\prime}(s)}+k_{v^{\prime}(s),i}=k_{v(s),i}$.  Therefore:
\begin{align}
	\label{remain7}
	k_{v(s),i}&=\text{\sc layer}(v(s))-\text{\sc layer}^{\prime}(i),
\end{align} 
where (\ref{remain7}) follows from (\ref{remain6}) and because $k_{v(s),v^{\prime}(s)}+k_{v^{\prime}(s),i}=k_{v(s),i}$.  Moreover,
\begin{align}
	\label{remain8}
	\text{\sc layer}^{\prime}(i)\geq\text{\sc layer}(i),
\end{align}
where (\ref{remain8}) follows from (\ref{shift2}).  Therefore:
\begin{align}
	\label{remain9}
	\text{\sc layer}(v(s))-\text{\sc layer}(i)&\geq k_{v(s),i},
\end{align}
where (\ref{remain9}) follows from substituting (\ref{remain8}) into (\ref{remain7}).  Now $i\in{\bf f}(v^{\prime}(s),d)$ since $i\in{\bf v}^{\prime}(s)$ and $v^{\prime}(s)\in{\bf f}(v(s),d)$ since (\ref{Case3v2s}) implies $v^{\prime}(s)\in{\bf v}(s)$.  Therefore $i\in{\bf f}(v(s),d)$.  Invoking (\ref{firstIndexCodingShiftedComplete2}) and (\ref{remain9}) together gives:
\begin{align}
	\label{remain91}
	\text{\sc layer}(v(s))-\text{\sc layer}(i)&=k_{v(s),i}.
\end{align}
Comparing (\ref{remain91}) and (\ref{remain7}) gives $\text{\sc layer}^{\prime}(i)=\text{\sc layer}(i)$.  
\end{proof}
To complete the proof of Lemma \ref{remain}, fix some $i\in A^{\prime}_{l}(S\setminus U)$.  It follows from (\ref{Aprime}) that $i\in{\bf v}^{\prime}(s)$ for some $s\in S\setminus U$.  Since $i\in{\bf v}^{\prime}(s)$, Claim \ref{remain}.2 implies that $\text{\sc layer}^{\prime}(i)=\text{\sc layer}(i)$.  Consider the following inequalities:
\begin{align}
	\nonumber
	&\text{\sc layer}(v(s))-\text{\sc layer}(i)\\
	\nonumber
	&\hspace{3mm}=\text{\sc layer}(v(s))-\text{\sc layer}(v^{\prime}(s))\\
	\label{remain10}
	&\hspace{32mm}+\text{\sc layer}(v^{\prime}(s))-\text{\sc layer}(i),\\
	\label{remain11}
	&\hspace{3mm}=k_{v(s),v^{\prime}(s)}+(\text{\sc layer}(v^{\prime}(s))-\text{\sc layer}(i)),\\
	\label{remain12}
	&\hspace{3mm}=k_{v(s),v^{\prime}(s)}+(\text{\sc layer}^{\prime}(v^{\prime}(s))-\text{\sc layer}^{\prime}(i)),\\
	\label{remain13}
	&\hspace{3mm}=k_{v(s),v^{\prime}(s)}+k_{v^{\prime}(s),i},\\
	\label{remain14}
	&\hspace{3mm}=k_{v(s),i},
\end{align}
where (\ref{remain10}) follows by adding and subtracting $\text{\sc layer}(v^{\prime}(s))$ from the left side, (\ref{remain11}) follows because (V5) implies $\text{\sc layer}(v(s))-\text{\sc layer}(v^{\prime}(s))=k_{v(s),v^{\prime}(s)}$ if $v^{\prime}(s)\in{\bf v}(s)$ and (\ref{Case3v2s}) implies $v^{\prime}(s)\in{\bf v}(s)$, (\ref{remain12}) follows because (\ref{Case3v2s}) implies $\text{\sc layer}^{\prime}(v^{\prime}(s))=\text{\sc layer}(v^{\prime}(s))$ and Claim \ref{remain}.2 implies $\text{\sc layer}^{\prime}(i)=\text{\sc layer}(i)$ if $i\in{\bf v}^{\prime}(s))$, and (\ref{remain13}) follows from (V5) which implies that $\text{\sc layer}^{\prime}(v^{\prime}(s))-\text{\sc layer}^{\prime}(i)=k_{v^{\prime}(s),i}$ if $i\in{\bf v}^{\prime}(s)$.  Since (\ref{remain14}) implies that $\text{\sc layer}(v(s))-\text{\sc layer}(i)=k_{v(s),i}$, it follows from (V5) that $i\in{\bf v}(s)$.  

Now $i\in A^{\prime}_{l}(S\setminus U)$, so (\ref{Aprime}) implies that $\text{\sc layer}^{\prime}(i)=l$.  Moreover, $i\in{\bf v}(s)$ for some $s\in S\setminus U$, so Claim \ref{remain}.2 implies $\text{\sc layer}^{\prime}(i)=\text{\sc layer}(i)$.  Therefore $\text{\sc layer}(i)=l$ which implies $i\in L_{l}$.  Since $i\in{\bf v}(s)$ and $s\in S\setminus U$ and $i\in L_{l}$, it follows from (\ref{A}) that $i\in A_{l}(S\setminus U)$.
\end{proof}

%% file: join.tex
\begin{lemma}
\label{join}
	For every $S\subseteq{\cal S}(d)$ and $l\in\{0,\ldots,|{\bf L}_{d}|-1\}$, $A_{l-1}(S\cap U)\subseteq A^{\prime}_{l}(S\cap U)$
\end{lemma}
\begin{proof}
First, we introduce some preliminary claims.

\hspace{-3.5mm}{\bf Claim \ref{join}.1} \textit{Suppose $i\in{\bf v}(s)$ for some $s\in U$. If $v^{\prime}(s)=v(s)$ then $i\in{\bf v}^{\prime}(s)$.}

\begin{proof}
Fix $i\in{\bf v}(s)$ and consider the following inequalities:
\begin{align}
\nonumber
	&\hspace{-2mm}\text{\sc layer}^{\prime}(v^{\prime}(s))-\text{\sc layer}^{\prime}(i)\\
	\label{shiftedcomplete1}
	&\hspace{15mm}=\text{\sc layer}^{\prime}(v(s))-\text{\sc layer}^{\prime}(i),\\
	\label{shiftedcomplete2}
	&\hspace{15mm}=(\text{\sc layer}(v(s))+1)-(\text{\sc layer}(i)+1),\\
	\nonumber
	&\hspace{15mm}=\text{\sc layer}(v(s))-\text{\sc layer}(i),\\
	\label{shiftedcomplete4}
	&\hspace{15mm}=k_{v(s),i},\\
	\label{shiftedcomplete5}
	&\hspace{15mm}=k_{v^{\prime}(s),i},
\end{align}
where (\ref{shiftedcomplete1}) follows because $v^{\prime}(s)=v(s)$ is the premise of Claim \ref{join}.1, (\ref{shiftedcomplete2}) follows from (\ref{shift2}) and because $s\in U$ and $i\in{\bf v}(s)$, (\ref{shiftedcomplete4}) follows from (V5) and because $i\in{\bf v}(s)$, and (\ref{shiftedcomplete5}) follows because $v^{\prime}(s)=v(s)$ is the premise of Claim \ref{join}.1.  It follows from (V5) and (\ref{shiftedcomplete5}) that $i\in{\bf v}^{\prime}(s)$.	
\end{proof} 

\hspace{-3.5mm}{\bf Claim \ref{join}.2} \textit{Suppose $i\in{\bf v}(s)$ for some $s\in U$. If $v^{\prime}(s)\neq v(s)$ then $i\in{\bf v}^{\prime}(s)$.}

\begin{proof}
First we prove some preliminary sub-claims.

\hspace{-3.5mm}{\bf Claim \ref{join}.2.a} \textit{If $s\in U$ and $v^{\prime}(s)\neq v(s)$ then $v^{\prime}(s)\in{\bf f}(s,v(s))$.}
\begin{proof}
The proof is by contradiction.  Suppose $v^{\prime}(s)\in{\bf f}(v(s),d)$.  It is convenient to consider $\text{\sc layer}^{\prime}(v^{\prime}(s))>\text{\sc layer}(v^{\prime}(s))$ and $\text{\sc layer}^{\prime}(v^{\prime}(s))=\text{\sc layer}(v^{\prime}(s))$ separately.  Suppose $\text{\sc layer}^{\prime}(v^{\prime}(s))>\text{\sc layer}(v^{\prime}(s))$:
\begin{align}
	\nonumber
	&\hspace{-2mm}k_{v(s),v^{\prime}(s)}\\
	\label{shiftedcomplete6}	
	&\hspace{5mm}\geq\text{\sc layer}(v(s))-\text{\sc layer}(v^{\prime}(s)),\\
	\label{shiftedcomplete7}
	&\hspace{5mm}=(\text{\sc layer}^{\prime}(v(s))-1)-(\text{\sc layer}^{\prime}(v^{\prime}(s))-1),\\
	\label{shiftedcomplete8}
	&\hspace{5mm}=\text{\sc layer}^{\prime}(v(s))-\text{\sc layer}^{\prime}(v^{\prime}(s)),
\end{align}
where (\ref{shiftedcomplete6}) follows from (\ref{firstIndexCodingShiftedComplete2}) and because $v^{\prime}(s)\in{\bf f}(v(s),d)$, (\ref{shiftedcomplete7}) follows from (\ref{shift2}), $s\in U$, and the assumption $\text{\sc layer}^{\prime}(v^{\prime}(s))>\text{\sc layer}(v^{\prime}(s))$.  Since $v(s)\in{\bf f}(s,v^{\prime}(s))$, (\ref{shiftedcomplete8}) contradicts (\ref{secondIndexShiftedComplete}).  Now suppose $\text{\sc layer}^{\prime}(v^{\prime}(s))=\text{\sc layer}(v^{\prime}(s))$: 
\begin{align}
	\label{shiftedcomplete9}
	k_{v(s),v^{\prime}(s)}&\geq\text{\sc layer}(v(s))-\text{\sc layer}(v^{\prime}(s)),\\
	\label{shiftedcomplete10}
	&>\text{\sc layer}(v(s))-\text{\sc layer}(v^{\prime}(s)),\\
	\label{shiftedcomplete11}
	&=(\text{\sc layer}^{\prime}(v(s))-1)-\text{\sc layer}^{\prime}(v^{\prime}(s)),
\end{align}
where (\ref{shiftedcomplete9}) follows from (\ref{firstIndexCodingShiftedComplete2}). Since $\text{\sc layer}(v^{\prime}(s))=\text{\sc layer}^{\prime}(v^{\prime}(s))$, it follows from (\ref{shift2}) that $v^{\prime}(s)\notin{\bf v}(j)$ for any $j\in U$.  Therefore $v^{\prime}(s)\notin{\bf v}(s)$ since $s\in U$ is the premise of Claim \ref{join}.2.a.  From (V5), equality in (\ref{shiftedcomplete9}) occurs only if $v^{\prime}(s)\in{\bf v}(s)$.  Since $v^{\prime}(s)\notin{\bf v}(s)$, (\ref{shiftedcomplete10}) follows from (\ref{shiftedcomplete9}).  Finally, (\ref{shiftedcomplete11}) follows because (\ref{shift2}) implies $\text{\sc layer}^{\prime}(v(s))=\text{\sc layer}(v(s))+1$ since $s\in U$.  Consider the following inequalities:  
\begin{align}
	\label{shiftedcomplete12}
	\text{\sc layer}^{\prime}(v(s))-\text{\sc layer}^{\prime}(v^{\prime}(s))&<k_{v(s),v^{\prime}(s)}+1,\\
	\label{shiftedcomplete13}
	&\leq k_{v(s),v^{\prime}(s)},
\end{align}
where (\ref{shiftedcomplete12}) follows from (\ref{shiftedcomplete11}), and (\ref{shiftedcomplete13}) follows because the encoding delays are integer-valued.  Since $v(s)\in{\bf f}(s,v^{\prime}(s))$, (\ref{shiftedcomplete13}) contradicts (\ref{secondIndexShiftedComplete}).  It follows from (\ref{secondIndexShiftedComplete}), (\ref{shiftedcomplete8}), and (\ref{shiftedcomplete13}) that $v^{\prime}(s)\in{\bf f}(s,v(s))$.
\end{proof}

\hspace{-3.5mm}{\bf Claim \ref{join}.2.b}. \textit{If $s\in U$ and $v^{\prime}(s)\neq v(s)$ then $\text{\sc layer}^{\prime}(v^{\prime}(s))=\text{\sc layer}(v^{\prime}(s))$.} 

\begin{proof}
	The proof is by contradiction.  Suppose $\text{\sc layer}^{\prime}(v^{\prime}(s))\neq\text{\sc layer}(v^{\prime}(s))$.  Since (\ref{shift2}) implies $\text{\sc layer}(v^{\prime}(s))\leq\text{\sc layer}^{\prime}(v^{\prime}(s))\leq\text{\sc layer}(v^{\prime}(s))+1$, it follows that:
	\begin{align}
		\label{claim2b1}
		\text{\sc layer}^{\prime}(v^{\prime}(s))&=\text{\sc layer}(v^{\prime}(s))+1.
	\end{align}
	Consider the following inequalities:
	\begin{align}
	\label{claim2b2}
	k_{v^{\prime}(s),v(s)}&<\text{\sc layer}(v^{\prime}(s))-\text{\sc layer}(v(s)),\\
	\label{claim2b3}
	&=(\text{\sc layer}^{\prime}(v^{\prime}(s))-1)-\text{\sc layer}(v(s)),\\
	\nonumber
	&=(\text{\sc layer}^{\prime}(v^{\prime}(s))-1)\\
	\label{claim2b4}
	&\hspace{27mm}-(\text{\sc layer}^{\prime}(v(s))-1),\\
	\label{claim2b5}
	&=\text{\sc layer}^{\prime}(v^{\prime}(s))-\text{\sc layer}^{\prime}(v(s)),
\end{align}
where (\ref{claim2b2}) follows from (\ref{secondIndexShiftedComplete2}) and because $v^{\prime}(s)\in{\bf f}(s,v(s))$, (\ref{claim2b3}) follows from (\ref{claim2b1}), and (\ref{claim2b4}) follows from (\ref{shift2}) and because $s\in U$ is the premise of Claim \ref{join}.2.  Since Claim \ref{join}.1 implies $v^{\prime}(s)\in{\bf f}(s,v(s))$, it follows that $v(s)\in{\bf f}(v^{\prime}(s),d)$.  Therefore (\ref{claim2b5}) contradicts (\ref{firstIndexCodingShiftedComplete}).
\end{proof}

\hspace{-3.5mm}{\bf Claim \ref{join}.2.c}. \textit{If $s\in U$ and $v^{\prime}(s)\neq v(s)$ then $\text{\sc layer}^{\prime}(v^{\prime}(s))-\text{\sc layer}^{\prime}(v(s))=k_{v^{\prime}(s),v(s)}$.}   

\begin{proof}
Consider the following inequalities:
\begin{align}
	\label{shiftedcomplete14}
	k_{v^{\prime}(s),v(s)}&<\text{\sc layer}(v^{\prime}(s))-\text{\sc layer}(v(s)),\\
	\label{preshiftedcomplete15}
	&=\text{\sc layer}^{\prime}(v^{\prime}(s))-\text{\sc layer}(v(s)),\\
	\label{shiftedcomplete15}
	&=\text{\sc layer}^{\prime}(v^{\prime}(s))-(\text{\sc layer}^{\prime}(v(s))-1),
\end{align}
where (\ref{shiftedcomplete14}) follows from (\ref{secondIndexShiftedComplete2}) and because $v^{\prime}(s)\in{\bf f}(s,v(s))$, (\ref{preshiftedcomplete15}) follows from Claim \ref{join}.2.b, and (\ref{shiftedcomplete15}) follows from (\ref{shift2}) since $s\in U$.  Consider the following inequalities:
\begin{align}
\label{shiftedcomplete16}
\text{\sc layer}^{\prime}(v^{\prime}(s))-\text{\sc layer}^{\prime}(v(s))&>k_{v^{\prime}(s),v(s)}-1,\\
\label{shiftedcomplete17}
&\geq k_{v^{\prime}(s),v(s)},
\end{align}
where (\ref{shiftedcomplete16}) follows from (\ref{shiftedcomplete15}), (\ref{shiftedcomplete17}) follows from (\ref{shiftedcomplete16}) and because the encoding delays are integer-valued.  Since Claim \ref{join}.2.a implies $v^{\prime}(s)\in{\bf f}(s,v(s))$, it follows that $v(s)\in{\bf f}(v^{\prime}(s),d)$.  Invoking (\ref{firstIndexCodingShiftedComplete}) and (\ref{shiftedcomplete17}) together gives: 
\begin{align}
\label{shiftedcomplete18}
\text{\sc layer}^{\prime}(v^{\prime}(s))-\text{\sc layer}^{\prime}(v(s))=k_{v^{\prime}(s),v(s)},
\end{align}
which completes the proof of Claim \ref{join}.2.c.
\end{proof}
To complete the proof of Claim \ref{join}.2, fix $i\in{\bf v}(s)$ for some $s\in U$ and  consider the following inequalities:
\begin{align}
	\nonumber
	&\text{\sc layer}^{\prime}(v^{\prime}(s))-\text{\sc layer}^{\prime}(i)\\
	\nonumber
	&=\text{\sc layer}^{\prime}(v^{\prime}(s))-\text{\sc layer}^{\prime}(v(s))\\
	\label{preshiftedcomplete20}
	&\hspace{30mm}+\text{\sc layer}^{\prime}(v(s))-\text{\sc layer}^{\prime}(i),\\
	\nonumber
	&=\text{\sc layer}^{\prime}(v^{\prime}(s))-\text{\sc layer}^{\prime}(v(s))\\
	\label{shiftedcomplete20}
	&\hspace{14mm}+(\text{\sc layer}(v(s))-1)-(\text{\sc layer}(i)-1),\\
	\nonumber
	&=(\text{\sc layer}^{\prime}(v^{\prime}(s))-\text{\sc layer}^{\prime}(v(s)))\\
	\label{shiftedcomplete21}
	&\hspace{29mm}+(\text{\sc layer}(v(s))-\text{\sc layer}(i)),\\
	\label{shiftedjoin1}
	&=k_{v^{\prime}(s),v(s)}+(\text{\sc layer}(v(s))-\text{\sc layer}(i))\\
	\label{shiftedjoin2}
	&=k_{v^{\prime}(s),v(s)}+k_{v(s),i},\\
	\label{shiftedcomplete23}
	&=k_{v^{\prime}(s),i},
\end{align}
where (\ref{preshiftedcomplete20}) follows by adding and subtracting $\text{\sc layer}^{\prime}(v(^{\prime}(s))$ from the left-hand side of the equation, (\ref{shiftedcomplete20}) follows from (\ref{shift2}) since $s\in U$ and $i\in{\bf v}(s)$, (\ref{shiftedcomplete21}) follows from rearranging (\ref{shiftedcomplete20}), (\ref{shiftedjoin1}) follows from Claim \ref{join}.2.c, and (\ref{shiftedjoin2}) follows from (V5) since $i\in{\bf v}(s)$.  It follows from (\ref{shiftedcomplete23}) and (V5) that $i\in{\bf v}^{\prime}(s)$.
\end{proof}

To complete the proof of Lemma \ref{join}, fix $i\in A_{l-1}(S\cap U)$.  It follows from (\ref{A}) that $i\in{\bf v}(s)$ for some $s\in S\cap U$ and $\text{\sc layer}(i)=l-1$.  Since $i\in A_{l-1}(U)$, it follows from (\ref{shift2}) that $\text{\sc layer}^{\prime}(i)=\text{\sc layer}(i)+1=l$.  There are two cases to consider: $v^{\prime}(s)=v(s)$ and $v^{\prime}(s)\neq v(s)$.  In both cases Claim \ref{join}.1 and Claim \ref{join}.2 respectively show that $i\in{\bf v}^{\prime}(s)$.  Since $i\in{\bf v}^{\prime}(s)$ for some some $s\in S\cap U$ and $\text{\sc layer}(i)=l$, it follows from (\ref{Aprime}) that $i\in A^{\prime}_{l}(S\cap U)$.
\end{proof}

%% file: induction.tex
\begin{lemma}
	\label{induction}
	$Z^{\prime}=({\cal S}(d)\setminus U)\cup Z$.
\end{lemma}
\begin{proof}
	The proof is by contradiction.  Suppose there is some $S\subseteq({\cal S}\setminus U)\cup Z$ that violates (\ref{flowdecomp2}) for $({\bf F},{\bf L}^{\prime}_{d})$.  Then:
	\begin{align}
		\label{lemmaVprep2}
		R_{S}&>\sum^{|{\bf  L}^{\prime}_{d}|-1}_{l=0}I(X_{A^{\prime}_{l}(S)};Y_{d}|X_{\tilde{A}^{\prime}_{l}(S)}).
	\end{align}
	To simplify the analysis, we first examine the case $S\cap U=\{\}$ and then consider $S\cap U\neq\{\}$.
	
	\hspace{-3.5mm}{\bf Case 1:}  $S\cap U=\{\}$.
	
	We will prove the following sequence of inequalities:
	\begin{align}
		\nonumber
		R_{U\cup S}&>\sum^{|{\bf L}_{d}|-1}_{k=0}I(X_{A_{k}(U)};Y_{d}|X_{\tilde{A}_{k}(U)})\\
		\label{lemmaU1}
		&\hspace{20mm}+\sum^{|{\bf L}^{\prime}_{d}|-1}_{l=0}I(X_{A^{\prime}_{l}(S)};Y_{d}|X_{\tilde{A}^{\prime}_{l}(S)}),\\
		\nonumber
		&\geq\sum^{|{\bf L}_{d}|-1}_{k=0}I(X_{A_{k}(U)};Y_{d}|X_{\tilde{A}_{k}(U)})\\
		\label{lemmaU2}
		&\hspace{3mm}+\sum^{|{\bf L}_{d}|-1}_{l=0}I(X_{A_{l}(S)\setminus A_{l}(U)};Y_{d}|X_{\tilde{A}_{l}(S)\setminus A_{l}(U)}),\\
		\nonumber
		&=\sum^{|{\bf L}_{d}|-1}_{k=0}(I(X_{A_{k}(U)};Y_{d}|X_{\tilde{A}_{k}(U)})\\
		\label{lemmaU3}
		&\hspace{8mm}+I(X_{A_{k}(S)\setminus A_{k}(U)};Y_{d}|X_{\tilde{A}_{k}(S)\setminus A_{k}(U)})),\\
		\label{lemmaU4}
		&=\sum^{|{\bf L}_{d}|-1}_{k=0}I(X_{A_{k}(U)\cup A_{k}(S)};Y_{d}|X_{\tilde{A}_{k}(S)\setminus A_{k}(U)}),\\
		\label{lemmaU5}
		&=\sum^{|{\bf L}_{d}|-1}_{k=0}I(X_{A_{k}(U\cup S)};Y_{d}|X_{\tilde{A}_{k}(U\cup S)}),
	\end{align}
which contradicts the assumption that $U$ is the largest subset that violates (\ref{flowdecomp}) since $S$ and $U$ are disjoint.  Define:
\begin{align}
	\label{induction3}
	F_{I(d)}&:=\cup_{j\in I(d)}F_{j}({\cal N})\setminus F_{d}({\cal N}),
\end{align}
where $I(d)$ is defined in (\ref{id}).  To justify (\ref{lemmaU1}) - (\ref{lemmaU5}), observe that $\tilde{A}^{\prime}_{l}(S)$ defined in (\ref{Atildeprime}) can be rewritten as follows for any $S\subseteq{\cal S}(d)$:
\begin{align}
	\label{induction1}
	\tilde{A}^{\prime}_{l}(S)&=F_{I(d)}\cup F_{d}({\cal N})\setminus(\cup^{|{\bf L}^{\prime}_{d}|-1}_{q=l+1}L^{\prime}_{q}\cup A^{\prime}_{l}(S)),\\
	\label{induction2}
	&=F_{I(d)}\cup(\cup^{l}_{q=0}L^{\prime}_{q})\setminus A^{\prime}_{l}(S),
\end{align}
where (\ref{induction1}) follows from (\ref{Atildeprime}) and (\ref{induction3}), and (\ref{induction2}) follows because $\cup^{|{\bf L}^{\prime}_{d}|-1}_{q=0}L^{\prime}_{q}=F_{d}({\cal N})$.  Similarly, $\tilde{A}_{l}(S)$ defined in (\ref{Atilde}) can be rewritten as follows for any $S\subseteq{\cal S}(d)$:
\begin{align}
	\label{induction7}
	\tilde{A}_{l}(S)&=F_{I(d)}\cup(\cup^{l}_{q=0}L_{q})\setminus A_{l}(S).
\end{align}  

Now observe that (\ref{lemmaU1}) follows from (\ref{lemmaVprep2}) and the definition of $U$.  To justify (\ref{lemmaU2}) we prove the following claims:

\hspace{-3.5mm}{\bf Claim \ref{induction}.1}  \textit{If $S_{1}\subseteq{\cal N}$, $S_{3}\subseteq S_{2}\subseteq{\cal N}$ and $S_{1}\cap S_{2}=\{\}$ then:}
\begin{align}
	\nonumber
	&\hspace{-2.5mm}I(X_{S_{1}};Y_{d}|X_{S_{2}})+I(X_{S_{3}};Y_{d}|X_{S_{2}\setminus S_{3}})\\
	&\hspace{35mm}=I(X_{S_{1}\cup S_{3}};Y_{d}|X_{S_{2}\setminus S_{3}}).
\end{align}
\begin{proof}
	Apply the chain rule.
\end{proof}

\hspace{-3.5mm}{\bf Claim \ref{induction}.2}  \textit{If $S_{1}\subseteq S_{2}\subseteq S_{3}\subseteq{\cal N}$ then:} 
\begin{align}
\label{induction9}
I(X_{S_{2}};Y_{d}|X_{S_{3}\setminus S_{2}})&\geq I(X_{S_{1}};Y_{d}|X_{S_{3}\setminus S_{1}}).
\end{align}

\begin{proof}
Consider the following inequalities:
\begin{align}
	\nonumber
	&\hspace{-2.5mm}I(X_{S_{1}};Y_{d}|X_{S_{3}\setminus S_{1}})+I(X_{S_{2}\setminus S_{1}};Y_{d}|X_{S_{3}\setminus(S_{1}\cup S_{2})})\\
	\label{preinduction4}
	&\hspace{21mm}=I(X_{S_{1}\cup(S_{2}\setminus S_{1})};Y_{d}|X_{S_{3}\setminus(S_{1}\cup S_{2})}),\\
	\label{induction4}
	&\hspace{21mm}=I(X_{S_{1}\cup(S_{2}\setminus S_{1})};Y_{d}|X_{S_{3}\setminus S_{2}}),\\
	\label{induction5}
	&\hspace{21mm}=I(X_{S_{1}\cup S_{2}};Y_{d}|X_{S_{3}\setminus S_{2}}),\\
	\label{induction6},
	&\hspace{21mm}=I(X_{S_{2}};Y_{d}|X_{S_{3}\setminus S_{2}}).
\end{align}
To justify (\ref{preinduction4}), let $S^{\prime}_{1}:=S_{1}$, $S^{\prime}_{2}:=S_{3}\setminus S_{1}$, and $S^{\prime}_{3}:=S_{2}\setminus S_{1}$.  First we verify that the premises of Claim \ref{induction}.1 are satisfied.  Since $S_{1}\subseteq S_{2}\subseteq S_{3}$, it follows that $(S_{2}\setminus S_{1})\subseteq(S_{3}\setminus S_{1})$, which implies $S^{\prime}_{3}\subseteq S^{\prime}_{2}$.  Furthermore, $S_{1}\cap(S_{3}\setminus S_{1})=\{\}$ which implies $S^{\prime}_{1}\cap S^{\prime}_{2}=\{\}$.  Finally, $S^{\prime}_{2}\setminus S^{\prime}_{3}=(S_{3}\setminus S_{1})\setminus (S_{2}\setminus S_{1})=S_{3}\setminus(S_{1}\cup(S_{2}\setminus S_{1}))=S_{3}\setminus(S_{1}\cup S_{2})$.  Now we can invoke Claim \ref{induction}.1 to justify (\ref{preinduction4}).  Since $S_{1}\subseteq S_{2}$, (\ref{induction4}) follows from (\ref{preinduction4}) and the fact that $S_{1}\subseteq S_{2}$ and (\ref{induction6}) follows from (\ref{induction5}) and the same fact.  By the non-negativity of mutual-information, $I(X_{S_{2}\setminus S_{1}};Y_{d}|X_{S_{3}\setminus(S_{1}\cup S_{2})})\geq0$.  Therefore (\ref{induction6}) implies (\ref{induction9}) which completes the proof of Claim \ref{induction}.2.
\end{proof}

\hspace{-3.5mm}{\bf Claim \ref{induction}.3} \textit{For any $S\subseteq{\cal S}(d)$ and every $l\in 0,\ldots,|{\bf L}_{d}|-1$:}
\begin{align}
\nonumber
&\hspace{-2mm}I(X_{A^{\prime}_{l}(S)};Y_{d}|X_{\tilde{A}^{\prime}_{l}(S)})\\
\label{induction8}
&\hspace{19mm}\geq I(X_{A_{l}(S)\setminus A_{l}(U)};Y_{d}|X_{\tilde{A}_{l}(S)\setminus A_{l}(U)})
	\end{align}	
	
\begin{proof}
To justify (\ref{induction8}), we invoke Claim \ref{induction}.2.  First we verify that the premises of Claim \ref{induction}.2 are satisfied.  Let  $S_{1}:=A_{l}(S)\setminus A_{l}(U)$, $S_{2}:=A^{\prime}_{l}(S)$, and $S_{3}:=F_{I(d)}\cup(\cup^{l}_{q=0}L^{\prime}_{q})$.  Lemma \ref{disjoint} implies $S_{1}\subseteq S_{2}$.  It follows from (\ref{Aprime}) that $A^{\prime}_{l}(S)\subseteq L^{\prime}_{l}$.  Therefore $S_{2}\subseteq S_{3}$.  Moreover, (\ref{induction2}) implies $\tilde{A}^{\prime}_{l}(S):=F_{I(d)}\cup(\cup^{l}_{q=0}L^{\prime}_{q})\setminus A^{\prime}_{l}(S)=S_{3}\setminus S_{2}$.  In addition:
\begin{align}
\label{preinduction14}
S_{3}\setminus S_{1}=F_{I(d)}\cup(\cup^{l}_{q=0}L^{\prime}_{q})\setminus(A_{l}(S)\setminus A_{l}(U)).	
\end{align}
For every $l=0,\ldots,|{\bf L}_{d}|-1$, (\ref{shift2}) implies:
\begin{align}
\label{induction14}
\cup^{l}_{q=0}L^{\prime}_{q}=\cup^{l}_{q=0}L_{q}\setminus A_{l}(U).	
\end{align}
Substituting (\ref{induction14}) into (\ref{preinduction14}) gives $S_{3}\setminus S_{1}=F_{I(d)}\cup(\cup^{l}_{q=0}L_{q}\setminus A_{l}(U))\setminus(A_{l}(S)\setminus A_{l}(U))=F_{I(d)}\cup(\cup^{l}_{q=0}L_{q})\setminus(A_{l}(S)\cup A_{l}(U))$.  It follows that $S_{3}\setminus S_{1}=\tilde{A}_{l}(S)\setminus A_{l}(U)$, since (\ref{induction7}) implies $\tilde{A}_{l}(S):=F_{I(d)}\cup(\cup^{l}_{q=0}L_{q})\setminus A_{l}(S)$.  Therefore (\ref{induction8}) follows from Claim \ref{induction}.2.
\end{proof}

To justify (\ref{lemmaU2}), invoke Claim \ref{induction}.3 for $l=0,\ldots,|{\bf L}_{d}|-1$ and observe that (\ref{shift2}) implies $|{\bf L}_{d}|\leq|{\bf L}^{\prime}_{d}|$.  Since mutual informations are always non-negative, (\ref{lemmaU2}) is justified.  

\hspace{-3.5mm}{\bf Claim \ref{induction}.4} \textit{For any $S\subseteq{\cal S}(d)$ and every $l\in 0,\ldots,|{\bf L}_{d}|-1$:}
\begin{align}
\nonumber
&\hspace{-4mm}I(X_{A_{k}(U)};Y_{d}|X_{\tilde{A}^{\prime}_{k}(U)})\\
\nonumber
&\hspace{13mm}+I(X_{A_{k}(S)\setminus A_{k}(U)};Y_{d}|X_{\tilde{A}_{k}(S)\setminus A_{k}(U)})\\
\label{induction8a}
&\hspace{13mm}=I(X_{A_{k}(U)\cup A_{k}(S)};Y_{d}|X_{\tilde{A}_{k}(S)\setminus A_{k}(U)})
	\end{align}
\begin{proof}
To prove (\ref{induction8a}), we invoke Claim 6.1.  First, we verify that the premises of Claim 6.1 are satisfied.  Let $S_{1}:=A_{k}(U)$, $S_{2}:=\tilde{A}_{k}(U)$, and $S_{3}:=A_{k}(S)\setminus A_{k}(U)$.  From (\ref{induction7}):
\begin{align}
	\label{induction8b}
	\tilde{A}_{k}(U)=F_{I(d)}\cup(\cup^{k}_{q=0}L_{q})\setminus A_{k}(U).
\end{align}
Since $A_{k}(S)\subseteq L_{k}$, (\ref{induction8b}) implies $S_{3}\subseteq S_{2}$.  By inspection, $S_{1}\cap S_{2}=\{\}$.  Finally,
\begin{align}
	\nonumber
	&\hspace{-3mm}S_{2}\setminus S_{3}\\
	\nonumber
	&\hspace{-3mm}=\tilde{A}_{k}(U)\setminus (A_{k}(S)\setminus A_{k}(U)),\\
	\nonumber
	&\hspace{-3mm}=(F_{I(d)}\cup(\cup^{l}_{q=0}L_{q})\setminus A_{k}(U))\setminus(A_{k}(S)\setminus A_{k}(U)),\\
%	\nonumber
%	&\hspace{3mm}\hspace{-9mm}\hspace{30mm}\setminus(A_{k}(S)\setminus A_{k}(U)),\\
	\nonumber
	&\hspace{-3mm}=(F_{I(d)}\cup(\cup^{l}_{q=0}L_{q}))\setminus(A_{k}(S)\cup A_{k}(U)),\\
	\nonumber
	%&\hspace{3mm}\hspace{-9mm}\hspace{30mm}\setminus(A_{k}(S)\cup A_{k}(U)),\\
	\nonumber
	&\hspace{-3mm}=(F_{I(d)}\cup(\cup^{l}_{q=0}L_{q})\setminus A_{k}(S))\setminus A_{k}(U),\\
	\label{induction8d}
	&\hspace{-3mm}=\tilde{A}_{k}(S)\setminus A_{k}(U),
\end{align}
where the equalities follow from normal set-theoretic operations and (\ref{induction8b}).  Furthermore,
	\begin{align}
		\nonumber
		S_{1}\cup S_{3}&=A_{k}(U)\cup(A_{k}(S)\setminus A_{k}(U))\\
		\label{induction8c}
		&=A_{k}(U)\cup A_{k}(S).
	\end{align}
	The premises of Claim 6.1 are satisfied in (\ref{induction8a}),  (\ref{induction8d}) and  (\ref{induction8c}). 
\end{proof}
Relabeling (\ref{lemmaU2}) gives (\ref{lemmaU3}).  To justify (\ref{lemmaU4}), invoke Claim 6.4 for $l=0,\ldots,|{\bf L}_{d}|-1$.  Now (\ref{A}) implies $i\in A_{k}(S\cup U)$ if and only if $i\in A_{k}(S)\cup A_{k}(U)$.  Therefore $A_{k}(S\cup U)=A_{k}(S)\cup A_{k}(U)$.  In addition:
\begin{align}
	\nonumber
	\tilde{A}_{k}(S)\setminus A_{k}(U)&=(F_{I(d)}\cup(\cup^{k}_{q=0}L_{q})\setminus A_{k}(S))\setminus A_{k}(U),\\
	\nonumber
	&=(F_{I(d)}\cup(\cup^{k}_{q=0}L_{q}))\setminus(A_{k}(S)\cup A_{k}(U)),\\
	\nonumber
	&=(F_{I(d)}\cup(\cup^{k}_{q=0}L_{q}))\setminus A_{k}(S\cup U),\\
	\label{induction8e}
	&=\tilde{A}_{k}(S\cup U),
\end{align} 
where the equalities follow from normal set-theoretic operations.  Applying (\ref{induction8e}) to (\ref{lemmaU4}) gives (\ref{lemmaU5}), which proves Lemma \ref{induction} for Case 1.

%Both $A^{\prime}_{l}(\cdot)$ and $\tilde{A}^{\prime}_{l}(\cdot)$ in (\ref{lemmaU1}) are defined with respect to $(F,\bar{L}^{\prime}_{d})$, whereas $A_{l}(\cdot)$ and $\tilde{A}_{l}(\cdot)$ in (\ref{lemmaU1}) and  (\ref{lemmaU2}) are defined with respect to $(F,\bar{L}_{d})$.  %To see why (\ref{lemmaU1}) leads to (\ref{lemmaU2}), recognize that $A_{l}(U)\subseteq\tilde{A}^{\prime}_{l}(S)$ since $S$ and $U$ are disjoint.  

\hspace{-3.5mm}{\bf Case 2:}  $S\cap U\neq\{\}$.  

Since $(S\setminus U)\cap U=\{\}$, Case 1 implies:
\begin{align}
\label{lemmaVprep1}
	R_{S\setminus U}&<\sum^{|{\bf L}^{\prime}_{d}|-1}_{l=0}I(X_{A^{\prime}_{l}(S\setminus U)};Y_{d}|X_{\tilde{A}^{\prime}_{l}(S\setminus U)}).
\end{align}
We will prove the following sequence of inequalities:
\begin{align}
		\nonumber
		&\hspace{-2.7mm}R_{S\cap U}\\
		\label{lemmaV1}
		&\hspace{-2.5mm}>\sum^{|{\bf L}^{\prime}_{d}|-1}_{l=0}I(X_{A^{\prime}_{l}(S\cap U)\setminus A^{\prime}(S\setminus U)};Y_{d}|X_{\tilde{A}^{\prime}_{l}(S\cap U)\setminus A^{\prime}_{l}(S\setminus U)}),\\
		\label{lemmaV2}
		&\hspace{-2.5mm}\geq\sum^{|{\bf L}_{d}|-1}_{l=0}I(X_{A_{l}(S\cap U)};Y_{d}|X_{\tilde{A}_{l}(S\cap U)}).
\end{align}
Since $S\subseteq({\cal S}(d)\setminus U)\cup Z$ by hypothesis, it follows that $(S\cap U)\subseteq Z$.  Therefore (\ref{lemmaV2}) contradicts the assumption that all subsets of $Z$ satisfy (\ref{flowdecomp}) for $({\bf F},{\bf L}_{d})$.  First we show (\ref{lemmaV1}).

\hspace{-3.5mm}{\bf Claim \ref{induction}.5}   \textit{For any $S\subseteq{\cal S}(d)$:}
\begin{align}
	\nonumber
	R_{S\cap U}>\sum^{|{\bf L}^{\prime}_{d}|-1}_{l=0}I(X_{A^{\prime}_{l}(S\cap U)\setminus A^{\prime}(S\setminus U)};Y_{d}|X_{\tilde{A}^{\prime}_{l}(S\cap U)\setminus A^{\prime}_{l}(S\setminus U)})
\end{align}

\begin{proof}
Fix any $l\in\{0,\ldots,|{\bf L}^{\prime}_{d}|-1\}$ and consider the following sequence of equalities:
\begin{align}
\nonumber
&\hspace{-3mm}I(X_{A^{\prime}_{l}(S\setminus U)};Y_{d}|X_{\tilde{A}^{\prime}_{l}(S\setminus U)})\\
   \nonumber
	&\hspace{3mm}+I(X_{A^{\prime}_{l}(S\cap U)\setminus A^{\prime}_{l}(S\setminus U)};Y_{d}|X_{\tilde{A}^{\prime}_{l}(S\cap U)\setminus A^{\prime}_{l}(S\setminus U)})\\
	\label{induction10}
	&\hspace{3mm}= I(X_{A^{\prime}_{l}(S\setminus U)\cup A^{\prime}_{l}(S\cap U)};Y_{d}|X_{\tilde{A}^{\prime}_{l}(S\cap U)\setminus A^{\prime}_{l}(S\setminus U)}),\\
	\label{induction11}
	&\hspace{3mm}= I(X_{A^{\prime}_{l}(S)};Y_{d}|X_{\tilde{A}^{\prime}_{l}(S)}).
\end{align}
To justify (\ref{induction10}), we invoke Claim \ref{induction}.1.  First, we verify that  the premises of Claim \ref{induction}.1 are satisfied.  Let $S_{1}:=A^{\prime}_{l}(S\setminus U)$, $S_{2}:=\tilde{A}^{\prime}_{l}(S\setminus U)$, and $S_{3}:=A^{\prime}_{l}(S\cap U)\setminus A^{\prime}_{l}(S\setminus U)$.  By inspection $S_{1}\cap S_{2}=\{\}$.  Observe that (\ref{Aprime}) implies $A^{\prime}_{l}(S\cap U)\subseteq L^{\prime}_{l}$ and (\ref{induction2}) implies $\tilde{A}^{\prime}_{l}(S\setminus U):=F_{I(d)}\cup(\cup^{l}_{q=0}L^{\prime}_{q})\setminus A^{\prime}_{l}(S\setminus U)$.  It follows that: 
\begin{align}
	\label{induction16}
	A^{\prime}_{l}(S\cap U)\setminus A^{\prime}_{l}(S\setminus U)\subseteq\tilde{A}^{\prime}_{l}(S\setminus U),
\end{align}
which implies $S_{3}\subseteq S_{2}$.  Moreover,  
\begin{align}
	\nonumber
	&\hspace{-9mm}S_{2}\setminus S_{3}\\
	\nonumber
	&\hspace{3mm}\hspace{-9mm}=\tilde{A}^{\prime}_{l}(S\setminus U)\setminus (A^{\prime}_{l}(S\cap U)\setminus A^{\prime}_{l}(S\setminus U))\\
	\nonumber
	&\hspace{3mm}\hspace{-9mm}=(F_{I(d)}\cup(\cup^{l}_{q=0}L^{\prime}_{q})\setminus A^{\prime}_{l}(S\setminus U))\\
	\nonumber
	&\hspace{3mm}\hspace{-9mm}\hspace{30mm}\setminus(A^{\prime}_{l}(S\cap U)\setminus A^{\prime}_{l}(S\setminus U)),\\
	\nonumber
	&\hspace{3mm}\hspace{-9mm}=F_{I(d)}\cup(\cup^{l}_{q=0}L^{\prime}_{q})\setminus(A^{\prime}_{l}(S\setminus U)\cup A^{\prime}_{l}(S\cap U)),\\
	\nonumber
	&\hspace{3mm}\hspace{-9mm}=(F_{I(d)}\cup(\cup^{l}_{q=0}L^{\prime}_{q})\setminus A^{\prime}_{l}(S\cap U))\setminus A^{\prime}_{l}(S\setminus U),\\
	\label{induction17}
	&\hspace{3mm}\hspace{-9mm}=\tilde{A}^{\prime}_{l}(S\cap U)\setminus A^{\prime}_{l}(S\setminus U),
\end{align}
where the equalities follow from normal set-theoretic operations and (\ref{induction2}).  Furthermore, 
\begin{align}
	\nonumber
	S_{1}\cup S_{3}&=A^{\prime}_{l}(S\setminus U)\cup(A^{\prime}_{l}(S\cap U)\setminus A^{\prime}_{l}(S\setminus U))\\
	\label{induction18}
	&=A^{\prime}_{l}(S\cap U)\cup A^{\prime}_{l}(S\setminus U).
\end{align}
Since the premises of Claim \ref{induction}.1 are satisfied in (\ref{induction16}), (\ref{induction17}), and (\ref{induction18}), invoking Claim \ref{induction}.1 proves (\ref{induction10}). 
\begin{align}
\nonumber
&\tilde{A}^{\prime}_{l}(S\cap U)\setminus A^{\prime}_{l}(S\setminus U)\\
\label{induction12}
&=(F_{I(d)}\cup(\cup^{l}_{q=0}L^{\prime}_{q})\setminus A^{\prime}_{l}(S\cap U))\setminus A^{\prime}_{l}(S\setminus U),\\
\label{induction13}
&=F_{I(d)}\cup(\cup^{l}_{q=0}L^{\prime}_{q})\setminus(A^{\prime}_{l}(S\cap U)\cup A^{\prime}_{l}(S\setminus U)),\\
\label{postinduction13}
&=\tilde{A}^{\prime}_{l}(S),
\end{align}
where (\ref{induction12}) follows from (\ref{induction2}) and (\ref{induction13}) follows from normal set-theoretic operations.  The definition of (\ref{Aprime}) implies that $i\in A^{\prime}_{l}(S)$ if and only if $i\in A^{\prime}_{l}(S\setminus U)\cup A^{\prime}_{l}(S\cap U)$.  Therefore $A^{\prime}_{l}(S)=A^{\prime}_{l}(S\setminus U)\cup A^{\prime}_{l}(S\cap U)$ which implies (\ref{postinduction13}). Therefore (\ref{induction11}) follows from (\ref{postinduction13}).  Since  $R_{S}=R_{S\setminus U}+R_{S\cap U}$, Claim \ref{induction}.4 follows from (\ref{lemmaVprep2}), (\ref{lemmaVprep1}), and  (\ref{induction11}).
\end{proof} 
We invoke Claim \ref{induction}.5 to justify (\ref{lemmaV1}).  To prove (\ref{lemmaV2}), we introduce more preliminary claims.
 
 \hspace{-3.5mm}{\bf Claim \ref{induction}.6}   \textit{For any $S\subseteq{\cal S}(d)$ and every $l=0,\ldots,|{\bf L}^{\prime}_{d}|-1$:}
\begin{align}
\nonumber
	&I(X_{A^{\prime}_{l}(S\cap U)\setminus A^{\prime}_{l}(S\setminus U)};Y_{d}|X_{\tilde{A}^{\prime}_{l}(S\cap U)\setminus A^{\prime}_{l}(S\setminus U)})\geq\\
	\label{induction19}
	&\hspace{15mm}I(X_{A_{l-1}(S\cap U)};Y_{d}|X_{\tilde{A}^{\prime}_{l}(S\setminus U)\setminus A_{l-1}(S\cap U)})
\end{align}
%(\ref{lemmaV1}) follows from (\ref{lemmaVprep2}), (\ref{lemmaVprep1}), and the following inequality for $l=0,\ldots,|{\bf L}^{\prime}|-1$:
%\begin{align}
%\nonumber
%&I(X_{A^{\prime}_{l}(S\setminus U)};Y_{d}|X_{\tilde{A}^{\prime}_{l}(S\setminus U)})\\
%   \nonumber
%	&\hspace{5mm}+I(X_{A^{\prime}_{l}(S\cap U)\setminus A^{\prime}_{l}(S\setminus U)};Y_{d}|X_{\tilde{A}^{\prime}_{l}(S\cap U)\setminus A^{\prime}_{l}(S\setminus U)})\\
%	\nonumber
%	&\hspace{5mm}= I(X_{A^{\prime}_{l}(S\setminus U)\cup A^{\prime}_{l}(S\cap U)};Y_{d}|X_{\tilde{A}^{\prime}_{l}(S\cap U)\setminus A^{\prime}_{l}(S\setminus U)})\\
%	\nonumber
%	&\hspace{5mm}= I(X_{A^{\prime}_{l}(S)};Y_{d}|X_{\tilde{A}^{\prime}_{l}(S)}).
%\end{align}
\begin{proof}
We invoke Claim \ref{induction}.2 to justify Claim \ref{induction}.6.  First we verify that the premises of Claim \ref{induction}.2 are satisfied.  Let $S_{1}:=A_{l-1}(S\cap U)$, $S_{2}:=A^{\prime}_{l}(S\cap U)\setminus A^{\prime}_{l}(S\setminus U)$, and $S_{3}:=\tilde{A}^{\prime}_{l}(S\setminus U)$.  First, we show that $S_{3}\setminus S_{2}=\tilde{A}^{\prime}_{l}(S\cap U)\setminus A^{\prime}_{l}(S\setminus U)$.  Consider the following sequence of equalities:
\begin{align}
	\nonumber
	&\tilde{A}^{\prime}_{l}(S\cap U)\setminus A^{\prime}_{l}(S\setminus U)\\
	\label{induction15}
	&\hspace{3mm}=F_{I(d)}\cup(\cup^{l}_{q=0}L^{\prime}_{q})\setminus(A^{\prime}_{l}(S\cap U)\cup A^{\prime}_{l}(S\setminus U))\\
	\nonumber
	&\hspace{3mm}=(F_{I(d)}\cup(\cup^{l}_{q=0}L^{\prime}_{q})\setminus A^{\prime}_{l}(S\setminus U))\\
	\nonumber
	&\hspace{3mm}\hspace{30mm}\setminus(A^{\prime}_{l}(S\cap U)\setminus A^{\prime}_{l}(S\setminus U)),\\
	\nonumber
	&\hspace{3mm}=\tilde{A}^{\prime}_{l}(S\setminus U)\setminus(A^{\prime}_{l}(S\cap U)\setminus A^{\prime}_{l}(S\setminus U))\\
	\label{induction20}
	&\hspace{3mm}=S_{3}\setminus S_{2}
\end{align}
where (\ref{induction15}) follows from (\ref{induction13}) and (\ref{induction20}) follows from (\ref{induction2}) and standard set-theoretic operations.  
%\begin{align}
%	\nonumber
%	&I(X_{A^{\prime}_{l}(S\cap U)\setminus A^{\prime}_{l}(S\setminus U)};Y_{d}|X_{\tilde{A}^{\prime}_{l}(S\cap U)\setminus A^{\prime}_{l}(S\setminus U)})\geq\\
%	\label{induction21}
%	&\hspace{15mm}I(X_{A_{l-1}(S\cap U)};Y_{d}|X_{\tilde{A}^{\prime}_{l}(S\setminus U)\setminus A_{l-1}(S\cap U)}). 
%\end{align}
Next we show that $S_{1}\subseteq S_{2}$.  It follows from Lemma \ref{join} that $A_{l-1}(S\cap U)\subseteq A^{\prime}_{l}(S\cap U)$.  Moreover, Lemma \ref{remain} implies $A^{\prime}_{l}(S\setminus U)\subseteq A_{l}(S\setminus U)$.  Now (\ref{A}) implies $A_{l-1}(S\cap U)\subseteq L_{l-1}$ and $A_{l}(S\setminus U)\subseteq L_{l}$.  Moreover, (L2) implies $L_{l-1}\cap L_{l}=\{\}$.  It follows that $A_{l-1}(S\cap U)\cap A^{\prime}_{l}(S\setminus U)=\{\}$.  Therefore $A_{l-1}(S\cap U)\subseteq A^{\prime}_{l}(S\cap U)\setminus A^{\prime}_{l}(S\setminus U)$, or equivalently $S_{1}\subseteq S_{2}$.  Finally, we show $S_{2}\subseteq S_{3}$.  Since (\ref{Aprime}) implies $A^{\prime}_{l}(S\cap U)\subseteq L^{\prime}_{l}$ it follows that $A^{\prime}_{l}(S\cap U)\setminus A^{\prime}_{l}(S\setminus U)\subseteq F_{I(d)}\cup(\cup^{l}_{q=0}L^{\prime}_{q})\setminus A^{\prime}_{l}(S\setminus U)$ which implies $S_{2}\subseteq S_{3}$.  Therefore Claim \ref{induction}.2 implies (\ref{induction19}).
%(Lemma \ref{join}), (\ref{induction16}) follows by substituting (\ref{induction14}) into (\ref{induction15}).  Moreover:
%From (\ref{induction13}) we have $\tilde{A}^{\prime}_{l}(S\cap U)\setminus A^{\prime}_{l}(S\setminus U)=F_{I(d)}\cup(\cup^{l}_{q=0}L^{\prime}_{q})\setminus\{A^{\prime}_{l}(S\cap U)\cup A^{\prime}_{l}(S\setminus U)\}=\{\cup^{l}_{q=0}L_{q}\}\setminus\{A^{\prime}_{l}(S\cap U)\cup A^{\prime}_{l}(S\setminus U)\cup A_{l}(U)\}$.
%Since $A_{l-1}(S\cap U)\subseteq A^{\prime}_{l}(S\cap U)$ (Lemma \ref{join}), $A^{\prime}_{l}(S\setminus U)\subseteq A_{l}(S\setminus U)$ (Lemma \ref{remain}), and $A_{l-1}(S\cap U)\cap A_{l}(S\setminus U)=\{\}$, it follows that $A_{l-1}(S\cap U)\subseteq A^{\prime}_{l}(S\cap U)\setminus A^{\prime}_{l}(S\setminus U)$. 
\end{proof} 

 \hspace{-3.5mm}{\bf Claim \ref{induction}.7}   \textit{For any $S\subseteq{\cal S}(d)$ and every $l=0,\ldots,|{\bf L}^{\prime}_{d}|-1$:}
\begin{align}
	\nonumber
	&\hspace{-5mm}I(X_{A_{l-1}(S\cap U)};Y_{d}|X_{\tilde{A}^{\prime}_{l}(S\setminus U)\setminus A_{l-1}(S\cap U)})\geq\\
	&\hspace{25mm}I(X_{A_{l-1}(S\cap U)};Y_{d}|X_{\tilde{A}_{l-1}(S\cap U)}).
\end{align}
\begin{proof}
Consider the following sequence of inequalities:
\begin{align}
	\nonumber
	&\tilde{A}^{\prime}_{l}(S\setminus U)\setminus A_{l-1}(S\cap U)\\
	\label{induction21}
	&=(F_{I(d)}\cup(\cup^{l}_{q=0}L^{\prime}_{q})\setminus A^{\prime}_{l}(S\setminus U))\setminus A_{l-1}(S\cap U),\\
	\label{induction22}
	&=F_{I(d)}\cup(\cup^{l}_{q=0}L^{\prime}_{q})\setminus (A^{\prime}_{l}(S\setminus U)\cup A_{l-1}(S\cap U)),\\
	\nonumber
	&=F_{I(d)}\cup(\cup^{l}_{q=0}L_{q})\\
	\label{induction23}
	&\hspace{17mm}\setminus(A^{\prime}_{l}(S\setminus U)\cup A_{l-1}(S\cap U)\cup A_{l}(U)),\\
	\nonumber
	&\supseteq F_{I(d)}\cup(\cup^{l}_{q=0}L_{q})\\
	\label{induction24}
	&\hspace{17mm}\setminus(A_{l}(S\setminus U)\cup A_{l-1}(S\cap U)\cup A_{l}(U)),\\
	\label{induction25}
	&\supseteq F_{I(d)}\cup(\cup^{l-1}_{q=0}L_{q})\setminus A_{l-1}(S\cap U),\\
	\label{induction26}
	&=\tilde{A}_{l-1}(S\cap U),
\end{align}
where (\ref{induction21}) follows from (\ref{induction2}), (\ref{induction22}) follows from standard set-theoretic operations, (\ref{induction23}) follows by substituting (\ref{induction14}) into (\ref{induction22}), (\ref{induction24}) follows because Lemma \ref{remain} implies $A^{\prime}_{l}(S\setminus U)\subseteq A_{l}(S\setminus U)$, (\ref{induction25}) follows because (\ref{A}) implies $A_{l}(U)\cup A_{l}(S\setminus U)\subseteq L_{l}$, and (\ref{induction26}) follows from the definition of $\tilde{A}_{l}(\cdot)$ in (\ref{induction7}).  Since removing independent conditional random variables reduces the mutual information, Claim \ref{induction}.7 follows from (\ref{induction26}).
\end{proof}
Combining Claim \ref{induction}.6 and Claim \ref{induction}.7 proves that (\ref{lemmaV2}) follows from (\ref{lemmaV1}).  Moreover, (\ref{lemmaV2}) implies:
\begin{align}
	\label{induction27}
	R_{S\cap U}&>\displaystyle\sum^{|{\bf L}^{\prime}_{d}|-1}_{l=0}I(X_{A_{l-1}(S\cap U)};Y_{d}|X_{\tilde{A}_{l-1}(S\cap U)}).
\end{align}
Observe that (\ref{shift2}) implies $|{\bf L}_{d}|\leq|{\bf L}^{\prime}_{d}|\leq|{\bf L}_{d}|+1$.  Since $ l\leq|{\bf L}_{d}^{\prime}|-1$ it follows that $l-1\leq|{\bf L}^{\prime}_{d}|-2\leq|{\bf L}_{d}|-1$.  Substituting $l^{\prime}=l-1$ into (\ref{induction27}) and relabelling $l^{\prime}=l$ yields (\ref{lemmaV2}) which completes the proof of Lemma \ref{induction}.
\end{proof}

%% file: grandefinale.tex
\begin{lemma}
	\label{grandefinale}
	${\bf R}\in{\cal R}({\bf F},{\bf L}_{d,n^{*}})$ for some $n^{*}\in\mathbb{N}$.
\end{lemma}

\begin{proof}
	The proof is by contradiction.  We prove some initial claims, but first introduce a classical definition of the limit inferior of any sequence of sets $\{S_{n}:n\in\mathbb{N}\}$.
\begin{align}
	\label{liminf}
	\liminf_{n} S_{n}&:=\cup^{\infty}_{n=1}\cap^{\infty}_{q=n}S_{q}.
\end{align} 
	
\hspace{-3.5mm}{\bf Claim \ref{grandefinale}.1} \textit{$\liminf\limits_{n} U_{n}\neq\{\}$.}
\begin{proof}
The proof is by contradiction.  Suppose $\liminf_{n}U_{n}=\{\}$ and Lemma \ref{grandefinale} is false.  If Lemma \ref{grandefinale} is false then $Z_{n}\neq{\cal S}(d)$ for all $n\in\mathbb{N}$.  For every $s\in{\cal S}(d)$, it follows from (\ref{liminf}) and the hypothesis $\liminf_{n}U_{n}=\{\}$ that some $n_{s}$ exists such that $s\notin U_{n_{s}}$.  Therefore (\ref{inductionupdate}) implies that $s\in Z_{n_{s}+1}$.  Moreover, (\ref{inductionupdate}) also implies $Z_{n}\subseteq Z_{n+1}$ for all $n\in\mathbb{N}$.  It follows that $Z_{n}={\cal S}(d)$, for all $n\geq\max_{s\in{\cal S}(d)}n_{s}$, which contradicts the hypothesis that $Z_{n}\neq{\cal S}(d)$ for all $n\in\mathbb{N}$.
\end{proof}

\hspace{-3.5mm}{\bf Claim \ref{grandefinale}.2} \textit{If $i\in{\bf v}_{n}(s)$ for some $s\in U_{n}$ then $i\in{\bf v}_{n+1}(s)$.}

\begin{proof}
There are two cases to consider.  If $v_{n+1}(s)=v_{n}(s)$ then Claim \ref{join}.1 implies $i\in{\bf v}_{n+1}(s)$.  Otherwise, if $v_{n+1}(s)\neq v_{n}(s)$ then Claim \ref{join}.2 implies $i\in{\bf v}_{n+1}(s)$.
\end{proof}
We introduce the following definitions: 
\begin{align}
\label{U}
U^{*}&:=\liminf\limits_{n} U_{n},\\
V^{*}_{n}&:=\{i\in{\bf v}_{n}(s):s\in U_{n}\setminus U^{*}\}.
\label{vstar}
\end{align}

\hspace{-3mm}{\bf Claim.\ref{grandefinale}.3} $\liminf\limits_{n}(U_{n}\setminus U^{*})=\{\}$.
\begin{proof}
The proof is by contradiction.  Suppose $\liminf_{n}(U_{n}\setminus U^{*})\neq\{\}$.  For some $s\in{\cal S}(d)$ and $n_{s}\in\mathbb{N}$, it follows that $s\in U_{n}\setminus U^{*}$ for all $n\geq n_{s}$.  Therefore $s\in\cap_{q\geq n_{s}} U_{q}$.  It follows from (\ref{liminf}) and (\ref{U}) that $s\in U^{*}$ which is a contradiction.
\end{proof}

%let $U^{*}:=\liminf\limits_{n} U_{n}$ and $V^{*}_{n}:=\{i\in{\bf v}_{n}(s):s\in U_{n}\setminus U^{*}\}$.
%\begin{align}
%\label{U}
%U^{*}&:=\liminf_{n} U_{n}
%\end{align}
%Define the sequence of sets $\{U^{*}_{n}:n\in\mathbb{N}\}$ as follows:
%\begin{align}
%	\label{ustar}
%	U^{*}_{n}:=U_{n}\setminus U^{*}.
%\end{align}
%By construction, there is no $s\in{\cal S}(d)$ such that $s\in U^{*}_{n}$ for all $n\geq n_{s}$ and some $n_{s}\in\mathbb{N}$ (if such an $s$ exists then $U\cup\{s\}\subseteq U_{n}$ for all $n\geq\max\{n_{U},n_{s}\}$ which contradicts the definition of $U$ as the largest such subset).  
%Define the sequence of sets $\{V^{*}_{n}:n\in\mathbb{N}\}$ as follows: 
%\begin{align}
%\label{vstar}
%V^{*}_{n}&:=\{i\in{\bf v}_{n}(s):s\in U^{*}_{n}\}.	
%\end{align}

\hspace{-3.5mm}{\bf Claim \ref{grandefinale}.4} \textit{$\liminf\limits_{n}V^{*}_{n}=\{\}$}.

\begin{proof}
We first prove some preliminary claims.

\hspace{-3.5mm}{\bf Claim \ref{grandefinale}.4.a} \textit{If $\text{\sc layer}_{n+1}(i)=\text{\sc layer}_{n}(i)+1$ and $i\in{\bf v}_{n}(s)$ for some $s\notin U_{n}$ then $i\notin{\bf v}_{n+1}(s)$}.

\begin{proof}
Fix $i\in{\bf v}_{n}(s)$ for some $s\notin U_{n}$.  Since $({\bf F},{\bf L}_{d,n+1})=\text{\sc shift}(({\bf F},{\bf L}_{d,n}),U_{n})$ and $s\notin U_{n}$, Claim \ref{remain}.1 implies that $\text{\sc layer}_{n+1}(j)=\text{\sc layer}_{n}(j)$ for some $j\in{\bf v}_{n}(s)$.  Therefore, Claim \ref{disjoint}.1 implies:
	\begin{align}
	\label{vnplus1}
		v_{n+1}(s)=\displaystyle\min_{\substack{j\in{\bf v}_{n}(s),\\ \text{\sc layer}_{n+1}(j)=\text{\sc layer}_{n}(j)}}k_{v_{n}(s),j}.
	\end{align}
	We have the following sequence of equalities:
	\begin{align}
		\nonumber
		&\hspace{-3mm}\text{\sc layer}_{n+1}(v_{n+1}(s))-\text{\sc layer}_{n+1}(i)\\
		\label{grandfinale32}
		&\hspace{8mm}=\text{\sc layer}_{n+1}(v_{n+1}(s))-(\text{\sc layer}_{n}(i)+1),\\
		\label{grandfinale37}
		&\hspace{8mm}=\text{\sc layer}_{n}(v_{n+1}(s))-(\text{\sc layer}_{n}(i)+1),\\		
		\nonumber
		&\hspace{8mm}=(\text{\sc layer}_{n}(v_{n}(s))-k_{v_{n}(s),v_{n+1}(s)})\\
		\label{grandfinale33}
		&\hspace{42mm}-(\text{\sc layer}_{n}(i)+1),\\
		\nonumber
		&\hspace{8mm}=\text{\sc layer}_{n}(v_{n}(s))-\text{\sc layer}_{n}(i)\\
		\label{grandfinale34}
		&\hspace{42mm}-k_{v_{n}(s),v_{n+1}(s)}-1,\\
		\label{grandfinale35}
		&\hspace{8mm}=k_{v_{n}(s),i}-k_{v_{n}(s),v_{n+1}(s)}-1,\\
		\label{grandfinale36}
		&\hspace{8mm}=k_{v_{n+1}(s),i}-1,
	\end{align}
	where (\ref{grandfinale32}) follows from the premise of Claim \ref{grandefinale}.4.a that $\text{\sc layer}_{n+1}(i)=\text{\sc layer}_{n}(i)+1$, (\ref{grandfinale37}) follows because (\ref{vnplus1}) implies $\text{\sc layer}_{n+1}(v_{n+1}(s))=\text{\sc layer}_{n}(v_{n+1}(s))$, (\ref{grandfinale33}) follows because (V5) implies $\text{\sc layer}_{n}(v_{n}(s))-\text{\sc layer}_{n}(v_{n+1}(s))=k_{v_{n}(s),v_{n+1}(s)}$ if $v_{n+1}(s)\in{\bf v}_{n}(s)$ and (\ref{vnplus1}) implies $v_{n+1}(s)\in{\bf v}_{n}(s)$, (\ref{grandfinale34}) follows by rearranging terms, and (\ref{grandfinale35}) follows from the premise of Claim \ref{grandefinale}.4.a that $i\in{\bf v}_{n}(s)$. 

To justify (\ref{grandfinale36}), it suffices to show that $i\in{\bf f}(v_{n+1}(s),d)$.  If $i\in{\bf f}(s,v_{n}(s))$ then $i\notin{\bf v}_{n}(s)$ which contradicts the premise of Claim \ref{grandefinale}.4.a that $i\in{\bf v}_{n}(s)$.  Moreover, if $i\in{\bf f}(v_{n}(s),v_{n+1}(s))$ then $i\notin{\bf v}_{n+1}(s)$ which proves Claim \ref{grandefinale}.4.a.  Therefore it suffices to assume $i\in{\bf f}(v_{n+1}(s),d)$.  It follows from (\ref{grandfinale36}) that $\text{\sc layer}_{n+1}(v_{n+1}(s))-\text{\sc layer}_{n+1}(i)\neq k_{v_{n+1}(s),i}$ which, together with (V5), implies $i\notin{\bf v}_{n+1}(s)$.
\end{proof}

\hspace{-3.5mm}{\bf Claim \ref{grandefinale}.4.b} \textit{If $\text{\sc layer}_{n+1}(i)=\text{\sc layer}_{n}(i)+1$ and $i\notin{\bf v}_{n}(s)$ for some $s\in U_{n}$ then $i\notin{\bf v}_{n+1}(s)$}.
	\begin{proof}
	There are two cases to consider.

	\hspace{-3.5mm}{\bf Case 1:} $v_{n+1}(s)=v_{n}(s)$.
	
	We have the following sequence of inequalities:
	\begin{align}
		\nonumber
		&\text{\sc layer}_{n+1}(v_{n+1}(s))-\text{\sc layer}_{n+1}(i)\\
		\label{grandfinale38}
		&\hspace{7mm}=\text{\sc layer}_{n+1}(v_{n}(s))-\text{\sc layer}_{n+1}(i),\\
		\label{grandfinale39}
		&\hspace{7mm}=(\text{\sc layer}_{n}(v_{n}(s))+1)-\text{\sc layer}_{n+1}(i),\\
		\label{postgrandfinale39}
		&\hspace{7mm}=(\text{\sc layer}_{n}(v_{n}(s))+1)-(\text{\sc layer}_{n}(i)+1),\\
		\nonumber
		&\hspace{7mm}=\text{\sc layer}_{n}(v_{n}(s))-\text{\sc layer}_{n}(i),\\
		\label{grandfinale41}
		&\hspace{7mm}\neq k_{v_{n}(s),i},\\
		\label{grandfinale42}
		&\hspace{7mm}=k_{v_{n+1}(s),i},
	\end{align}
	where (\ref{grandfinale38}) follows from the premise of Case 1 that $v_{n+1}(s)=v_{n}(s)$, (\ref{grandfinale39}) follows from (\ref{shift4}) and because $s\in U_{n}$, (\ref{postgrandfinale39}) follows from the premise of Claim \ref{grandefinale}.4.b that $\text{\sc layer}_{n+1}(i)=\text{\sc layer}_{n}(i)+1$, (\ref{grandfinale41}) follows because $i\notin{\bf v}_{n}(s)$ and (V5) implies $\text{\sc layer}_{n}(v_{n}(s))-\text{\sc layer}_{n}(i)=k_{v_{n}(s),i}$ if $i\in{\bf v}_{n}(s)$, and (\ref{grandfinale42}) follows because $v_{n+1}(s)=v_{n}(s)$ by case assumption.  It follows from (\ref{grandfinale42}) and (V5) that $i\notin{\bf v}_{n+1}(s)$.
	
	\hspace{-3.5mm}{\bf Case 2:} $v_{n+1}(s)\neq v_{n}(s)$.
	The premise of Claim \ref{grandefinale}.4.b that $s\in U_{n}$ and the premise of Case 2 that $v_{n+1}(s)\neq v_{n}(s)$, together with Claim \ref{join}.2.c imply:
	\begin{align}
	\nonumber
	&k_{v_{n+1}(s),v_{n}(s)}\\
	\label{kvnplus}
	&\hspace{7mm}=\text{\sc layer}_{n+1}(v_{n+1}(s))-\text{\sc layer}_{n+1}(v_{n}(s)).
	\end{align} 
We have the following inequalities:
	\begin{align}
		\nonumber
		&\hspace{-3mm}\text{\sc layer}_{n+1}(v_{n+1}(s))-\text	{\sc layer}_{n+1}(i)\\
		\nonumber
		&=\text{\sc layer}_{n+1}(v_{n+1}(s))-\text{\sc layer}_{n+1}(v_{n}(s))\\
		\label{grandfinale43}
		&\hspace{17mm}+\text{\sc layer}_{n+1}(v_{n}(s))-\text{\sc layer}_{n+1}(i),\\
		\nonumber
		&=k_{v_{n+1}(s),v_{n}(s)}+\text{\sc layer}_{n+1}(v_{n}(s))\\
		\label{grandfinale44}
		&\hspace{47mm}-\text{\sc layer}_{n+1}(i),\\
		\nonumber
		&=k_{v_{n+1}(s),v_{n}(s)}+(\text{\sc layer}_{n}(v_{n}(s))+1)\\
		\label{grandfinale45}
		&\hspace{47mm}-\text{\sc layer}_{n+1}(i),\\
		\nonumber
		&=k_{v_{n+1}(s),v_{n}(s)}+(\text{\sc layer}_{n}(v_{n}(s))+1)\\
		\label{postgrandfinale45}
		&\hspace{42mm}-(\text{\sc layer}_{n}(i)+1),\\
		\nonumber
		&=k_{v_{n+1}(s),v_{n}(s)}\\
		\nonumber
		&\hspace{21mm}+(\text{\sc layer}_{n}(v_{n}(s))-\text{\sc layer}_{n}(i)),\\
		\label{grandfinale47}
		&\neq k_{v_{n+1}(s),v_{n}(s)}+k_{v_{n}(s),i},\\
		\label{grandfinale48}
		&=k_{v_{n+1}(s),i},
	\end{align}
	where (\ref{grandfinale43}) follows by subtracting and adding $\text{\sc layer}_{n+1}(v_{n}(s))$ to the left side, (\ref{grandfinale44}) follows from (\ref{kvnplus}), (\ref{grandfinale45}) follows from (\ref{shift4}) and because $s\in U_{n}$, (\ref{postgrandfinale45}) follows from the premise of Claim \ref{grandefinale}.4.b that $\text{\sc layer}_{n+1}(i)=\text{\sc layer}_{n}(i)+1$, and (\ref{grandfinale47}) follows from (V5) and the premise of Claim \ref{grandefinale}.4.b that $i\notin{\bf v}_{n}(s)$.  It follows from (\ref{grandfinale48}) and (V5) that $i\notin{\bf v}_{n+1}(s)$.
	\end{proof}
	
	\hspace{-3.5mm}{\bf Claim \ref{grandefinale}.4.c} \textit{If $\text{\sc layer}_{n+1}(i)=\text{\sc layer}_{n}(i)+1$ and $i\notin{\bf v}_{n}(s)$ for some $s\notin U_{n}$ then $i\notin{\bf v}_{n+1}(s)$}.
	
	\begin{proof}
	Since $s\notin U_{n}$, it follows from Claim \ref{disjoint}.1 and Claim \ref{remain}.1 that $v_{n+1}(s)$ is defined by (\ref{vnplus1}).  Without loss of generality, assume $i\in{\bf f}(v_{n+1}(s),d)$ since the claim is trivial otherwise.  We have the following sequence of equalities:
	\begin{align}
		\nonumber
		&\text{\sc layer}_{n+1}(v_{n+1}(s))-\text{\sc layer}_{n+1}(i)\\
		\label{grandfinale49}
		&=\text{\sc layer}_{n}(v_{n+1}(s))-\text{\sc layer}_{n+1}(i),\\
		\label{grandfinale50}
		&=\text{\sc layer}_{n}(v_{n}(s))-k_{v_{n}(s),v_{n+1}(s)}-\text{\sc layer}_{n+1}(i),\\
		\nonumber
		&=\text{\sc layer}_{n}(v_{n}(s))-k_{v_{n}(s),v_{n+1}(s)}\\
		\label{grandfinale51}
		&\hspace{45mm}-(\text{\sc layer}_{n}(i)+1),\\
		\nonumber
		&=\text{\sc layer}_{n}(v_{n}(s))-\text{\sc layer}_{n}(i)\\
		\nonumber
		&\hspace{45mm}-k_{v_{n}(s),v_{n+1}(s)}-1,\\
		\label{grandfinale53}
		&\leq k_{v_{n}(s),i}-k_{v_{n}(s),v_{n+1}(s)}-1,\\
		\label{grandfinale54}
		&=k_{v_{n+1}(s),i}-1,\\
		\label{grandfinale55}
		&<k_{v_{n+1}(s),i},
	\end{align}
	where (\ref{grandfinale49}) follows because (\ref{vnplus1}) implies $\text{\sc layer}_{n+1}(v_{n+1}(s))=\text{\sc layer}_{n}(v_{n+1}(s))$, (\ref{grandfinale50}) follows from (\ref{kvnplus}), (\ref{grandfinale51}) follows from the premise of Claim \ref{grandefinale}.3.c that $\text{\sc layer}_{n+1}(i)=\text{\sc layer}_{n}(i)+1$, and (\ref{grandfinale53}) follows because (\ref{pregrandfinale3}) implies $\text{\sc layer}_{n}(v_{n}(s))-\text{\sc layer}_{n}(i)\leq k_{v_{n}(s),i}$ if $i\in{\bf f}(v_{n}(s),d)$.  Since (\ref{vnplus1}) implies $v_{n+1}(s)\in{\bf v}_{n}(s)$ and $i\in{\bf f}(v_{n+1}(s),d)$ by assumption, it follows that $i\in{\bf f}(v_{n}(s),d)$.  Furthermore (\ref{grandfinale54}) follows because $v_{n+1}(s)\in{\bf f}(v_{n}(s),i)$.  It follows from (V5) and (\ref{grandfinale55}) that $i\notin{\bf v}_{n+1}(s)$.
	\end{proof}
	
	\hspace{-3.5mm}{\bf Claim \ref{grandefinale}.4.d} \textit{Suppose $s\notin U_{n_{s}}$ for some $n_{s}\in\mathbb{N}$.  If $\text{\sc layer}_{n+1}(i)=\text{\sc layer}_{n}(i)+1$ for all $n\geq n_{s}$ then $i\notin{\bf v}_{n}(s)$ for all $n\geq n_{s}+1$.}
	
	\begin{proof}
	The proof is by induction.  First we prove the initial condition that $i\notin{\bf v}_{n_{s}+1}(s)$ if $\text{\sc layer}_{n_{s}+1}(i)=\text{\sc layer}_{n_{s}}(i)+1$ for some $s\notin U_{n_{s}}$.  There are two cases to consider for the initial condition.

First suppose $i\notin{\bf v}_{n_{s}}(s)$.  Since $s\notin U_{n_{s}}$, and $i\notin{\bf v}_{n_{s}}(s)$, it follows from Claim \ref{grandefinale}.4.c that $i\notin{\bf v}_{n_{s}+1}(s)$.  Next suppose $i\in{\bf v}_{n_{s}}(s)$.  Since $s\notin U_{n_{s}}$ and $i\in{\bf v}_{n_{s}}(s)$, it follows from Claim \ref{grandefinale}.4.a that $i\notin{\bf v}_{n_{s}+1}(s)$.

Now assume the inductive hypothesis that $i\notin{\bf v}_{n}(s)$ for some fixed $n\geq n_{s}+1$.  If $s\in U_{n}$ then Claim \ref{grandefinale}.4.b implies $i\notin{\bf v}_{n+1}(s)$.  If $s\notin U_{n}$ then Claim \ref{grandefinale}.4.c implies $i\notin{\bf v}_{n+1}(s)$.  It follows that $i\notin{\bf v}_{n}(s)$ for all $n\geq n_{s}+1$.
	\end{proof}

	We can now complete the proof of Claim \ref{grandefinale}.4.  The proof is by contradiction.  Suppose $\liminf\limits_{n}V^{*}_{n}\neq\{\}$.  For some $i\in F_{d}({\cal N})$ and $n_{i}\in\mathbb{N}$, it follows that $i\in V^{*}_{n}$ for all $n\geq n_{i}$.  Therefore (\ref{vstar}) implies $i\in v_{n}(s)$ for some $s\in U_{n}\setminus U^{*}$ and all $n\geq n_{i}$.  It follows from (\ref{shift4}) that: 
	\begin{align}
	\label{layernplus1}
	\text{\sc layer}_{n+1}(i)=\text{\sc layer}_{n}(i)+1\hspace{1mm}\text{for all}\hspace{1mm}n\geq n_{i}.	
	\end{align}
Fix $s\in U_{n_{i_{s}}}\setminus U^{*}$ such that $i\in{\bf v}_{n_{i_{s}}}(s)$ for some $n_{i_{s}}\geq n_{i}$.  Since $\liminf_{n} U_{n}\setminus U^{*}=\{\}$ from Claim \ref{grandefinale}.3, it follows that $s\notin U_{n_{s}}\setminus U^{*}$ for some $n_{s}\geq n_{i_{s}}$.  Moreover, $s\notin U^{*}$ by construction, so $s\notin U_{n_{s}}$.  Since $n_{s}\geq n_{i}$, Claim \ref{grandefinale}.4.d and (\ref{layernplus1}) imply that $i\notin{\bf v}_{n}(s)$ for all $n\geq n_{s}+1$.  Similarly, for every $s\in U_{n}\setminus U^{*}$ such that $i\in{\bf v}_{n}(s)$ for some $n\geq n_{i}$, there exists some $n_{s}\geq n_{i}$ such that $i\notin{\bf v}_{n}(s)$ for all $n\geq n_{s}+1$.  For all other $s\notin U^{*}$, we set $n_{s}:=0$.  It follows that $i\notin{\bf v}_{n}(s)$ for any $s\in U_{n}\setminus U^{*}$ and all $n\geq \max_{s}n_{s}+1$.  Therefore $i\notin V^{*}_{n}$ for all $n\geq \max_{s}n_{s}+1$ which is a contradiction.      
\end{proof}

\hspace{-3.5mm}{\bf Claim \ref{grandefinale}.5} $V_{n}(U^{*})=F_{d}(U^{*})$ $\forall n\geq K_{5}$ and some $K_{5}\in\mathbb{N}$.
\begin{proof}
%There exists some $n_{V}\in\mathbb{N}$ such that $V_{n}(U)=F_{d}(U)$ for all $n\geq n_{V}$
We first prove a preliminary claim.

\hspace{-3.5mm}{\bf Claim \ref{grandefinale}.5.a}\textit{ $\lim_{n\rightarrow\infty}{\bf v}_{n}(s)$ exists for every $s\in U^{*}$.} 

\begin{proof}
To prove Claim \ref{grandefinale}.5.a, we show for every $s\in U^{*}$, there exists some $n_{s}\in\mathbb{N}$ such that ${\bf v}_{n}(s)={\bf v}(s)$ for all $n\geq n_{s}$.

Suppose $i\in{\bf v}_{n_{i}}(s)$ for some $n_{i}\in\mathbb{N}$ and $s\in U^{*}$.  It follows from (\ref{U}), that $s\in U_{n}$ for all $n\geq n_{i}$.  Therefore Claim \ref{grandefinale}.2 implies $i\in{\bf v}_{n}(s)$ for all $n\geq n_{i}$.  Hence, $V_{n}(\{s\})\subseteq V_{n+1}(\{s\})\subseteq F_{d}(\{s\})$ for all $n\in\mathbb{N}$.  Since $\{V_{n}(\{s\}):n\in\mathbb{N}\}$ is a monotonically increasing sequence of finite sets, ordered by inclusion with bounded size, there exists some $V(\{s\})\subseteq F_{d}(\{s\})$ and $n_{s}\in\mathbb{N}$ such that: 
\begin{align}
	\label{vdn}
	V_{n}(\{s\})=V(\{s\})\hspace{1mm}\text{for all}\hspace{1mm}n\geq n_{s}.
\end{align}
By construction in (V1)-(V5), both ${\bf v}_{n+1}(s)$ and ${\bf v}_{n}(s)$ are subsequences of ${\bf f}(s,d)$.  Since (\ref{vdn}) implies $V_{n}(\{s\})=V_{n+1}(\{s\})=V(\{s\})$, it follows that ${\bf v}_{n+1}(s)$ and ${\bf v}_{n}(s)$ have the same nodes.  Therefore ${\bf v}_{n+1}(s)={\bf v}_{n}(s)={\bf v}(s)$ for all $n\geq n_{s}$. 
\end{proof}

%\hspace{-3.5mm}{\bf Claim \ref{grandefinale}.5.b} \textit{$\lim_{n\rightarrow\infty}V_{n}(U^{*})$ exists.}

%\textit{There exists some $V(U)\subseteq F_{d}(U)$ such that $V_{n}(U)=V(U)$ for all $n\geq n_{V}$ and some $n_{V}\in\mathbb{N}$.}

%\begin{proof}
%By definition, $V_{n}(U^{*}):=\{i\in{\bf v}_{n}(s):s\in U^{*}\}$.  The proof follows from Claim \ref{grandefinale}.5.a.
%To prove Claim \ref{grandefinale}.5.b, we show that $V_{n}(U^{*})=V(U^{*})$ for some  $V(U^{*})\subseteq F_{d}(U^{*})$ and all $n\geq n_{V}$, where $n_{V}\in\mathbb{N}$.
%If $i\in V_{n_{i}}(U^{*})$ for some $n_{i}\in\mathbb{N}$, then $i\in{\bf v}_{n_{i}}(s)$ for some $s\in U^{*}$.  It follows from (\ref{U}), that $s\in U_{n}$ for all $n\geq n_{i}$.  Therefore Claim \ref{grandefinale}.2 implies $i\in{\bf v}_{n}(s)$ for all $n\geq n_{i}$.  It follows that $V_{n}(U^{*})\subseteq V_{n+1}(U^{*})\subseteq F_{d}(U^{*})$ for all $n\in\mathbb{N}$, so $\{V_{n}(U^{*}):n\in\mathbb{N}\}$ is a monotonically increasing sequence of universally bounded finite sets, with inclusion acting as partial ordering. Hence, there exists some $n_{V}\in\mathbb{N}$ and $V(U^{*})\subseteq F_{d}(U^{*})$ such that:
 %\begin{align}
 %		V_{n}(U^{*})=V(U^{*})\hspace{1mm}\text{for all}\hspace{1mm}n\geq n_{V},
 %	\end{align} 
 %	which completes the proof of Claim \ref{grandefinale}.5.a.
%\end{proof}

We can now complete the proof of Claim \ref{grandefinale}.5.  The proof is by contradiction.  For every $s\in U^{*}$, Claim \ref{grandefinale}.5.a implies ${\bf v}_{n}(s)={\bf v}(s)$ for all $n\geq n_{s}$ and some $n_{s}\in\mathbb{N}$.  Define $V(U^{*}):=\{i\in{\bf v}(s):s\in U^{*}\}$.  It follows that $V_{n}(U^{*})=V(U^{*})$ for all $n\geq\max_{s\in U^{*}}n_{s}$.  Therefore $\lim_{n\rightarrow\infty}V_{n}(U^{*})=V(U^{*})$.  Suppose $V(U^{*})\neq F_{d}(U)$, and pick some $i\in F_{d}(U^{*})\setminus V(U^{*})$.  Now $i\in{\bf f}(s,d)$ for some $s\in U^{*}$.  There are two cases to consider.

\hspace{-3.5mm}{\bf Case 1:} $i\in{\bf f}(v(s),d)$.

For all $n\in\mathbb{N}$:
\begin{align}
	\label{grandfinale61}
	\text{\sc layer}_{n}(v_{n}(s))-\text{\sc layer}_{n}(i)\leq k_{v_{n}(s),i},
\end{align}
where (\ref{grandfinale61}) follows from (\ref{pregrandfinale3}).  Claim \ref{grandefinale}.5.a implies ${\bf v}_{n}(s)={\bf v}(s)$ for all $n\geq n_{s}$.  For all $n\geq n_{s}$, Claim \ref{grandefinale}.5.a and (\ref{grandfinale61}) imply:
\begin{align}
	\label{grandfinale62}
	\text{\sc layer}_{n}(v(s))-\text{\sc layer}_{n}(i)\leq k_{v(s),i}.
\end{align} 
For all $n\in\mathbb{N}$:
\begin{align}
	\label{grandfinale63}
	\text{\sc layer}_{n}(i)\leq\text{\sc layer}_{n+1}(i)\leq\text{\sc layer}_{n}(i)+1,
\end{align}
where (\ref{grandfinale63}) follows from (\ref{shift4}).  The definition of $U^{*}$ in (\ref{U}) implies that $U^{*}\subseteq U_{n}$ for all $n\geq K_{1}$ and some $K_{1}\in\mathbb{N}$.  For all $n\geq\max\{n_{s},K_{1}\}$:
\begin{align}
	\label{grandfinale64}
	\text{\sc layer}_{n+1}(v(s))=\text{\sc layer}_{n}(v(s))+1.
\end{align}
To justify (\ref{grandfinale64}), observe that Claim \ref{grandefinale}.5.a implies  $v_{n}(s)=v(s)$ for all $n\geq n_{s}$.  Therefore (\ref{grandfinale64}) follows from (\ref{shift4}) and because $U^{*}\subseteq U_{n}$ for all $n\geq K_{1}$.  

Define the sequence $a_{n}:=\text{\sc layer}_{n}(v(s))-\text{\sc layer}_{n}(i)$.  Both   (\ref{grandfinale63}) and (\ref{grandfinale64}) imply $a_{n}$ is monotonically non-decreasing for all $n\geq\max\{n_{s}, K_{1}\}$.  We will show $a_{n}$ is actually unbounded from above, which contradicts (\ref{grandfinale62}).  

Since $\liminf_{n}V^{*}_{n}=\{\}$ from Claim \ref{grandefinale}.4, there exists an infinite sequence $\{n_{q}:q\in\mathbb{N}\}$ where $n_{0}:=\max\{n_{s},K_{1}\}$ such that $i\notin V^{*}_{n_{q}}$.  It follows from (\ref{vstar}) that $i\notin{\bf v}_{n_{q}}(s^{\prime})$ for all $s^{\prime}\in U_{n_{q}}\setminus U^{*}$ and $q\in\mathbb{N}$.  Moreover, since $i\in F_{d}(U)\setminus V(U)$, it follows that $i\notin{\bf v}(s^{\prime})$ for all $s^{\prime}\in U^{*}$.  Hence, $i\notin{\bf v}(s^{\prime})$ for any $s^{\prime}\in U_{n_{q}}$.  It follows from (\ref{shift4}) that:
\begin{align}
	\label{grandfinale65}
	\text{\sc layer}_{n_{q}+1}(i)=\text{\sc layer}_{n_{q}}(i)\hspace{1mm}\text{for all}\hspace{1mm}q\in\mathbb{N}.
\end{align}
For every $q\in\mathbb{N}$:
\begin{align}
\label{grandfinale66}
a_{n_{q}+1}&:=\text{\sc layer}_{n_{q}+1}(v_{n_{q}+1}(s))-\text{\sc layer}_{n_{q}+1}(i),\\
\label{grandfinale67}
&=(\text{\sc layer}_{n_{q}}(v_{n_{q}}(s))+1)-\text{\sc layer}_{n_{q}+1}(i),\\
\label{postgrandfinale67}
&=(\text{\sc layer}_{n_{q}}(v_{n_{q}}(s))+1)-\text{\sc layer}_{n_{q}}(i),\\
\nonumber
&=\text{\sc layer}_{n_{q}}(v_{n_{q}}(s))-\text{\sc layer}_{n_{q}}(i)+1,\\
\label{grandfinale68}
&=a_{n_{q}}+1,
\end{align}
where (\ref{grandfinale66}) follows from the definition of $a_{n}$, (\ref{grandfinale67}) follows from (\ref{grandfinale64}) and because $n_{0}:=\max\{n_{s},K_{1}\}$, (\ref{postgrandfinale67}) follows from (\ref{grandfinale65}), and (\ref{grandfinale68}) follows from the definition of $a_{n}$.  It follows from (\ref{grandfinale68}), that $a_{n}$ is unbounded from above.  Therefore $a_{n}$ is monotonically non-decreasing and unbounded from above which contradicts (\ref{grandfinale62}).  It follows that $i\in V(U^{*})$ if $i\in{\bf f}(v(s),d)$.

\hspace{-3.5mm}{\bf Case 2:} $i\in{\bf f}(s,v(s))$.

For all $n\in\mathbb{N}$:
\begin{align}
	\label{grandfinale69}
	\text{\sc layer}_{n}(i)-\text{\sc layer}_{n}(v_{n}(s))>k_{i,v_{n}(s)},
\end{align}
where (\ref{grandfinale69}) follows from (\ref{pregrandfinale4}).  Claim \ref{grandefinale}.5.a implies ${\bf v}_{n}(s)={\bf v}(s)$ for all $n\geq n_{s}$ and some $n_{s}\in\mathbb{N}$.  Therefore: 
\begin{align}
	\label{grandfinale70}
	\text{\sc layer}_{n}(i)-\text{\sc layer}_{n}(v(s))>k_{i,v(s)}.
\end{align}
where (\ref{grandfinale70}) follows from (\ref{grandfinale69}) and Claim \ref{grandefinale}.5.a. 

Define the sequence $b_{n}:=\text{\sc layer}_{n}(i)-\text{\sc layer}_{n}(v(s))$.  Both (\ref{grandfinale63}) and (\ref{grandfinale64}) imply $b_{n}$ is monotonically non-increasing for all $n\geq\max\{n_{s}, K_{1}\}$.  We will show that $b_{n}$ is actually unbounded from below which contradicts (\ref{grandfinale70}).  

As before, we construct an infinite sequence $\{n_{q}:q\in\mathbb{N}\}$ where $n_{0}:=\max\{n_{s},K_{1}\}$ such that $i\notin V^{*}_{n_{q}}$.  For every $q\in\mathbb{N}$: 
\begin{align}
	\label{grandfinale71}
	b_{n_{q}+1}&:=\text{\sc layer}_{n_{q}+1}(i)-\text{\sc layer}_{n_{q}+1}(v_{n_{q}+1}(s)),\\
	\label{grandfinale72}
	&=\text{\sc layer}_{n_{q}+1}(i)-(\text{\sc layer}_{n_{q}}(v_{n_{q}}(s))+1),\\
	\label{postgrandfinale72}
	&=\text{\sc layer}_{n_{q}}(i)-(\text{\sc layer}_{n_{q}}(v_{n_{q}}(s))+1),\\
	\nonumber
	&=\text{\sc layer}_{n_{q}}(i)-\text{\sc layer}_{n_{q}}(v_{n_{q}}(s))-1,\\
	\label{grandfinale74}
	&=b_{n_{q}}-1,
\end{align}
where (\ref{grandfinale71}) follows from the definition of $b_{n}$, (\ref{grandfinale72}) follows from (\ref{grandfinale64}) and the fact that $n_{0}:=\max\{n_{s},n_{U}\}$,  (\ref{postgrandfinale72}) follows from (\ref{grandfinale65}), and (\ref{grandfinale74}) follows from the definition of $b_{n}$.  It follows from (\ref{grandfinale74}), that $b_{n}$ is unbounded from below.  Therefore $b_{n}$ is monotonically non-increasing and unbounded from below, which contradicts (\ref{grandfinale70}).  It follows that $i\in V(U^{*})$ if $i\in{\bf f}(s,v(s))$. 

Case 1 and Case 2 imply $V(U^{*})=F_{d}(U^{*})$ which completes the proof of Claim \ref{grandefinale}.5. 
\end{proof}

\hspace{-3.5mm}{\bf Claim \ref{grandefinale}.6} \textit{For some function $l(n)$, $K_{6}\in\mathbb{N}$, and all $n\geq K_{6}$:}
\begin{align}
	\label{FDU}
	F_{d}(U^{*})&\subseteq\displaystyle\cup^{|{\bf L}_{d,n}|-1}_{q=l(n)}L_{q,n},\\
	\label{FDU2}
	F_{d}({\cal N})\setminus F_{d}(U^{*})&\subseteq\cup^{l(n)}_{q=0}L_{q,n},\\
	\label{FDU3}
	V_{n}(U_{n}\setminus U^{*})&\subseteq\cup^{l(n)}_{q=0}L_{q,n}.
	%\label{FDN}
	%F_{d}({\cal N})\setminus F_{d}(U^{*})&=\lim_{n\rightarrow\infty}\displaystyle\cup^{l(n)}_{q=0}L_{q,n}.
\end{align}
\begin{proof}
%\textit{For some $n_{F}\in\mathbb{N}$ and every $n\geq n_{F}$, there exists some $l(n)$ such that:}
We introduce some preliminary notation.  For all $i,j\in F_{d}({\cal N})$ define:
\begin{align}
	\label{enij}
	e_{n}(i,j):=\text{\sc layer}_{n}(i)-\text{\sc layer}_{n}(j).
\end{align}
To prove (\ref{FDU}) and (\ref{FDU2}), it suffices to establish the following limits.  If $i,j\in F_{d}(U^{*})$ then:
\begin{align}
	\label{limenij1}
	\lim_{n\rightarrow\infty}e_{n}(i,j)&=e(i,j),
\end{align}
where $e(i,j)$ is a constant with respect to $n$.  If $i\in F_{d}(U^{*})$ and $j\in F_{d}({\cal N})\setminus F_{d}(U^{*})$ then:
\begin{align}
	\label{limenij2}
	\lim_{n\rightarrow\infty}e_{n}(i,j)&=\infty.
\end{align}
Both (\ref{limenij1}) and (\ref{limenij2}) also imply (\ref{FDU3}).  To justify (\ref{FDU3}), suppose by contradiction that $i\in F_{d}(U^{*})\cap V_{n}(U_{n}\setminus U^{*})$ and $j\in V_{n}(U_{n}\setminus U^{*})\setminus F_{d}(U^{*})$ for some $s\in U_{n}\setminus U^{*}$ and $i,j\in v_{n}(s)$.  Then,
\begin{align}
	\nonumber
	&\hspace{-6mm}\text{\sc layer}_{n}(i)-\text{\sc layer}_{n}(j)\\
	\nonumber
	&\hspace{7mm}=(\text{\sc layer}_{n}(v_{n}(s))-\text{\sc layer}_{n}(j)),\\
	\label{FDU4}
	&\hspace{13mm}-(\text{\sc layer}_{n}(v_{n}(s))-\text{\sc layer}_{n}(i)),\\
	\label{FDU5}
	&\hspace{7mm}=k_{v_{n}(s),j}-k_{v_{n}(s),i},\\
	\nonumber
	&\hspace{7mm}\leq k_{v_{n}(s),j},\\
	\label{FDU6}
	&\hspace{7mm}\leq k_{s,j},
\end{align}
where (\ref{FDU4}) follows by adding and subtracting $\text{\sc layer}_{n}(v_{n}(s))$, (\ref{FDU5}) follows from (V5) and the assumption that $i,j\in v_{n}(s)$, and (\ref{FDU6}) follows from the definition of $k_{v_{n}(s),j}$ in  (V5).  The bound in (\ref{FDU6}) contradicts (\ref{limenij2}).  Therefore  (\ref{limenij1}) and (\ref{limenij2}) implies (\ref{FDU})-(\ref{FDU3}).  It remains to prove (\ref{limenij1}) and (\ref{limenij2}). 

%For every $n\in\mathbb{N}$, we have the following sequence of equalities:
%\begin{align}
%	\nonumber
%	&F_{d}({\cal N})\\
%	\label{grandfinale77}
%	&\hspace{3mm}=V_{n}({\cal N}),\\
%	\nonumber
%	&\hspace{3mm}=[F_{d}({\cal N})\setminus V_{n}(U_{n})]\cup V_{n}(U_{n}),\\
%	\nonumber
%	&\hspace{3mm}=[F_{d}({\cal N})\setminus V_{n}(U_{n})]\cup[V_{n}(U)\cup V_{n}(U_{n}\setminus U)],\\
%	\label{grandfinale80}
%	&\hspace{3mm}=[(F_{d}({\cal N})\setminus V_{n}(U_{n}))\cup V_{n}(U_{n}\setminus U)]\cup V_{n}(U),
%\end{align}
%where (\ref{grandfinale77}) follows because $({\bf F},{\bf L}_{d,0})$ is complete (by construction in Lemma \ref{completeexists}), so Lemma \ref{complete} implies that $({\bf F},{\bf L}_{d,n})$ is complete for all $n\in\mathbb{N}$.  It follows from (\ref{grandfinale80}) that $i\in F_{d}({\cal N})$ implies $i\in (F_{d}({\cal N})\setminus V_{n}(U))\cup V_{n}(U_{n}\setminus U)\cup V_{n}(U)$.  
The definition of $U^{*}$ in (\ref{U}) implies that $U^{*}\subseteq U_{n}$ for all $n\geq K_{1}$ and some $K_{1}\in\mathbb{N}$.  For all $i\in V_{n}(U^{*})$ and $n\geq K_{1}$, it follows that:
\begin{align}
	\label{grandfinale81}
	\text{\sc layer}_{n+1}(i)&=\text{\sc layer}_{n}(i)+1,
\end{align}
where (\ref{grandfinale81}) follows from (\ref{shift4}) and the fact that $U^{*}\subseteq U_{n}$ for all $n\geq K_{1}$.

First, we justify (\ref{limenij1}).  Claim \ref{grandefinale}.5, implies that $V_{n}(U^{*})=F_{d}(U^{*})$ for all $n\geq K_{5}$ and some $K_{5}\in\mathbb{N}$.  Fix $i,j\in F_{d}(U^{*})$. For all $n\geq K_{5}$:
\begin{align}
	\label{grandfinale84}
	e_{n+1}(i,j)&=\text{\sc layer}_{n+1}(i)-\text{\sc layer}_{n+1}(j),\\
	\label{grandfinale85}
	&=(\text{\sc layer}_{n}(i)+1)+(\text{\sc layer}_{n}(j)+1),\\
	\nonumber
	&=\text{\sc layer}_{n}(i)-\text{\sc layer}_{n}(j),\\
	\label{grandfinale87}
	&=e_{n}(i,j)
	%\label{grandfinale88}
	%&=e_{n_{V}}(i,j),
\end{align}
where (\ref{grandfinale84}) follows from (\ref{enij}), (\ref{grandfinale85}) follows from (\ref{grandfinale81}) and because $F_{d}(U^{*})=V_{n}(U^{*})$ for $n\geq K_{5}$, and (\ref{grandfinale87}) follows from (\ref{enij}). Set $e(i,j):=e_{n}(i,j)$ for $n=K_{5}$.  Now (\ref{grandfinale87}) implies that $e_{n}(i,j)=e(i,j)$ for all $n\geq K_{5}$, which proves (\ref{limenij1}).

Now we justify (\ref{limenij2}).  Fix $i\in F_{d}(U^{*})$ and $j\notin F_{d}(U^{*})$.  Invoking Claim \ref{grandefinale}.4, we construct an infinite sequence $\{n_{q}:q\in\mathbb{N}\}$ where $n_{0}=K_{5}$ such that $j\notin V^{*}_{n_{q}}$.  First we show that $j$ satisfies:
\begin{align}
	\label{grandfinale83}
	\text{\sc layer}_{n_{q}+1}(j)=\text{\sc layer}_{n_{q}}(j)\hspace{2mm}\text{for all}\hspace{1mm}q\in\mathbb{N}.
\end{align}
Since $V_{n_{q}}(U^{*})\subseteq F_{d}(U^{*})$ for all $q\in\mathbb{N}$, it follows that $j\notin V_{n_{q}}(U^{*})$.  If $j\notin V_{n_{q}}(U^{*})$ and $j\notin V^{*}_{n_{q}}$, it follows that $j\notin {\bf v}_{n_{q}}(s)$ for any $s\in U^{*}$ and, as implied by (\ref{vstar}), any $s\in U_{n_{q}}\setminus U^{*}$.  Therefore $j\notin{\bf v}_{n_{q}}(s)$ for any $s\in U_{n_{q}}$, so (\ref{grandfinale83}) follows from (\ref{shift4}). 
For all $n\geq K_{5}$:
\begin{align}
	\label{grandfinale93}
	e_{n_{q}+1}(i,j)&=\text{\sc layer}_{n_{q}+1}(i)-\text{\sc layer}_{n_{q}+1}(j),\\
	\label{grandfinale94}
	&=(\text{\sc layer}_{n_{q}}(i)+1)-\text{\sc layer}_{n_{q}+1}(j),\\
	\label{pregrandfinale95}
	&=(\text{\sc layer}_{n_{q}}(i)+1)-\text{\sc layer}_{n_{q}}(j),\\
	\nonumber
	&=\text{\sc layer}_{n_{q}}(i)-\text{\sc layer}_{n_{q}}(j)+1,\\
	\label{grandfinale96}
	&=e_{n_{q}}(i,j)+1,
\end{align}
where (\ref{grandfinale93}) follows from (\ref{enij}), (\ref{grandfinale94}) follows from (\ref{grandfinale81}) and because $F_{d}(U^{*})=V_{n}(U^{*})$ for $n\geq K_{5}$, (\ref{pregrandfinale95}) follows from (\ref{grandfinale83}), and (\ref{grandfinale96}) follows from (\ref{enij}).  More generally, for all $n\geq K_{5}$:
\begin{align}
	\label{grandfinale97}
	e_{n+1}(i,j)&=\text{\sc layer}_{n+1}(i)-\text{\sc layer}_{n+1}(j),\\
	\label{grandfinale98}
	&=(\text{\sc layer}_{n}(i)+1)-\text{\sc layer}_{n+1}(j),\\
	\label{grandfinale99}
	&\geq(\text{\sc layer}_{n}(i)+1)-(\text{\sc layer}_{n}(j)+1),\\
	\nonumber
	&=\text{\sc layer}_{n}(i)-\text{\sc layer}_{n}(j),\\
	\label{grandfinale100}
	&=e_{n}(i,j),
\end{align}
where (\ref{grandfinale97}) follows from (\ref{enij}), (\ref{grandfinale98}) follows because (\ref{grandfinale81}) applies for $n\geq K_{1}$ and because $F_{d}(U^{*})=V_{n}(U^{*})$ for $n\geq K_{5}$, (\ref{grandfinale99}) follows because (\ref{shift4}) implies $\text{\sc layer}_{n}(j)\leq\text{\sc layer}_{n+1}(j)\leq\text{\sc layer}_{n}(j)+1$, and (\ref{grandfinale100}) follows from (\ref{enij}).  Combining (\ref{grandfinale96}) and (\ref{grandfinale100}) implies $\lim_{n\rightarrow\infty}e_{n}(i,j)=\infty$ if $i\in F_{d}(U^{*})$ and $j\in F_{d}({\cal N})\setminus F_{d}(U^{*})$, which proves (\ref{limenij2}).
\end{proof}

We can now complete the proof of Lemma \ref{grandefinale}.  The proof is by contradiction.  If the Lemma is false, then Claim \ref{grandefinale}.1 implies there exists some non-empty $U^{*}$ defined by (\ref{U}) for all $n\geq K_{1}$ and some $K_{1}\in\mathbb{N}$.  For some function $l(n)$, all $q\in\{l(n)+1,\ldots,|{\bf L}_{d,n}|-1\}$, all $n\geq K_{6}$ and some $K_{6}\in\mathbb{N}$, Claim \ref{grandefinale}.5 and \ref{grandefinale}.6 imply:   
%Claim \ref{grandefinale}.5 implies $V_{n}(U^{*})=F_{d}(U^{*})$ for all $n\geq n_{5}$ and some $n_{5}\in\mathbb{N}$.  Claim \ref{grandefinale}.6 implies that $F_{d}(U^{*})=\cup^{|{\bf L}_{d,n}|-1}_{q=l(n)+1}L_{q,n}$ for all $n\geq n_{6}$ and some $n_{6}\in\mathbb{N}$.  For all $n\geq\max\{n_{5},n_{6}\}$ it follows that:
%\begin{align}
%	\label{FDU2}
%	F_{d}(U^{*})&=\cup^{|{\bf L}_{d,n}|-1}_{q=l(n)+1}A_{q,n}(U^{*}).
%\end{align}
%For all $n\geq n_{6}$, Claim \ref{grandefinale}.6 implies that:
%\begin{align}
%	\label{FDN2}
%	F_{d}({\cal N})\setminus F_{d}(U^{*})&=\cup^{l(n)}_{q=0}L_{q,n}.
%\end{align}
\begin{align}
	\label{aln2}
	A_{q,n}(U^{*})&=L_{q,n}.
\end{align}
To justify (\ref{aln2}), suppose by contradiction, there is some $i\in L_{q,n}$ and $q\in\{l(n)+1,\ldots,|{\bf L}_{d,n}|-1\}$ such that $i\notin A_{q,n}(U^{*})$ for some $n\geq K_{6}$.  It follows from (\ref{FDU}) and (\ref{FDU2}) that $i\in F_{d}(U^{*})$.  But $i\notin A_{q,n}(U^{*})$ implies $i\notin v_{n}(s)$ for all $s\in U^{*}$ which contradicts Claim \ref{grandefinale}.5.  

For all $q\in\{l(n)+1,\ldots,|{\bf L}_{d,n}|-1\}$ and $n\geq K_{6}$: 
\begin{align}
	\nonumber
	&\tilde{A}_{q,n}(U^{*})\\
	\label{newgrandefinale2}
	&=F_{I(d)}\cup(\cup^{q}_{k=0}L_{k,n})\setminus A_{q,n}(U^{*}),\\
	\nonumber
	&=F_{I(d)}\cup(\cup^{l(n)}_{k=0}L_{k,n})\\
	\label{newgrandefinale3}
	&\hspace{29mm}\cup(\cup^{q}_{k=l(n)+1}L_{k,n})\setminus A_{q,n}(U^{*}),\\
	\nonumber
	&=F_{I(d)}\cup(F_{d}({\cal N})\setminus F_{d}(U^{*}))\\
	\label{newgrandefinale4}
	&\hspace{29mm}\cup(\cup^{q}_{k=l(n)+1}L_{k,n})\setminus A_{q,n}(U^{*}),\\
	\label{newgrandefinale5}
	&=\tilde{F}_{d}(U^{*})\cup(\cup^{q}_{k=l(n)+1}L_{k,n})\setminus A_{q,n}(U^{*}),\\
	\label{newgrandefinale6}
	&=\tilde{F}_{d}(U^{*})\cup(\cup^{q-1}_{k=l(n)+1}A_{k,n}(U^{*})),
\end{align}
where (\ref{newgrandefinale2}) follows from (\ref{alntilde}), (\ref{newgrandefinale3}) follows because $q\in\{l(n)+1,\ldots,|{\bf L}_{d,n}|-1\}$ by assumption, (\ref{newgrandefinale4}) follows from (\ref{FDU2}), (\ref{newgrandefinale5}) follows from the definition of $\tilde{F}_{d}(U^{*})$, and (\ref{newgrandefinale6}) follows from (\ref{aln2}).  

For every $s\notin U^{*}$, Claim \ref{grandefinale}.3 implies $s\notin U_{n_{s}}\setminus U^{*}$ for some $n_{s}\in\mathbb{N}$.  Therefore $s\notin U_{n_{s}}$.  Lemma \ref{induction} implies $s\in Z_{n_{s}+1}$.  Define $K_{7}:=\max_{s\notin U^{*}}n_{s}+1$.  Since Lemma \ref{induction} also implies $Z_{n}\subseteq Z_{n+1}$ for all $n\in\mathbb{N}$, it follows that $({\cal S}(d)\setminus U^{*})\subseteq Z_{n}$ for all $n\geq K_{7}$.  For all $n\geq\max\{K_{6},K_{7}\}$:
\begin{align}
	\nonumber
	&R_{U_{n}\setminus U^{*}}\\
	\label{newgrandfinale12}
	&\hspace{4mm}<\sum^{|{\bf L}_{d,n}|-1}_{q=0}I(X_{A_{q,n}(U_{n}\setminus U^{*})};Y_{d}|X_{\tilde{A}_{q,n}(U_{n}\setminus U^{*})}),\\
	\label{newgrandfinale13}
	&\hspace{4mm}=\sum^{l(n)}_{q=0}I(X_{A_{q,n}(U_{n}\setminus U^{*})};Y_{d}|X_{\tilde{A}_{q,n}(U_{n}\setminus U^{*})}),
\end{align}
where (\ref{newgrandfinale12}) follows because $(U_{n}\setminus U^{*})\subseteq({\cal S}(d)\setminus U^{*})\subseteq Z_{n}$ and all subsets of $Z_{n}$ satisfy (\ref{flowdecomp3}) by definition, and (\ref{newgrandfinale13}) follows from (\ref{FDU3}).  For all $n\geq\max\{K_{6},K_{7}\}$:
\begin{align}
	\label{newgrandfinale7}
	R_{U_{n}}&>\sum^{|{\bf L}_{d,n}|-1}_{q=0}I(X_{A_{q,n}(U_{n})};Y_{d}|X_{\tilde{A}_{q,n}(U_{n})}),\\
	\nonumber
	&=\sum^{l(n)}_{q=0}I(X_{A_{q,n}(U_{n})};Y_{d}|X_{\tilde{A}_{q,n}(U_{n})})\\
	\nonumber
	&\hspace{4mm}+\sum^{|{\bf L}_{d,n}|-1}_{q=l(n)+1}I(X_{A_{q,n}(U_{n})};Y_{d}|X_{\tilde{A}_{q,n}(U_{n})}),\\
	\nonumber
	&=\sum^{l(n)}_{q=0}I(X_{A_{q,n}(U_{n})};Y_{d}|X_{\tilde{A}_{q,n}(U_{n})})\\
	\label{newgrandfinale9}
	&\hspace{5mm}+\sum^{|{\bf L}_{d,n}|-1}_{q=l(n)+1}I(X_{A_{q,n}(U^{*})};Y_{d}|X_{\tilde{A}_{q,n}(U^{*})}),\\
	\nonumber
	&=\sum^{l(n)}_{q=0}I(X_{A_{q,n}(U_{n})};Y_{d}|X_{\tilde{A}_{q,n}(U_{n})})\\
	\label{newgrandfinale10}
	&\hspace{23mm}+I(X_{F_{d}(U^{*})};Y_{d}|X_{\tilde{F}_{d}(U^{*})}),\\
	\nonumber
	&=\sum^{l(n)}_{q=0}I(X_{A_{q,n}(U_{n}\setminus U^{*})};Y_{d}|X_{\tilde{A}_{q,n}(U_{n}\setminus U^{*})})\\
	\label{newgrandfinale11}
	&\hspace{23mm}+I(X_{F_{d}(U^{*})};Y_{d}|X_{\tilde{F}_{d}(U^{*})}),
\end{align}
where (\ref{newgrandfinale7}) follows from the definition of $U_{n}$, (\ref{newgrandfinale9}) follows from (\ref{aln2}), (\ref{newgrandfinale10}) follows from (\ref{FDU2}), and (\ref{newgrandfinale11}) follows from Claim \ref{grandefinale}.5, (\ref{FDU}) and (\ref{FDU3}).
%observe that (\ref{aln2}) implies $A_{q,n}(U^{*})=L_{q,n}$ for all $n\geq\{n_{5},n_{6}\}$.  Claim \ref{grandefinale}.1 and (\ref{U}) imply that $U^{*}\subseteq U_{n}$ for all $n\geq n_{1}$.  Since $A_{q,n}(U_{n})\subseteq L_{q,n}$ for all $n\in\mathbb{N}$, it follows that  $A_{q,n}(U^{*})\subseteq A_{q,n}(U_{n})\subseteq L_{q,n}$ for all $n\geq n_{1}$.  Therefore, $A_{q,n}(U_{n})=A_{q,n}(U^{*})$ for all $n\geq\max\{n_{1},n_{5},n_{6}\}$ which implies (\ref{newgrandfinale9}).  Since $n\geq\max\{n_{5},n_{6}\}$, we invoke (\ref{FDU2}) and (\ref{newgrandefinale6}) with the  Chain Rule to justify (\ref{newgrandfinale10}).  Moreover (\ref{newgrandfinale11}) follows because Claim \ref{grandefinale}.7 and (\ref{aln}) imply  $A_{q,n}(U_{n}\setminus U^{*})=L_{q,n}$ for all $n\geq n_{7}$ and every $q\in\{0,\ldots,l(n)\}$.  Since $A_{q,n}(U_{n}\setminus U^{*})\subseteq A_{q,n}(U_{n})\subseteq L_{q,n}$ for all $n\in\mathbb{N}$, it follows that $A_{q,n}(U_{n}\setminus U^{*})=A_{q,n}(U_{n})$ for all $n\geq n_{7}$ and every  $q\in\{0,\ldots,l(n)\}$.

Now $R_{U_{n}}=R_{U_{n}\setminus U^{*}}+R_{U^{*}}$.  Since $n\geq\max\{K_{6},K_{7}\}$ we can invoke (\ref{newgrandfinale13}).  Together (\ref{newgrandfinale13}) and (\ref{newgrandfinale11}) imply:
\begin{align}
	\label{actualgrandefinale}
	R_{U^{*}}&>I(X_{F_{d}(U^{*})};Y_{d}|X_{\tilde{F}_{d}(U^{*})}).
\end{align}
By selection ${\bf R}\in{\cal R}_{d}({\bf F})$ satisfies (\ref{biggiewiggie}) for all $S\subseteq{\cal S}(d)$, which contradicts (\ref{actualgrandefinale}).  It follows that the non-empty $U^{*}$ defined in (\ref{U}) does not exist.  Therefore Claim \ref{grandefinale}.1 implies ${\bf R}\in{\cal R}({\bf F},{\bf L}_{d,n})$ for some $n\in\mathbb{N}$.
\end{proof}

%% file: example.tex
\section{The Diamond Relay Channel}
\label{Example}
  %virtual sources and flows
  \begin{figure*}[t]
        \center{\includegraphics[width=\textwidth]
        {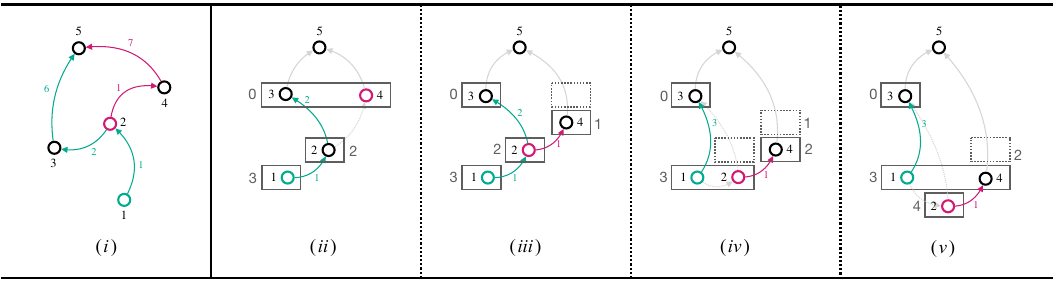}}
       \caption{(i) ${\bf F}=\{{\bf f}(1),{\bf f}(2)\}$ where ${\bf f}(1)=1\xrightarrow{1}2\xrightarrow{2}3\xrightarrow{6}5$ and ${\bf f}(2)=2\xrightarrow{1}4\xrightarrow{7}5$ (ii) ${\bf L}_{5,1}=(\{3,4\},\{\},\{2\},\{1\})$ with virtual flows ${\bf v}(1)=1\xrightarrow{1}2\xrightarrow{2}3$ and ${\bf v}(2)=4$ (iii) ${\bf L}_{5,2}=(\{3\},\{4\},\{2\},\{1\})$ with virtual flows ${\bf v}(1)=1\xrightarrow{1}2\xrightarrow{2}3$ and ${\bf v}(2)=2\xrightarrow{1}4$ (iv) ${\bf L}_{5,3}=(\{3\},\{\},\{4\},\{1,2\})$ with virtual flows ${\bf v}(1)=1\xrightarrow{3}3$ and ${\bf v}(2)=2\xrightarrow{1}4$ (v) ${\bf L}_{5,4}=(\{3\},\{\},\{\},\{1,4\},\{2\})$ with virtual flows ${\bf v}(1)=1\xrightarrow{3}3$ and ${\bf v}(2)=2\xrightarrow{1}4$.}   
	\label{figure4}
        \label{figure4}
 \end{figure*}
Given ${\cal N}=\{1,2,3,4,5\}$, let ${\bf F}=\{{\bf f}(1),{\bf f}(2)\}$ where ${\bf f}(1)=1\xrightarrow{1}2\xrightarrow{2}3\xrightarrow{6}5$ and ${\bf f}(2)=2\xrightarrow{1}4\xrightarrow{7}5$, both depicted in Figure \ref{figure4}(i).  In block $b$, node 1 encodes the pair $(m_{1}(b),``1")$, node 2 encodes $(m_{1}(b-1),m_{2}(b))$, node 3 encodes $(m_{1}(b-3),``1")$, and node 4 encodes $(``1",m_{2}(b-1))$, where $m_{1}(b)\in\{1,\ldots,2^{nR_{1}}\}$ and $m_{2}(b)\in\{1,\ldots,2^{nR_{2}}\}$ for all  $b=1,\ldots,B$.  

To encode the message vector $(m_{1},m_{2})$ where $m_{1}\in\{1,\ldots,2^{nR_{1}}\}$ and $m_{2}\in\{1,\ldots,2^{nR_{2}}\}$, node $i$ sends the codeword ${\bf x}_{i}(w)$ where $w\in\{1,\ldots,2^{n(R_{1}+R_{2})}\}$ is the unique index that maps to the message pair $(m_{1},m_{2})$ and $i\in\{1,2,3,4\}$.  To simplify the exposition, we express ${\bf x}_{i}(w)$ as a function of the message vector itself, so ${\bf x}_{2}(m_{1},m_{2}):={\bf x}_{2}(w)$.  Since $m_{2}=1$ at node 1 and node 3, and $m_{1}=1$ at node 4, we define ${\bf x}_{1}(m_{1}):={\bf x}_{1}(w)$, ${\bf x}_{3}(m_{1}):={\bf x}_{3}(w)$, and ${\bf x}_{4}(m_{2}):={\bf x}_{4}(w)$.    

By definition, ${\cal R}_{5}({\bf F})$ is set of rate vectors ${\bf R}=(R_{1},R_{2})$ that satisfies:
 \begin{align}
 	\label{outerR1}
 	R_{1}&<I(X_{1}X_{2}X_{3};Y_{5}|X_{4})\\
 	\label{outerR2}
 	R_{2}&<I(X_{2}X_{4};Y_{5}|X_{1}X_{3})\\
 	\label{outerR1R2}
 	R_{1}+R_{2}&<I(X_{1}X_{2}X_{3}X_{4};Y_{5})
 \end{align}
  
 Four different decoding schemes collectively achieve all the rate vectors in ${\cal R}_{5}({\bf F})$.
 
\subsection{The first decoding scheme}

Set ${\bf L}_{5,1}=(\{3,4\},\{\},\{2\},\{1\})$.  It follows that $v(1)=1$ since $\text{\sc layer}(1)-\text{\sc layer}(2)=1\leq k_{1,2}$ and $\text{\sc layer}(1)-\text{\sc layer}(3)=3\leq k_{1,3}$ which satisfy (\ref{firstIndexCoding}).  Similarly, $v(2)=4$ since $\text{\sc layer}(2)-\text{\sc layer}(4)=2>1$ which satisfies (\ref{secondIndexCoding}).  Then ${\bf v}(1)=1\xrightarrow{1}2\xrightarrow{2}3$ and ${\bf v}(2)=4$.  The virtual flows are depicted in Figure \ref{figure4}(ii).  

In block $b$, node $5$ decodes $(m_{1}(b-3),m_{2}(b-1))$ by finding the unique pair $\hat{m}_{1}\in\{1,\ldots,2^{nR_{1}}\}$ and $\hat{m}_{2}\in\{1,\ldots,2^{nR_{2}}\}$ that jointly satisfy the following typicality checks:
	\begin{align}
		\nonumber
		&({\bf x}_{3}(\hat{m}_{1}),{\bf x}_{4}(\hat{m}_{2}),{\bf Y}_{5}(b))\in T^{(n)}_{\epsilon}(X_{3},X_{4},Y_{5})\\
		\nonumber
		&({\bf x}_{2}(\hat{m}_{1},m_{2}(b-2)),{\bf X}_{3}(b-2),{\bf X}_{4}(b-2),{\bf Y}_{5}(b-2))\\
		\label{typetwo}
		&\hspace{40mm}\in T^{(n)}_{\epsilon}(X_{2},X_{3},X_{4},Y_{5})\\
		\nonumber
		&({\bf x}_{1}(\hat{m}_{1}),{\bf X}_{2}(b-3),{\bf X}_{3}(b-3),{\bf X}_{4}(b-3),{\bf Y}_{5}(b-3))\hspace{4mm}\\
		\nonumber
		&\hspace{40mm}\in T^{(n)}_{\epsilon}(X_{2},X_{3},X_{4},Y_{5})
	\end{align} 
	In each block, node 2 sends its own source message and a source message from node 1.  The message $m_{2}(b-2)$ has already been decoded by node $5$ in block $b$, so $\hat{m}_{1}$ uniquely determines ${\bf x}_{2}(\hat{m}_{1},m_{2}(b-2))$ in (\ref{typetwo}).  The probability of error goes to zero if ${\bf R}$ satisfies the following conditions:
	\begin{align}
		\nonumber
		R_{1}&<I(X_{3};Y_{5}|X_{4})+I(X_{2};Y_{5}|X_{3}X_{4})\\
		\nonumber
		&\hspace{40mm}+I(X_{1};Y_{5}|X_{2}X_{3}X_{4})\\
		\label{scheme1r1}
		&=I(X_{1}X_{2}X_{3};Y_{5}|X_{4})\\
		\label{scheme1r2}
		R_{2}&<I(X_{4};Y_{5}|X_{3})\\
		\nonumber
		R_{1}+R_{2}&<I(X_{3}X_{4};Y_{5})+I(X_{2};Y_{5}|X_{3}X_{4})\\
		\nonumber
		&\hspace{35mm}+I(X_{1};Y_{5}|X_{2}X_{3}X_{4})\\
		\label{scheme1r1r2}
		&=I(X_{1}X_{2}X_{3}X_{4};Y_{5})
	\end{align}
	Each of the three inequalities above addresses an error event.  The probability that $m_{2}$ is decoded correctly and $m_{1}$ is not, goes to zero if (\ref{scheme1r1}) is satisfied.  Similarly, the probability that $m_{1}$ is decoded correctly and $m_{2}$ is not, goes to zero if (\ref{scheme1r2}) is satisfied.  Finally, the probability that both $m_{1}$ and $m_{2}$ are not decoded correctly goes to zero if (\ref{scheme1r1r2}) is satisfied.  Let $({\bf F},{\bf L}^{\prime}_{5,1}):=\text{\sc shift}(({\bf F},{\bf L}_{5,1}),\{2\})$.
	
	\subsection{The second decoding scheme}
	Set ${\bf L}_{5,2}=(\{3\},\{4\},\{2\},\{1\})$.  It follows that $v(1)=1$ since $\text{\sc layer}(1)-\text{\sc layer}(2)=1\leq k_{1,2}$ and $\text{\sc layer}(1)-\text{\sc layer}(3)=3\leq k_{1,3}$ which satisfy (\ref{firstIndexCoding}).  Similarly, $v(2)=2$ since $\text{\sc layer}(2)-\text{\sc layer}(4)\leq k_{2,4}=1$ which satisfies (\ref{firstIndexCoding}).  Then ${\bf v}(1)=1\xrightarrow{1}2\xrightarrow{2}3$ and ${\bf v}(2)=2\xrightarrow{1}4$.  The virtual flows are depicted in Figure \ref{figure4}(iii).  Observe that $({\bf F},{\bf L}^{\prime}_{5,1})=({\bf F},{\bf L}_{5,2})$.   
	
	In block $b$, node $5$ decodes $(m_{1}(b-3),m_{2}(b-2))$ by finding the unique pair $\hat{m}_{1}\in\{1,\ldots,2^{nR_{1}}\}$ and $\hat{m}_{2}\in\{1,\ldots,2^{nR_{2}}\}$ that jointly satisfy the following typicality checks:  
  \begin{align}
		\nonumber
		&({\bf x}_{3}(\hat{m}_{1}),{\bf Y}_{5}(b))\in T^{(n)}_{\epsilon}(X_{3},Y_{5})\\
		\nonumber
		&({\bf x}_{4}(\hat{m}_{2}),{\bf X}_{3}(b-1),{\bf Y}_{5}(b-1))\in T^{(n)}_{\epsilon}(X_{3},X_{4},Y_{5})\\
		\nonumber
		&({\bf x}_{2}(\hat{m}_{1},\hat{m}_{2}),{\bf X}_{3}(b-2),{\bf X}_{4}(b-2),{\bf Y}_{5}(b-2))\\
		\nonumber
		&\hspace{53mm}\in T^{(n)}_{\epsilon}(X_{2},X_{3},X_{4},Y_{5})\\
		\nonumber
		&({\bf x}_{1}(\hat{m}_{1}),{\bf X}_{2}(b-3),{\bf X}_{3}(b-3),{\bf X}_{4}(b-3),{\bf Y}_{5}(b-3))\hspace{4mm}\\
		\nonumber
		&\hspace{53mm}\in T^{(n)}_{\epsilon}(X_{2},X_{3},X_{4},Y_{5})
	\end{align} 
  The probability of error goes to zero if ${\bf R}$ satisfies the following conditions:
	\begin{align}
		\nonumber
		R_{1}&<I(X_{3};Y_{5})+I(X_{2};Y_{5}|X_{3}X_{4})\\
		\nonumber
		&\hspace{40mm}+I(X_{1};Y_{5}|X_{2}X_{3}X_{4})\\
		\label{scheme2r1}
		&=I(X_{3};Y_{5})+I(X_{1}X_{2};Y_{5}|X_{3}X_{4})\\
		\label{scheme2r2}
		R_{2}&<I(X_{4};Y_{5}|X_{3})+I(X_{2};Y_{5}|X_{3}X_{4})\\
		\nonumber
		R_{1}+R_{2}&<I(X_{3};Y_{5})+I(X_{4};Y_{5}|X_{3})\\
		\nonumber
		&\hspace{5mm}+I(X_{2};Y_{5}|X_{3}X_{4})+I(X_{1};Y_{5}|X_{2}X_{3}X_{4})\\
		\label{scheme2r1r2}
		&=I(X_{1}X_{2}X_{3}X_{4};Y_{5})
	\end{align}
	Let $({\bf F},{\bf L}^{\prime}_{5,2})=\text{\sc shift}({\bf F},{\bf L}_{5,2},\{2\})$.  

\subsection{The third decoding scheme}	
	Set ${\bf L}_{5,3}=(\{3\},\{\},\{4\},\{1,2\})$.  It follows that $v(1)=1$ since $\text{\sc layer}(1)-\text{\sc layer}(2)=0\leq k_{1,2}$ and $\text{\sc layer}(1)-\text{\sc layer}(3)=3=k_{1,3}$ which satisfy (\ref{firstIndexCoding}).  Similarly, $v(2)=2$ since $\text{\sc layer}(2)-\text{\sc layer}(4)=1=k_{2,4}$ which satisfies (\ref{firstIndexCoding}).  Then ${\bf v}(1)=1\xrightarrow{3}3$ and ${\bf v}(2)=2\xrightarrow{1}4$.  The virtual flows are depicted in Figure \ref{figure4}(iv).  Observe that $({\bf F},{\bf L}^{\prime}_{5,2})=({\bf F},{\bf L}_{5,3})$.
	
	In block $b$, node $5$ decodes $(m_{1}(b-3),m_{2}(b-3))$ by finding the unique pair $\hat{m}_{1}\in\{1,\ldots,2^{nR_{1}}\}$ and $\hat{m}_{2}\in\{1,\ldots,2^{nR_{2}}\}$ that jointly satisfy the following typicality checks: 
	 \begin{align}
		\nonumber
		&({\bf x}_{3}(\hat{m}_{1}),{\bf Y}_{5}(b))\in T^{(n)}_{\epsilon}(X_{3},Y_{5})\\
		\nonumber
		&({\bf x}_{4}(\hat{m}_{2}),{\bf X}_{3}(b-2),{\bf Y}_{5}(b-2))\in T^{(n)}_{\epsilon}(X_{3},X_{4},Y_{5})\\
		\nonumber
		&({\bf x}_{1}(\hat{m}_{1}),{\bf x}_{2}(m_{1}(b-4),\hat{m}_{2}),{\bf X}_{3}(b-3),{\bf X}_{4}(b-3),\\
		\label{scheme3type3}
		&\hspace{15mm}{\bf Y}_{d}(b-3))\in T^{(n)}_{\epsilon}(X_{1},X_{2},X_{3},X_{4},Y_{5})
	\end{align} 
	The message $m_{1}(b-4)$ has already been decoded by node 5 in block $b$, so $\hat{m}_{2}$ uniquely determines ${\bf x}_{2}(m_{1}(b-4),\hat{m}_{2})$ in (\ref{scheme3type3}).  The probability of error goes to zero if ${\bf R}$ satisfies the following conditions:
	\begin{align}
		\label{scheme3r1}
		R_{1}&<I(X_{3};Y_{5})+I(X_{1};Y_{5}|X_{2}X_{3}X_{4})\\
		\label{scheme3r2}
		R_{2}&<I(X_{4};Y_{5}|X_{3})+I(X_{2};Y_{5}|X_{1}X_{3}X_{4})\\
		\nonumber
		R_{1}+R_{2}&<I(X_{3};Y_{5})+I(X_{4};Y_{5}|X_{3})\\
		\nonumber
		&\hspace{25mm}+I(X_{1}X_{2};Y_{5}|X_{3}X_{4})\\
		\label{scheme3r1r2}
		&=I(X_{1}X_{2}X_{3}X_{4};Y_{5})
	\end{align}
	Let $({\bf F},{\bf L}^{\prime}_{5,3})=\text{\sc shift}({\bf F},{\bf L}_{5,3}),\{2\})$.

\subsection{The fourth decoding scheme}
Set ${\bf L}_{5,4}=(\{3\},\{\},\{\},\{1,4\},\{2\})$.  It follows that $v(1)=1$ since $\text{\sc layer}(1)-\text{\sc layer}(2)=0\leq k_{1,2}$ and $\text{\sc layer}(1)-\text{\sc layer}(3)=3=k_{1,3}$ which satisfy (\ref{firstIndexCoding}).  Similarly, $v(2)=2$ since $\text{\sc layer}(2)-\text{\sc layer}(4)=1=k_{2,4}$ which satisfies (\ref{firstIndexCoding}).  Then ${\bf v}(1)=1\xrightarrow{3}3$ and ${\bf v}(2)=2\xrightarrow{1}4$.  The virtual flows are depicted in Figure \ref{figure4}(v).  Observe that $({\bf F},{\bf L}^{\prime}_{5,3})=({\bf F},{\bf L}_{5,4})$.  

In block $b$, node 5 decodes $(m_{1}(b-3),m_{2}(b-4))$ by finding the unique pair $\hat{m}_{1}\in\{1,\ldots,2^{nR_{1}}\}$ and $\hat{m}_{2}\in\{1,\ldots,2^{nR_{2}}\}$ that jointly satisfy the following typicality checks: 
	\begin{align}
		\nonumber
		&({\bf x}_{3}(\hat{m}_{1}),{\bf Y}_{5}(b))\in T^{(n)}_{\epsilon}(X_{3},Y_{5})\\
		\nonumber
		&({\bf x}_{1}(\hat{m}_{1}),{\bf x}_{4}(\hat{m}_{2}),{\bf X}_{3}(b-3),{\bf Y}_{5}(b-3))\\
		\nonumber
		&\hspace{50mm}\in T^{(n)}_{\epsilon}(X_{1},X_{3},X_{4},Y_{5})\\
		\nonumber
		&({\bf x}_{2}(m_{1}(b-5),\hat{m}_{2}),{\bf X}_{1}(b-4),{\bf X}_{3}(b-4),{\bf X}_{4}(b-4),\\
		\nonumber
		&\hspace{28mm}{\bf Y}_{5}(b-4))\in T^{(n)}_{\epsilon}(X_{1},X_{2},X_{3},X_{4},Y_{5})
	\end{align} 
	The message $m_{1}(b-5)$ has already been decoded by node 5 in block $b$, so $\hat{m}_{2}$ uniquely determines ${\bf x}_{2}(m_{1}(b-5),\hat{m}_{2})$ in (\ref{scheme3type3}).  The probability of error goes to zero if ${\bf R}$ satisfies the following conditions:
	\begin{align}
		\label{scheme4r1}
		R_{1}&<I(X_{3};Y_{5})+I(X_{1};Y_{5}|X_{3}X_{4})\\
		\nonumber
		R_{2}&<I(X_{4};Y_{5}|X_{1}X_{3})+I(X_{2};Y_{5}|X_{1}X_{3}X_{4})\\
		\label{scheme4r2}
		&=I(X_{2}X_{4};Y_{5}|X_{1}X_{3})\\
		\nonumber
		R_{1}+R_{2}&<I(X_{3};Y_{5})+I(X_{1}X_{4};Y_{5}|X_{3})\\
		\nonumber
		&\hspace{25mm}+I(X_{2};Y_{5}|X_{1}X_{3}X_{4})\\
		\label{scheme4r1r2}
		&=I(X_{1}X_{2}X_{3}X_{4};Y_{5})
	\end{align}	
	\subsection{The achievability of ${\cal R}_{5}({\bf F})$}
	To prove that ${\cal R}_{5}({\bf F})$ is achievable, it suffices to show that any rate vector in the region defined by (\ref{outerR1})-(\ref{outerR1R2}) is in the region defined by (\ref{scheme1r1})-(\ref{scheme1r1r2}) or (\ref{scheme2r1})-(\ref{scheme2r1r2}) or (\ref{scheme3r1})-(\ref{scheme3r1r2}) or (\ref{scheme4r1})-(\ref{scheme4r1r2}).  Suppose ${\bf R}$ is not in (\ref{scheme1r1})-(\ref{scheme1r1r2}).  Since (\ref{scheme1r1}) and (\ref{scheme1r1r2}) define the boundaries of ${\cal R}_{5}({\bf F})$, ${\bf R}$ must violate (\ref{scheme1r2}).  Then  (\ref{scheme1r1r2}) implies that ${\bf R}$ satisfies (\ref{scheme2r1}).  If ${\bf R}$ satisfies (\ref{scheme2r2}) the proof is finished since (\ref{scheme2r1r2}) is a boundary of ${\cal R}_{5}({\bf F})$.  If ${\bf R}$ violates (\ref{scheme2r2}) then (\ref{scheme2r1r2}) implies that ${\bf R}$ satisfies (\ref{scheme3r1}).  If ${\bf R}$ satisfies (\ref{scheme3r2}) the proof is finished since (\ref{scheme3r1r2}) is a boundary of ${\cal R}_{5}({\bf F})$.  Otherwise (\ref{scheme3r1r2}) implies that ${\bf R}$ satisfies (\ref{scheme4r1}).  Here the proof is finished  since (\ref{scheme4r2}) and (\ref{scheme4r1r2}) are boundaries of ${\cal R}_{5}({\bf F})$.

%% file: appendixbcd.tex
\section{The Two-Source Multiple-Access Relay Channel}
\label{appendix2}
Given ${\cal N}=\{1,2,3,4\}$, let ${\bf F}=\{{\bf f}(1),{\bf f}(2)\}$, where ${\bf f}(1)=1\xrightarrow{1}3\xrightarrow{1}4$ and ${\bf f}(2)=2\xrightarrow{1}3\xrightarrow{1}4$ as depicted in Figure \ref{figure6}(i).  In block $b$, node 1 encodes the pair $(m_{1}(b),``1")$, node 2 encodes $(``1",m_{2}(b))$, and node 3 encodes $(m_{1}(b-1),m_{2}(b-1))$.  

To encode the message vector $(m_{1},m_{2})$ where $m_{1}\in\{1,\ldots,2^{nR_{1}}\}$ and $m_{2}\in\{1,\ldots,2^{nR_{2}}\}$, node $i$ sends the codeword ${\bf x}_{i}(w)$ where $w\in\{1,\ldots,2^{n(R_{1}+R_{2})}\}$ is the unique index that maps to the message pair $(m_{1},m_{2})$ and $i\in\{1,2,3\}$.  To simplify the exposition, we express ${\bf x}_{i}(w)$ as a function of the message vector itself, so ${\bf x}_{3}(m_{1},m_{2}):={\bf x}_{3}(w)$.  Since $m_{2}=1$ at node 1, and $m_{1}=1$ at node 2, we define ${\bf x}_{1}(m_{1}):={\bf x}_{1}(w)$ and ${\bf x}_{2}(m_{2}):={\bf x}_{2}(w)$.

 Let ${\bf F}=\{{\bf f}(1),{\bf f}(2)\}$.  By definition, ${\cal R}_{4}({\bf F})$ is set of rate vectors ${\bf R}=(R_{1},R_{2})$ that satisfies:
 \begin{align}
 	\label{outerR1A2}
 	R_{1}&<I(X_{1}X_{3};Y_{4}|X_{2})\\
 	\label{outerR2A2}
 	R_{2}&<I(X_{2}X_{3};Y_{4}|X_{1})\\
 	\label{outerR1R2A2}
 	R_{1}+R_{2}&<I(X_{1}X_{2}X_{3};Y_{4})
 \end{align}
  
 Three different decoding schemes collectively achieve all the rate vectors in ${\cal R}_{4}({\bf F})$ as depicted in Figure \ref{polytope}(i).
%marc
  \begin{figure*}[t]
        \center{\includegraphics[width=\textwidth]
        {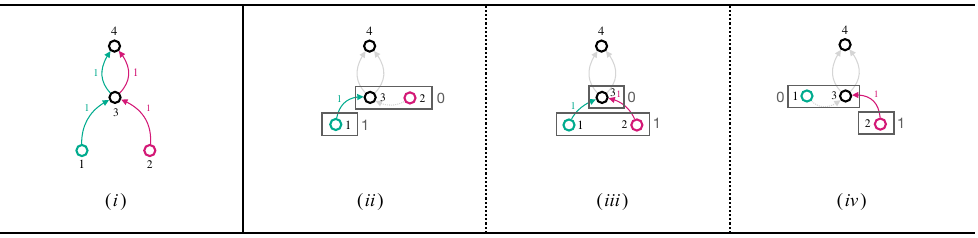}}
       \caption{(i) ${\bf F}=\{{\bf f}(1),{\bf f}(2)\}$ where ${\bf f}(1)=1\xrightarrow{1}3\xrightarrow{1}4$ and ${\bf f}(2)=2\xrightarrow{1}3\xrightarrow{1}4$ (ii) ${\bf L}_{5,1}=(\{2,3\},\{1\})$ with virtual flows ${\bf v}(1)=1\xrightarrow{1}3$ and ${\bf v}(2)=2$ (iii) ${\bf L}_{5,2}=(\{3\},\{1,2\})$ with virtual flows ${\bf v}(1)=1\xrightarrow{1}3$ and ${\bf v}(2)=2\xrightarrow{1}3$ (iv) ${\bf L}_{5,3}=(\{1,3\},\{2\})$ with virtual flows ${\bf v}(1)=1$ and ${\bf v}(2)=1\xrightarrow{1}3$.}   
	\label{figure6}
 \end{figure*}

\subsection{The first decoding scheme}

Set ${\bf L}_{4,1}=(\{2,3\},\{1\})$.  It follows that $v(1)=1$ since $\text{\sc layer}(1)-\text{\sc layer}(3)=1=k_{1,3}$ which satisfies (\ref{firstIndexCoding}).  Similarly, $v(2)=2$ since $\text{\sc layer}(2)-\text{\sc layer}(3)=0\leq k_{2,3}$ which satisfies (\ref{firstIndexCoding}).  Then ${\bf v}(1)=1\xrightarrow{1}3$ and ${\bf v}(2)=2\xrightarrow{1}3$ as depicted in Figure \ref{figure6}(ii).  

In block $b$, node 4 decodes $(m_{1}(b-1),m_{2}(b))$ by finding the unique pair $\hat{m}_{1}\in\{1,\ldots,2^{nR_{1}}\}$ and $\hat{m}_{2}\in\{1,\ldots,2^{nR_{2}}\}$ that jointly satisfies the following typicality checks:
	\begin{align}
		\nonumber
		&({\bf x}_{3}(\hat{m}_{1},m_{2}(b-1)),{\bf x}_{2}(\hat{m}_{2}),{\bf Y}_{4}(b))\in T^{(n)}_{\epsilon}(X_{2},X_{3},Y_{4})\\
		\nonumber
		&({\bf x}_{1}(\hat{m}_{1}),{\bf X}_{3}(b-1),{\bf X}_{2}(b-1),{\bf Y}_{4}(b-1))\\
		\label{marctypcheck}
		&\hspace{40mm}\in T^{(n)}_{\epsilon}(X_{1},X_{2},X_{3},Y_{4})
	\end{align} 
The message $m_{2}(b-1)$ has already been decoded by node 4 in block $b$, so $\hat{m}_{1}$ uniquely determines ${\bf x}_{3}(\hat{m}_{1},m_{2}(b-1))$ in (\ref{marctypcheck}).  The probability of error goes to zero if ${\bf R}$ satisfies the following conditions:
	\begin{align}
		\nonumber
		R_{1}&<I(X_{1};Y_{4}|X_{2}X_{3})+I(X_{3};Y_{2}|X_{1})\\
		\label{scheme1r1A2}
		&=I(X_{1}X_{3};Y_{2}|X_{2})\\
		\label{scheme1r2A2}
		R_{2}&<I(X_{1};Y_{4}|X_{2})\\
		\nonumber
		R_{1}+R_{2}&<I(X_{2}X_{3};Y_{4})+I(X_{1};Y_{4}|X_{2}X_{3})\\
		\label{scheme1r1r2A2}
		&=I(X_{1}X_{2}X_{3};Y_{4})
	\end{align}
	In Figure \ref{polytope}(i), ${\cal R}({\bf F},{\bf L}_{4,1})$ denotes the region of rate vectors that satisfies (\ref{scheme1r1A2})-(\ref{scheme1r1r2A2}).
	
	Each of the three inequalities above addresses an error event.  The probability that $\hat{m}_{2}$ is decoded correctly and $\hat{m}_{1}$ is not, goes to zero if (\ref{scheme1r1A2}) is satisfied.  Similarly, the probability that $\hat{m}_{1}$ is decoded correctly and $\hat{m}_{2}$ is not, goes to zero if (\ref{scheme1r2A2}) is satisfied.  Finally, the probability that both $\hat{m}_{1}$ and $\hat{m}_{2}$ are not decoded correctly goes to zero if (\ref{scheme1r1r2A2}) is satisfied.
	
	\subsection{The second decoding scheme}
	Set ${\bf L}_{4,2}=(\{3\},\{1,2\})$.  It follows that $v(1)=1$ since $\text{\sc layer}(1)-\text{\sc layer}(3)=1\leq k_{1,3}$ which satisfies (\ref{firstIndexCoding}).  Similarly, $v(2)=2$ since $\text{\sc layer}(2)-\text{\sc layer}(3)=1\leq k_{2,3}$ which satisfies (\ref{firstIndexCoding}).  Then ${\bf v}(1)=1\xrightarrow{1}3$ and ${\bf v}(2)=2\xrightarrow{1}3$ as depicted in Figure \ref{figure6}(iii).  Observe that $({\bf F},{\bf L}_{4,2})=\text{\sc shift}(({\bf F},{\bf L}_{4,1}),\{2\})$.   
	
	In block $b$, node 4 decodes $(m_{1}(b-1),m_{2}(b-1))$ by finding the unique pair $\hat{m}_{1}\in\{1,\ldots,2^{nR_{1}}\}$ and $\hat{m}_{2}\in\{1,\ldots,2^{nR_{2}}\}$ that jointly satisfies the following typicality checks:  
  \begin{align}
		\nonumber
		&({\bf x}_{3}(\hat{m}_{1},\hat{m}_{2}),{\bf Y}_{4}(b))\in T^{(n)}_{\epsilon}(X_{3},Y_{4})\\
		\nonumber
		&({\bf x}_{1}(\hat{m}_{1}),{\bf x}_{2}(\hat{m}_{2}),{\bf X}_{3}(b-1),{\bf Y}_{4}(b-1))\\
		\nonumber
		&\hspace{51mm}\in T^{(n)}_{\epsilon}(X_{1}X_{2}X_{3}Y_{4})
	\end{align} 
  The probability of error goes to zero if ${\bf R}$ satisfies the following conditions:
	\begin{align}
		\nonumber
		R_{1}&<I(X_{3};Y_{4})+I(X_{1};Y_{4}|X_{2}X_{3})\\
		\label{scheme2r1A2}
		&=I(X_{1}X_{3};Y_{4}|X_{2})\\
		\nonumber
		R_{2}&<I(X_{3};Y_{4})+I(X_{2};Y_{4}|X_{1}X_{3})\\
		\label{scheme2r2A2}
		&=I(X_{2}X_{3};Y_{4}|X_{1})\\
		\nonumber
		R_{1}+R_{2}&<I(X_{3};Y_{4})+I(X_{2};Y_{4}|X_{1}X_{3})\\
		\label{scheme2r1r2A2}
		&=I(X_{1}X_{2}X_{3};Y_{4})
	\end{align}
	In Figure \ref{polytope}(i), ${\cal R}({\bf F},{\bf L}_{4,2})$ denotes the region of rate vectors that satisfies (\ref{scheme2r1A2})-(\ref{scheme2r1r2A2}).
	
\subsection{The third decoding scheme}	
	Set ${\bf L}_{4,3}=(\{1,3\},\{2\})$.  It follows that $v(1)=1$ since $\text{\sc layer}(1)-\text{\sc layer}(2)=-1\leq k_{1,2}$ which satisfies (\ref{firstIndexCoding}).  Similarly, $v(2)=2$ since $\text{\sc layer}(2)-\text{\sc layer}(3)=1=k_{2,3}$ which satisfies (\ref{firstIndexCoding}).  Then ${\bf v}(1)=1$ and ${\bf v}(2)=2\xrightarrow{1}3$ as depicted in Figure \ref{figure6}(iv).  Observe that $({\bf F},{\bf L}_{4,2})=\text{\sc shift}({\bf F},{\bf L}_{4,3}),\{3\})$.  
	
	In block $b$, node 4 decodes $(m_{1}(b),m_{2}(b-1))$ by finding the unique pair $\hat{m}_{1}\in\{1,\ldots,2^{nR_{1}}\}$ and $\hat{m}_{2}\in\{1,\ldots,2^{nR_{2}}\}$ that jointly satisfies the following typicality checks: 
	 \begin{align}
		\nonumber
		&({\bf x}_{1}(\hat{m}_{1}),{\bf x}_{3}(m_{1}(b-1),\hat{m}_{2}),{\bf Y}_{4}(b))\\
		\label{scheme3type3A2}
		&\hspace{40mm} \in T^{(n)}_{\epsilon}(X_{1},X_{3},Y_{4})\\
		\nonumber
		&({\bf x}_{2}(\hat{m}_{2}),{\bf X}_{1}(b-1),{\bf X}_{3}(b-1),{\bf Y}_{4}(b-1))\\
		\nonumber
		&\hspace{40mm}\in T^{(n)}_{\epsilon}(X_{1},X_{2},X_{3},Y_{4})
	\end{align} 
	The message $m_{1}(b-1)$ has already been decoded by node 4 in block $b$, so $\hat{m}_{2}$ uniquely determines ${\bf x}_{3}(m_{1}(b-1),\hat{m}_{2})$ in (\ref{scheme3type3A2}).  The probability of error goes to zero if ${\bf R}$ satisfies the following conditions:
	
	\begin{align}
		\label{scheme3r1A2}
		R_{1}&<I(X_{1};Y_{4}|X_{3})\\
		\nonumber
		R_{2}&<I(X_{3};Y_{4}|X_{1})+I(X_{2};Y_{4}|X_{1}X_{3})\\
		\label{scheme3r2A2}
		&=I(X_{2}X_{3};Y_{4}|X_{1})\\
		\nonumber
		R_{1}+R_{2}&<I(X_{1}X_{3};Y_{4})+I(X_{2};Y_{4}|X_{1}X_{3})\\
		\label{scheme3r1r2A2}		
		&=I(X_{1}X_{2}X_{3};Y_{4})
	\end{align}
	In Figure \ref{polytope}(i), ${\cal R}({\bf F},{\bf L}_{4,3})$ denotes the region of rate vectors that satisfies (\ref{scheme3r1A2})-(\ref{scheme3r1r2A2}).
	
	\subsection{The achievability of ${\cal R}_{4}({\bf F})$}
	To show that ${\cal R}_{4}({\bf F})$ is achievable, it suffices to show that any rate vector in the region defined by (\ref{outerR1A2})-(\ref{outerR1R2A2}) is in the region defined by (\ref{scheme1r1A2})-(\ref{scheme1r1r2A2}) or (\ref{scheme2r1A2})-(\ref{scheme2r1r2A2}) or (\ref{scheme3r1A2})-(\ref{scheme3r1r2A2}).  Suppose ${\bf R}$ is not in (\ref{scheme1r1A2})-(\ref{scheme1r1r2A2}).  Since (\ref{scheme1r1A2}) and (\ref{scheme1r1r2A2}) define the boundaries of ${\cal R}_{4}({\bf F})$, ${\bf R}$ must violate (\ref{scheme1r2A2}).  Then (\ref{scheme1r1r2A2}) implies that ${\bf R}$ satisfies (\ref{scheme2r1A2}).  If ${\bf R}$ satisfies (\ref{scheme2r2A2}) the proof is finished since (\ref{scheme2r1r2A2}) is a boundary of ${\cal R}_{4}({\bf F})$.  If ${\bf R}$ violates (\ref{scheme2r2A2}) then (\ref{scheme2r1r2A2}) implies that ${\bf R}$ satisfies (\ref{scheme3r1A2}).  Now the proof is finished since (\ref{scheme3r2A2}) and (\ref{scheme3r1r2A2}) are boundaries of ${\cal R}_{4}({\bf F})$.
	
	Observe that ${\cal R}_{4}({\bf F})=\cup^{3}_{k=1}{\cal R}({\bf F},{\bf L}_{4,k})$ as depicted in Figure \ref{polytope}(i); the achievable regions of the individual coding schemes collectively achieve the DF outer-bound.

\section{The Three-Source Multiple-Access Relay Channel}
\label{appendix3}
%threesourcemarc
  \begin{figure*}[t]
        \center{\includegraphics[width=\textwidth]
        {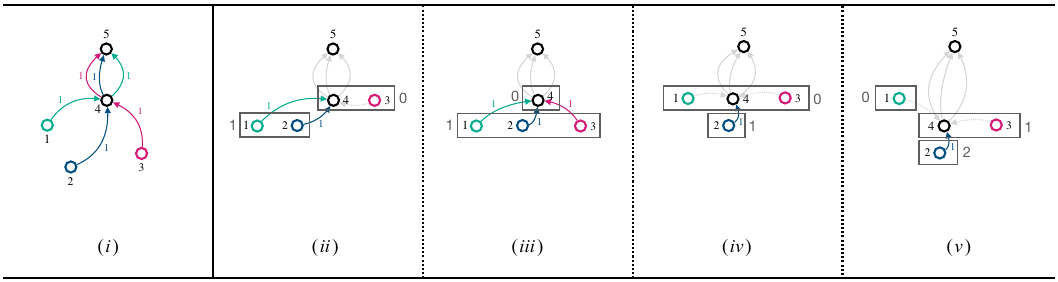}}
   \caption{(i) ${\bf F}=\{{\bf f}(1),{\bf f}(2),{\bf f}(3)\}$ where ${\bf f}(1)=1\xrightarrow{1}4\xrightarrow{1}5$, ${\bf f}(2)=2\xrightarrow{1}4\xrightarrow{1}5$ and ${\bf f}(3)=3\xrightarrow{1}4\xrightarrow{1}5$ (ii) ${\bf L}_{5,3}=(\{2,4\},\{1,3\})$ with virtual flows ${\bf v}(1)=1\xrightarrow{1}4$, ${\bf v}(2)=2\xrightarrow{1}4$, and ${\bf v}(3)=3$ (iii) ${\bf L}_{5,1}=(\{4\},\{1,2,3\})$ with virtual flows ${\bf v}(1)=1\xrightarrow{1}4$, ${\bf v}(2)=2\xrightarrow{1}4$, and ${\bf v}(3)=3\xrightarrow{1}4$ (iv) ${\bf L}_{5,11}=(\{2,3,4\},\{1\})$ with virtual flows ${\bf v}(1)=1$, ${\bf v}(2)=2\xrightarrow{1}4$, and ${\bf v}(3)=3$ (v) ${\bf L}_{5,7}=(\{2\},\{3,4\},\{1\})$ with virtual flows ${\bf v}(1)=1$, ${\bf v}(2)=2\xrightarrow{1}4$, and ${\bf v}(3)=3$.}   
	\label{figure7}
 \end{figure*}
Given ${\cal N}=\{1,2,3,4,5\}$, let ${\bf F}=\{{\bf f}(1),{\bf f}(2),{\bf f}(3)\}$ where ${\bf f}(1)=1\xrightarrow{1}4\xrightarrow{1}5$, ${\bf f}(2)=2\xrightarrow{1}4\xrightarrow{1}5$ and ${\bf f}(3)=3\xrightarrow{1}4\xrightarrow{1}5$  depicted in Figure \ref{figure7}(i).  In block $b$, node 1 encodes $(m_{1}(b),``1",``1")$, node 2 encodes $(``1",m_{2}(b),``1")$, node 3 encodes $(``1",``1",m_{3}(b))$ and node 4 encodes $(m_{1}(b-1),m_{2}(b-1),m_{3}(b-1))$.  

To encode the message vector $(m_{1},m_{2},m_{3})$ where $m_{1}\in\{1,\ldots,2^{nR_{1}}\}$, $m_{2}\in\{1,\ldots,2^{nR_{2}}\}$ and $m_{3}\in\{1,\ldots,2^{nR_{3}}\}$, node $i$ sends the codeword ${\bf x}_{i}(w)$ where $w\in\{1,\ldots,2^{n(R_{1}+R_{2}+R_{3})}\}$ is the unique index that maps to the message pair $(m_{1},m_{2},m_{3})$ and $i\in\{1,2,3,4\}$.  To simplify the exposition, we express ${\bf x}_{i}(w)$ as a function of the message vector itself, so ${\bf x}_{4}(m_{1},m_{2},m_{3}):={\bf x}_{4}(w)$.  Since $m_{2}=``1"$ and $m_{3}=``1"$ at node 1, $m_{1}=1$ and $m_{3}=``1"$ at node 2, and $m_{1}=``1"$ and $m_{2}=``1"$ at node 3, we define ${\bf x}_{1}(m_{1}):={\bf x}_{1}(w)$, ${\bf x}_{2}(m_{2}):={\bf x}_{2}(w)$, and ${\bf x}_{3}(m_{3})={\bf x}_{3}(w)$.
 
By definition, ${\cal R}_{5}({\bf F})$ is set of rate vectors ${\bf R}=(R_{1},R_{2},R_{3})$ that satisfies:
 \begin{align}
 	\label{outerR1A3}
 	R_{1}&<I(X_{1}X_{4};Y_{5}|X_{2}X_{3})\\
 	%\label{outerR2A3}
 	\nonumber
 	R_{2}&<I(X_{2}X_{4};Y_{5}|X_{1}X_{3})\\
 	%\label{outerR3A3}
 	\nonumber
 	R_{3}&<I(X_{3}X_{4};Y_{5}|X_{1}X_{2})\\
 	%\label{outerR1R2A3}
 	\nonumber
 	R_{1}+R_{2}&<I(X_{1}X_{2}X_{4};Y_{5}|X_{3})\\
 	\label{outerR1R3A3}
 	%\nonumber
 	R_{1}+R_{3}&<I(X_{1}X_{3}X_{4};Y_{5}|X_{2})\\
 	%\label{outerR2R3A3}
 	\nonumber
 	R_{2}+R_{3}&<I(X_{2}X_{3}X_{4};Y_{5}|X_{1})\\
 	%\label{outerR1R2R3A3}
 	\nonumber
 	R_{1}+R_{2}+R_{3}&<I(X_{1}X_{2}X_{3}X_{4};Y_{5})
 \end{align}

The layerings in Figure \ref{polytope}(ii) reduce to four groups, each equivalent up to permutations of nodes 1, 2, and 3.  One group is $\{{\bf L}_{5,2}, {\bf L}_{5,3}, {\bf L}_{5,4}\}$, a second group is $\{{\bf L}_{5,1}\}$, a third group is $\{{\bf L}_{5,11},{\bf L}_{5,12},{\bf L}_{5,13}\}$, and a fourth group is $\{{\bf L}_{5,5},{\bf L}_{5,6},{\bf L}_{5,7},{\bf L}_{5,8},{\bf L}_{5,9},{\bf L}_{5,10}\}$.  It suffices to consider one partition from each group: ${\bf L}_{5,3}:=(\{2,4\},\{1,3\})$, ${\bf L}_{5,1}:=(\{4\},\{1,2,3\}$, ${\bf L}_{5,11}:=(\{2,3,4\},\{1\})$, and ${\bf L}_{5,7}:=(\{2\},\{3,4\},\{1\})$, depicted in Figures \ref{figure7}(ii), (iii), (iv), and (v) respectively.   
  
\subsection{The first group of decoding schemes}

For ${\bf L}_{5,3}=(\{4,2\},\{1,3\})$ the virtual flows are  ${\bf v}(1)=1\xrightarrow{1}4$,  ${\bf v}(2)=2$, and ${\bf v}(3)=3\xrightarrow{1}4$ as depicted in Figure \ref{figure7}(ii).  In block $b$, node 5 decodes $(m_{1}(b-1),m_{2}(b),m_{3}(b-1))$ by finding the unique triple $\hat{m}_{1}\in\{1,\ldots,2^{nR_{1}}\}$,  $\hat{m}_{2}\in\{1,\ldots,2^{nR_{2}}\}$ and $\hat{m}_{3}\in\{1,\ldots,2^{nR_{1}}\}$ that jointly satisfies the following typicality checks:
 \begin{align}
		\nonumber
		&({\bf x}_{4}(\hat{m}_{1},m_{2}(b-1),\hat{m}_{3}),{\bf x}_{2}(\hat{m}_{2}),{\bf Y}_{5}(b))\\
		\label{scheme1typeA3}
		&\hspace{45mm}\in T^{(n)}_{\epsilon}(X_{2},X_{4},Y_{5})\\
		\nonumber
		&({\bf x}_{1}(\hat{m}_{1}),{\bf x}_{3}(\hat{m}_{3}),{\bf X}_{4}(b-1),{\bf X}_{2}(b-1),\\
		\nonumber
		&\hspace{20mm}\bar{Y}_{5}(b-1))\in T^{(n)}_{\epsilon}(X_{1},X_{2},X_{3},X_{4},Y_{5})
	\end{align} 
	The message $m_{2}(b-1)$ has already been decoded by node 5 in block $b$, so $\hat{m}_{1}$ and $\hat{m}_{3}$ uniquely define ${\bf x}_{4}(\hat{m}_{1},m_{2}(b-1),\hat{m}_{3})$ in (\ref{scheme1typeA3}).  The probability of error goes to zero if ${\bf R}$ satisfies the following conditions:  
\begin{align}
	%\label{scheme1R1A3}
	\nonumber
	R_{1}&<I(X_{1};Y_{5}|X_{2}X_{3}X_{4})+I(X_{4};Y_{5}|X_{2})\\
	\nonumber
	R_{2}&<I(X_{2};Y_{5}|X_{4})\\
	\nonumber
	R_{3}&<I(X_{3};Y_{5}|X_{1}X_{2}X_{4})+I(X_{4};Y_{5}|X_{2})\\
	%\label{scheme1R1R2A3}
	\nonumber
	R_{1}+R_{2}&<I(X_{1};Y_{5}|X_{2}X_{3}X_{4})+I(X_{2}X_{4};Y_{5})\\
	\nonumber
	R_{1}+R_{3}&<I(X_{1}X_{3};Y_{5}|X_{2}X_{4})+I(X_{4};Y_{5}|X_{2})\\
	\label{scheme1R1R3A3}
	&=I(X_{1}X_{3}X_{4};Y_{5}|X_{2})\\
	\nonumber
	R_{2}+R_{3}&<I(X_{3};Y_{5}|X_{1}X_{2}X_{4})+I(X_{2}X_{4};Y_{5})\\
	\nonumber
	R_{1}+R_{2}+R_{3}&<I(X_{1}X_{3};Y_{5}|X_{4}X_{2})+I(X_{2}X_{4};Y_{5})\\
	\nonumber
	&=I(X_{1}X_{2}X_{3}X_{4};Y_{5})
\end{align}
The rate region satisfying the previous inequalities is denoted by ${\cal R}({\bf F},{\bf L}_{5,3})$.  The regions ${\cal R}({\bf F},{\bf L}_{5,2})$, ${\cal R}({\bf F},{\bf L}_{5,3})$ and ${\cal R}({\bf F},{\bf L}_{5,4})$ are equivalent to the same polymatroid up to permutations of $X_{1}$, $X_{2}$, and $X_{3}$.

\subsection{The second group of decoding schemes}

For ${\bf L}_{5,5}=(\{4\},\{1,2,3\})$ the virtual flows are  ${\bf v}(1)=1\xrightarrow{1}4$,  ${\bf v}(2)=2\xrightarrow{1}4$, and ${\bf v}(3)=3\xrightarrow{1}4$ as depicted in Figure \ref{figure7}(iii).  In block $b$, node 5 decodes $(m_{1}(b-1), m_{2}(b-1), m_{3}(b-1))$ by finding the unique triple $\hat{m}_{1}\in\{1,\ldots,2^{nR_{1}}\}$,  $\hat{m}_{2}\in\{1,\ldots,2^{nR_{2}}\}$ and $\hat{m}_{3}\in\{1,\ldots,2^{nR_{3}}\}$ that jointly satisfies the following typicality checks:
 \begin{align}
		\nonumber
		&({\bf x}_{4}(\hat{m}_{1},\hat{m}_{2},\hat{m}_{3}),{\bf Y}_{5}(b))\in T^{(n)}_{\epsilon}(X_{4},Y_{5})\\
		\nonumber
		&({\bf x}_{1}(\hat{m}_{1}),{\bf x}_{2}(\hat{m}_{2}),{\bf x}_{3}(\hat{m}_{3}),{\bf X}_{4}(b-1),\\
		\nonumber
		&\hspace{20mm}{\bf Y}_{5}(b-1))\in T^{(n)}_{\epsilon}(X_{1},X_{2},X_{3},X_{4},Y_{5})
	\end{align} 
The probability of error goes to zero if ${\bf R}$ satisfies the following conditions:  
\begin{align}
	%\label{scheme2R1A3}
	\nonumber
	R_{1}&<I(X_{1};Y_{5}|X_{2}X_{3}X_{4})+I(X_{4};Y_{5})\\
	\nonumber
	R_{2}&<I(X_{2};Y_{5}|X_{1}X_{3}X_{4})+I(X_{4};Y_{5})\\
	\nonumber
	R_{3}&<I(X_{3};Y_{5}|X_{1}X_{2}X_{4})+I(X_{4};Y_{5})\\
	%\label{scheme2R1R2A3}
	\nonumber
	R_{1}+R_{2}&<I(X_{1}X_{2};Y_{5}|X_{2}X_{3}X_{4})+I(X_{4};Y_{5})\\
	\nonumber
	R_{1}+R_{3}&<I(X_{1}X_{3};Y_{5}|X_{2}X_{4})+I(X_{4};Y_{5})\\
	\nonumber
	R_{2}+R_{3}&<I(X_{2}X_{3};Y_{5}|X_{3}X_{4})+I(X_{4};Y_{5})\\
	\nonumber
	R_{1}+R_{2}+R_{3}&<I(X_{1}X_{2}X_{3};Y_{5}|X_{4})+I(X_{4};Y_{5})\\
	\nonumber
	&=I(X_{1}X_{2}X_{3}X_{4};Y_{5})
\end{align}
The rate region satisfying the previous inequalities is denoted by ${\cal R}_{5}({\bf F},{\bf L}_{5,1})$.  

\subsection{The third group of decoding schemes}

For ${\bf L}_{5,11}=(\{2,4,3\},\{1\})$ the virtual flows are  ${\bf v}(1)=1\xrightarrow{1}4$,  ${\bf v}(2)=2$, and ${\bf v}(3)=3$ as depicted in Figure \ref{figure7}(iv).  In block $b$, node 5 decodes $(m_{1}(b-1), m_{2}(b), m_{3}(b))$ by finding the unique triple $\hat{m}_{1}\in\{1,\ldots,2^{nR_{1}}\}$,  $\hat{m}_{2}\in\{1,\ldots,2^{nR_{2}}\}$ and $\hat{m}_{3}\in\{1,\ldots,2^{nR_{3}}\}$ that jointly satisfies the following typicality checks:
 \begin{align}
		\nonumber
		&({\bf x}_{2}(\hat{m}_{2}),{\bf x}_{4}(\hat{m}_{1},m_{2}(b-1),m_{3}(b-1)),{\bf x}_{3}(\hat{m}_{3})),{\bf Y}_{5}(b))\\
		\label{scheme3typeA1}
		&\hspace{35mm}\in T^{(n)}_{\epsilon}(X_{2},X_{3},X_{4},Y_{5})\\
		\nonumber
		&({\bf x}_{1}(\hat{m}_{1}),{\bf X}_{2}(b-1),{\bf X}_{4}(b-1),{\bf X}_{3}(b-1),\\
		\nonumber
		&\hspace{20mm}{\bf Y}_{5}(b-1))\in T^{(n)}_{\epsilon}(X_{1},X_{2},X_{3},X_{4},Y_{5})
	\end{align} 
The messages $m_{2}(b-1)$ and $m_{3}(b-1)$ have already been decoded by node 5 in block $b$, so $\hat{m}_{1}$ uniquely defines ${\bf x}_{4}(\hat{m}_{1},m_{2}(b-1),m_{3}(b-1))$ in (\ref{scheme3typeA1}).  The probability of error goes to zero if ${\bf R}$ satisfies the following conditions:  
\begin{align}
\nonumber
	R_{1}&<I(X_{1};Y_{5}|X_{2}X_{3}X_{4})+I(X_{4};Y_{5}|X_{2}X_{3})\\
	\label{scheme3R1A3}
	&=I(X_{1}X_{4};Y_{5}|X_{2}X_{3})\\
	%\label{scheme3R2A3}
	\nonumber
	R_{2}&<I(X_{2};Y_{5}|X_{3}X_{4})\\
	\nonumber
	R_{3}&<I(X_{3};Y_{5}|X_{2}X_{4})\\
	\nonumber
	R_{1}+R_{2}&<I(X_{1};Y_{5}|X_{2}X_{3}X_{4})+I(X_{2}X_{4};Y_{5}|X_{3})\\
	\nonumber
	&=I(X_{1}X_{2}X_{4};Y_{5}|X_{3})\\
	\nonumber
	R_{1}+R_{3}&<I(X_{1};Y_{5}|X_{2}X_{3}X_{4})+I(X_{3}X_{4};Y_{5}|X_{2})\\
	\nonumber
	&=I(X_{1}X_{3}X_{4};Y_{5}|X_{2})\\
	%\label{scheme3R2R3A3}
	\nonumber
	R_{2}+R_{3}&<I(X_{2}X_{3};Y_{5}|X_{4})\\
	\nonumber
	R_{1}+R_{2}+R_{3}&<I(X_{1};Y_{5}|X_{2}X_{3}X_{4})+I(X_{2}X_{3}X_{4};Y_{5})\\
	%\label{scheme3R1R2R3A3}
	\nonumber
	&=I(X_{1}X_{2}X_{3}X_{4};Y_{5})
\end{align}
The rate region satisfying the previous inequalities is denoted by ${\cal R}({\bf F},{\bf L}_{5,11})$.

\subsection{The fourth group of decoding schemes}
For ${\bf L}_{5,7}=(\{2\},\{3,4\},\{1\})$ the virtual flows are  ${\bf v}(1)=1\xrightarrow{1}4$,  ${\bf v}(2)=2$, and ${\bf v}(3)=3$ as depicted in Figure \ref{figure7}(iv).  In block $b$, node 5 decodes $(m_{1}(b-2), m_{2}(b), m_{3}(b-1))$ by finding the unique triple $\hat{m}_{1}\in\{1,\ldots,2^{nR_{1}}\}$,  $\hat{m}_{2}\in\{1,\ldots,2^{nR_{2}}\}$ and $\hat{m}_{3}\in\{1,\ldots,2^{nR_{3}}\}$ that jointly satisfies the following typicality checks:
 \begin{align}
		\nonumber
		&({\bf x}_{2}(\hat{m}_{2}),{\bf Y}_{5}(b))\in T^{(n)}_{\epsilon}(X_{2},Y_{5})\\
		\nonumber
		&({\bf x}_{4}(\hat{m}_{1},m_{2}(b-1),m_{3}(b-1)),{\bf x}_{3}(\hat{m}_{3}),{\bf X}_{2}(b-1),\\
		\label{scheme4typeA2}
		&\hspace{25mm}{\bf Y}_{5}(b-1))\in T^{(n)}_{\epsilon}(X_{2},X_{3},X_{4},Y_{5})\\
		\nonumber
		&({\bf x}_{1}(m_{1}),{\bf X}_{2}(b-2),{\bf X}_{3}(b-2),{\bf X}_{4}(b-2),{\bf Y}_{5}(b-2))\\
		\nonumber
		&\hspace{30mm}\in T^{(n)}_{\epsilon}(X_{1},X_{2},X_{3},X_{4},Y_{5})
\end{align}
The messages $m_{2}(b-1)$ and $m_{3}(b-1)$ have already been decoded by node 5 in block $b$, so $\hat{m}_{1}$ uniquely determines ${\bf x}_{4}(\hat{m}_{1},m_{2}(b-1),m_{3}(b-1))$ in (\ref{scheme4typeA2}).  The probability of error goes to zero if ${\bf R}$ satisfies the following conditions:  
\begin{align}
\nonumber
	R_{1}&<I(X_{1};Y_{5}|X_{2}X_{3}X_{4})+I(X_{4};Y_{5}|X_{2}X_{3})\\
	\label{scheme4R1A3}
	&=I(X_{1}X_{4};Y_{5}|X_{2}X_{3})\\
	\nonumber
	R_{2}&<I(X_{2};Y_{5})\\
	\nonumber
	R_{3}&<I(X_{3};Y_{5}|X_{2}X_{4})\\
	\nonumber
	R_{1}+R_{2}&<I(X_{1};Y_{5}|X_{2}X_{3}X_{4})+I(X_{4};Y_{5}|X_{2}X_{3})\\
	\nonumber
	&\hspace{5mm}+I(X_{2};Y_{5})\\
	%\label{scheme4R1R2A3}
	\nonumber
	&=I(X_{1}X_{4};Y_{5}|X_{2}X_{3})+I(X_{2};Y_{5})\\
	\nonumber
	R_{1}+R_{3}&<I(X_{1};Y_{5}|X_{2}X_{3}X_{4})+I(X_{3}X_{4};Y_{5}|X_{2})\\
	\label{scheme4R1R3A3}
	&=I(X_{1}X_{3}X_{4};Y_{5}|X_{2})\\
	\nonumber
	R_{2}+R_{3}&<I(X_{3};Y_{5}|X_{2}X_{4})+I(X_{2};Y_{5})\\
	\nonumber
	R_{1}+R_{2}+R_{3}&<I(X_{1};Y_{5}|X_{2}X_{3}X_{4})+I(X_{3}X_{4};Y_{5}|X_{2})\\
	\nonumber
	&\hspace{5mm}+I(X_{2};Y_{5})\\
	\nonumber
	&=I(X_{1}X_{2}X_{3}X_{4};Y_{5})
\end{align}
The rate region satisfying the previous inequalities is denoted by ${\cal R}({\bf F},{\bf L}_{5,7})$.

\subsection{The Achievability of ${\cal R}_{5}({\bf F})$}

It is difficult to show ${\cal R}_{5}({\bf F})=\cup^{13}_{k=1}{\cal R}({\bf F},{\bf L}_{5,k})$ using the same geometric arguments as Appendix B:(Section D) for the two-source multiple-relay channel.  An alternative proof appears in \cite{Sankar} that first finds an offset encoding scheme for each corner point in ${\cal R}_{5}({\bf F})$ and then invokes convexity and timesharing.  

The layerings $\{{\bf L}_{5,k}:1\leq k\leq 13\}$ are related through $\text{\sc shift}$ operation.  From Figures \ref{figure7}(ii) and (iii), ${\bf L}_{5,1}=\text{\sc shift}({\bf L}_{5,3},\{2\})$.  Similarly, from Figures \ref{figure7}(iii) and (iv), ${\bf L}_{5,1}=\text{\sc shift}({\bf L}_{5,11},\{2,3\})$.  From Figures \ref{figure7}(iv) and (v), ${\bf L}_{5,7}=\text{\sc shift}({\bf L}_{5,11},\{1,3\})$.  The geometric relationships between these rate regions are shown in Figure \ref{polytope}(ii).